%% file: graphs-gal151214.tex
\documentclass{amsproc}

\newtheorem{theorem}{Theorem}[section]
\newtheorem{thm}{Theorem}

\newtheorem{lemma}[theorem]{Lemma}
\newtheorem{proposition}{Proposition}[section]

\theoremstyle{definition}

\newtheorem{example}[theorem]{Example}

\theoremstyle{remark}
\newtheorem{remark}[theorem]{Remark}

\numberwithin{equation}{section}

\newcommand{\abs}[1]{\lvert#1\rvert}

\begin{document}

\title[Quantum ergodicity on regular graphs]{Quantum ergodicity on regular graphs}

\author{Nalini Anantharaman}
\address{IRMA, Universit\'{e} de Strasbourg, 7 rue Ren\'e-Descartes, 67084 STRASBOURG
CEDEX, FRANCE}
\email{anantharaman@math.unistra.fr}
\thanks{}
\ 

\input{def.tex}

\date{}


\begin{abstract} We give three different proofs of the main result of \cite{ALM}, establishing quantum ergodicity -- a form of delocalization --
for eigenfunctions of the laplacian on large regular graphs of fixed degree. These three proofs are much shorter than the original one, quite different from one another, and we feel
that each of the four proofs sheds a different light on the problem. The goal of this exploration is to find a proof that could be adapted for other
models of interest in mathematical physics, such as the Anderson model on large regular graphs, regular graphs with weighted edges, or possibly certain models of non-regular graphs.
A source of optimism in this direction is that we are able to extend the last proof to the case of {\em{anisotropic}} random walks on large regular graphs.

\end{abstract}

\maketitle
\section{Introduction} 
\subsection{} ``Quantum ergodicity'' usually refers to the study of the delocalization of eigenfunctions of the laplacian, on a compact Riemannian manifold, in the high-energy limit. What is usually called the Quantum Ergodicity Theorem (or Shnirelman theorem) is the following:

\bigskip
{\bf Quantum Ergodicity Theorem (Shnirelman theorem, \cite{Sni, CdV85, Zel87}).} Let $(M, g)$ be a compact Riemannian manifold, with the metric normalized so that $\Vol(M) = 1$. Call $\Delta$ the Laplace-Beltrami operator on $M $. Assume that the geodesic flow of $M $ is {\em ergodic} with respect to the Liouville measure. Let $(\psi_n)_{n\in \IN}$ be an orthonormal basis of $L^2(M, g)$ made of eigenfunctions of the laplacian~:
$$\Delta\psi_n=-\lambda_n \psi_n,\qquad \lambda_n\leq \lambda_{n+1}\To +\infty.$$
Let $a$ be a continuous function on $M $ such that $\int_M a(x)d\Vol(x)=0$. Then
\begin{equation}\label{e:variance}\frac{1}{N(\lambda)}\sum_{n, \lambda_n\leq \lambda}|\la \psi_n, a\psi_n\ra_{L^2(M)}|^2 \Lim_{\lambda\To +\infty}0\end{equation}
where the normalizing factor is the eigenvalue counting function $$N(\lambda)=|\{n, \lambda_n\leq \lambda\}| .$$

Note that $\la \psi_n, a\psi_n\ra_{L^2(M)}=\int_M a(x)|\psi_n(x)|^2d\Vol(x).$
Equation \eqref{e:variance} implies that there exists a subset $S\subset \IN$ of density $1$ such that
the sequence of measures $(|\psi_n(x)|^2d\Vol(x))_{n\in S}$ converges weakly to the uniform measure $d\Vol(x)$.

Actually, the full statement of the theorem says that
there exists a subset $S\subset \IN$ of density $1$ such that
\begin{equation}\label{e:subsequenceOp}\la \psi_n, A\psi_n\ra \Lim_{n\To +\infty, n\in S} \int_{SM} \sigma^0(A) dL\end{equation} 
for every pseudodifferential operator $A$ of order $0$ on $M $. On the right-hand side, $\sigma^0(A)$ is the principal symbol of $A$, that is a function on the sphere bundle $SM$, and $L$ is the normalized Liouville measure, that is, the ``uniform measure'' on $SM$.
 
\bigskip

Graphs are convenient toy-models to study, rigorously or numerically, the localization/delocalization properties of eigenfunctions of Schr\"odinger operators. However, for the so-called ``quantum graphs'', that is, $1$-dimensional CW-complexes with $\Delta=\frac{d^2}{dx^2}$ on the edges and suitable ``Kirchhoff'' matching conditions on the vertices, ``quantum ergodicity'' never holds (at least for rationally independent lengths of the edges \cite{CdV13}). See also \cite{BKW04, KMW03, Gnu10, BKS07} for other results pertaining to eigenvalue or eigenfunction statistics on (families of) compact quantum graphs.

In \cite{ALM}, a result of ``quantum ergodicity'' was obtained for the first time in a completely discrete setting, namely for the discrete laplacian on a family of discrete graphs. The results are recalled in Section \ref{s:results}; the method employed in \cite{ALM} relies heavily on the Fourier transformation on infinite regular trees, giving the impression that the proof is very specific to {\em{regular}} graphs (for which each vertex has the same number of neighbours).
However, the statement of the theorem
makes sense for arbitrary sequences of graphs, so that it is natural to try to find a proof that does not use the Fourier transform and thus could be adapted to other models
of interest in mathematical physics (non-regular graphs, regular graphs with weighted adjacency matrices, Anderson model...).
The purpose of this note is to provide three different proofs of the main result of \cite{ALM}, recalled below as Theorems \ref{t:main} and \ref{t:gen}. These proofs are considerably shorter than the one already published, and in principle they look more susceptible of an adaptation to the non-regular models mentioned above. This optimistic view is supported by the results of the last section, where we are able to extend the last proof to the case of {\em{anisotropic}} (but homogeneous) random walks on large regular graphs (see Theorem \ref{t:anis}). This proof uses some spectral analysis
on trees, but actually enlightens the assumptions under which one can hope to prove quantum ergodicity with our method~: it has a chance to work if one studies eigenfunctions of a {\em{nearest neighbour}} self-adjoint operator
on a graph {\em{looking locally like a tree}}, in a region of the spectrum where the spectral measure for the operator on the limiting infinite tree is {\em{absolutely continuous}}. Adaptation to other models satisfying these properties seems at hand, but would require substantial additional work, and is deferred to separate publications. 

It is important to note that we deal with {\em{deterministic}} sequences of (regular) graphs, satisfying some geometric conditions
(expansion and convergence to a tree in the Benjamini-Schramm topology). Our results apply in particular to random regular graphs,
but also to to explicit deterministic families of graphs, for instance the Ramanujan graphs constructed in \cite{LPS88}, or the Cayley graphs of
${\rm SL}_2(\Bbb F_p)$ considered in \cite{BG}.

\subsection{\label{s:results}}Consider a sequence of $(q+1)$-regular connected graphs $(G_N)_{N\in\IN}$.
More precisely, we assume that $G_N$ is a quotient of the $(q+1)$-regular tree ${\mathfrak{X}}$ by a group of automorphisms $\Gamma_N$, that is, $G_N= \Gamma_N\backslash {\mathfrak{X}}$, where the elements of $\Gamma_N$ act without fixed points\footnote{This assumption does not seem crucial, but is used in one calculation in the paper.} on the vertices of ${\mathfrak{X}}$.
We write $V({\mathfrak{X}})$ and $E({\mathfrak{X}})$ respectively for the vertices and the (non-oriented) edges of ${\mathfrak{X}}$;
and we denote by $V_N =\Gamma_N\backslash V({\mathfrak{X}})$ the vertices of $G_N$, $E_N= \Gamma_N\backslash E({\mathfrak{X}})$
the (non-oriented) edges of $G_N$. We will work in the limit $|V_N|\To +\infty$, and for notational simplicity we will assume that $|V_N|=N$.

When needed we will denote by $B({\mathfrak{X}})$ the set of oriented edges of $ {\mathfrak{X}}$ (each element of $E({\mathfrak{X}})$ corresponds to two elements of $B({\mathfrak{X}})$). We denote by $B_N=\Gamma_N\backslash B( {\mathfrak{X}})$ the set of oriented edges in $G_N$. There is a map $B({\mathfrak{X}})\To V({\mathfrak{X}})\times V({\mathfrak{X}})$, $e\mapsto (o(e), t(e))$ (origin and terminus of $e$), and an involution $e\mapsto \hat e$ of  $B({\mathfrak{X}})$ (reversal of edges) such that $o(\hat e)=t(e)$ for all $e$. These maps are also well-defined on the quotient $G_N$.

We will use the notation $d_{\mathfrak{X}}(x, y)$ for the distance between two vertices $x, y$ of the tree ${\mathfrak{X}}$, and $d_{G_N}(x, y)$ for the distance between two vertices $x, y$ of the graph $G_N$; most of the time, when it does not cause any confusion, we will simply use the same notation $d(x, y)$ for both. The notation $B(x, r)$ will stand for the ball of center $x$ and radius $r$, either for $x$ in the tree ${\mathfrak{X}}$ or in the graph $G_N$. We will denote by $\tau(r)=(q+1)q^{r-1}$ the cardinality of the sphere of center $x$ and radius $r$ in the tree ${\mathfrak{X}}$, and $\tilde \tau (r)$ the cardinality of the ball of center $x$ and radius $r$, that is, $\tilde\tau(r)=1+(q+1)\sum_{1\leq k\leq r}q^{k-1}$.
  
  Consider the adjacency operator defined on functions on $V({\mathfrak{X}})$
by
\begin{equation}\label{e:laplacian}
 \cA f(x)= \sum_{{x\sim y} } f(y).
\end{equation}
It is related to the discrete laplacian by 
\begin{equation}\cA -(q+1)I=\Delta.\label{e:qI}
\end{equation}
The operator $\cA$ preserves the space of $\Gamma_N$-invariant functions (that may be identified to functions on $V_N$).
Since we assumed that $\Gamma_N$ acts without fixed vertices, $\cA$ is self-adjoint on $\ell^2(V_N)$ with the uniform measure on the finite set $V_N$.

We will assume the following conditions on our sequence of graphs~:

\bigskip

(EXP) The sequence of graphs is a family of expanders. More precisely, there exists $\beta>0$ such that the spectrum of $(q+1)^{-1}\cA$ on $\ell^2(V_N)$ is contained in $\{1\}\cup [-1+\beta, 1-\beta]$ for all $N$.

\bigskip

(BST) For all $R$, $$\frac{|\{x\in V_N, \rho(x)<R\}|}{N}\Lim_{N\To \infty} 0$$
where $\rho(x)$ is the injectivity radius at $x$ (meaning the largest $\rho$ such that the ball $B(x, \rho)$ in $G_N$ is a tree).

\bigskip

(BST) means that our sequence of graphs converges to a tree in the sense of Benjamini-Schramm \cite{BS}.
 In particular, this condition is satisfied if the girth goes to infinity. Under the condition (BST), one has convergence of the spectral measure according to the Kesten-McKay law \cite{Kes59, McK81}.
 Call $(\lambda^{(N)}_1,\ldots, \lambda^{(N)}_N)$ the eigenvalues of $\cA$ on $G_N$;
 for any interval $I\subset \R$,
 $$\frac1N|\{ j, \lambda^{(N)}_j\in I\}| \Lim_{N\To +\infty} \int_I {\mathrm{m}}(\lambda)d\lambda$$
 where ${\mathrm{m}}(\lambda)$ is a completely explicit probability density, namely the spectral measure of a Dirac mass $\delta_o$ 
 for the operator $\cA$, defined by its moments
 \begin{equation}\label{e:planch}\int \lambda^k {\mathrm{m}}(\lambda) d\lambda=\la \delta_o, \cA^k\delta_o\ra_{\ell^2(\mathfrak X)}\end{equation}
 (this is the number of paths in the tree starting at $o$ and returning to $o$ after $k$ steps).
 We won't need the expression of ${\mathrm{m}}$ here, but let us recall that it is smooth and positive on $(-2\sqrt{q}, 2\sqrt{q})$ and vanishes elsewhere.
 In particular, most of the $\lambda^{(N)}_j$ are in $(-2\sqrt{q}, 2\sqrt{q})$, an interval strictly smaller than $[-(q+1), q+1]$ (this part of the spectrum is called the {\em{tempered spectrum}}).

\bigskip

The main result of \cite{ALM} is recalled below as Theorem \ref{t:main}, or in its stronger form as Theorem \ref{t:gen}. 

\begin{thm}\cite{ALM} \label{t:main} Let $(G_N)=(V_N, E_N)$ be a sequence of $(q+1)$-regular graphs with $|V_N| =N$. Assume that $(G_N)$ satisfies (BST) and (EXP).

Let $(\psi^{(N)}_1,\ldots, \psi^{(N)}_N)$ be an orthonormal eigenbasis of eigenfunctions of $\cA$ in $\ell^2(V_N)$.


Let $a_N : V_N \To \C$ be a sequence of functions such that  $\sup_N \sup_{x\in V_N} |a_N(x)|\leq 1.$ Define $\la {a_N}\ra=\frac1N
\sum_{x\in V_N} a_N(x)$

Then 
\begin{equation}\label{e:qvar}\frac1{N } \sum_{j=1}^N \left| \la \psi^{(N)}_j, a_N\psi^{(N)}_j\ra_{\ell^2(V_N)} - \la {a_N}\ra\right|^2\Lim_{N\To +\infty} 0.\end{equation}
\end{thm}

Note that   $\la \psi^{(N)}_j, a_N\psi^{(N)}_j\ra_{\ell^2(V_N)} $ is the scalar product between $\psi^{(N)}_j$ and $a_N\psi^{(N)}_j$, its explicit expression is
$\sum_{x\in V_N}a_N(x)|\psi^{(N)}_j(x)|^2$. The interpretation of Theorem \ref{t:main} is that we are trying to measure the distance between the two probability measures on $V_N$,
$$\sum_{x\in V_N}|\psi^{(N)}_j(x)|^2 \delta_x \quad \mbox{ and }\quad \frac1N \sum_{x\in V_N} \delta_x \qquad \mbox{ (uniform measure)}$$
in a rather weak sense (just by testing the function $a_N$ against both). What \eqref{e:qvar} tells us is that for large $N$ and for most $j$, this distance is small.

We can generalize Theorem \ref{t:main} by replacing a function $a_N$ with any finite range operator~:
\begin{thm}\cite{ALM} \label{t:gen}Let $(G_N)=(V_N, E_N)$ be a sequence of $(q+1)$-regular graphs with $|V_N| =N$. Assume that $(G_N)$ satisfies (BST) and (EXP).

Call $(\lambda^{(N)}_1,\ldots, \lambda^{(N)}_N)$ the eigenvalues of $\cA$ on $G_N$, and let $(\psi^{(N)}_1,\ldots, \psi^{(N)}_N)$ be a corresponding orthonormal eigenbasis.


Fix $D\in \N$. Let ${\mathbf K}={\mathbf K}_N$ be a sequence of operators on $\ell^2(V_N)$ whose kernels\footnote{We caution that we do not use the word kernel in the sense of ``null space''. For us the kernel means
$K(x, y)=\la \delta_x, {\mathbf K}\delta_y\ra_{\ell^2(V_N)}$, that is to say the matrix of ${\mathbf K}$ in the basis of Dirac masses $(\delta_x)$.} $K : V_N \times V_N\To \C$ are such that $K(x, y)=0$ for $d(x, y)>D$ (we will say that $K$ is supported at distance $\leq D$ from the diagonal). 

Assume that
$   \sup_N \sup_{x, y\in V_N}|K(x, y)| \leq 1.$

Then there exists a number $ \la {\mathbf K}\ra_{\lambda^{(N)}_j}$ such that
$$\frac1{N} \sum_{j} \left| \la \psi^{(N)}_j, {\mathbf K}\psi^{(N)}_j\ra_{\ell^2(V_N)} - \la {\mathbf K}\ra_{\lambda^{(N)}_j}
\right|^2\Lim_{N\To +\infty} 0.$$

The expression of $ \la {\mathbf K}\ra_{\lambda^{(N)}_j}$ is explicit (given below) and depends only on $K$ and $\lambda^{(N)}_j$.

 \end{thm}

\begin{remark}  \label{r:sph}The explicit formula for $ \la {\mathbf K}\ra_{\lambda}$ we obtain here is much simpler than the recipe given in the original paper \cite{ALM}.
Eigenvalues of the adjacency matrix $\cA$ of a $(q+1)$-regular graph may be parametrized by a ``spectral parameter''  $s$ thanks to the relation
\begin{equation}\label{e:parametrization}
 \lambda=\lambda(s)=  q^{1/2 +is}+q^{1/2-is}=2\sqrt{q}\cos(s\ln q)
\end{equation}
where $s\in \IR\cup i\IR$ since $\lambda\in\R$ ($s\in \R$ is equivalent to $|\lambda|\leq 2\sqrt{q}$, corresponding to the ``tempered spectrum'').  

Let us then denote $\Phi_\lambda(d)$ the spherical function of parameter $\lambda$:
$$\Phi_\lambda(d)= q^{-d /2}\left(\frac{2}{q+1}\cos(d s\ln q) +\frac{q-1}{q+1}\frac{\sin((d+1)s\ln q)}{\sin(s\ln q)}\right)$$
(for $d\in \IN$).
Then
$$ \la {\mathbf K}\ra_{\lambda}=\frac1N\sum_{x, y\in V_N}  K( x, y)\Phi_\lambda( d(x, y))$$ 
 
Note that the scalar product $ \la \psi^{(N)}_j, {\mathbf K}\psi^{(N)}_j\ra_{\ell^2(V_N)}$ is
$$\sum_{x, y\in V_N} K(x, y)\overline{\psi^{(N)}_j(x)}\psi^{(N)}_j(y)$$
so our result says that, for large $N$ and for most $j$, the function $\psi^{(N)}_j(x)\overline{\psi^{(N)}_j}(y)$ is ``close'' to the universal quantity $\frac{1}N\Phi_{\lambda^{(N)}_j}( d(x, y))$ (in a weak distance defined by
testing $K(x, y)$ against both).
 
 Keeping in mind our idea to find a statement which will hold for more general models than the regular graphs, we note that for $\lambda\in (-2\sqrt{q}, 2\sqrt{q})$ the spherical function is related to the Green function by the following identity~:  
 $$ \frac{\Im m \,  G_{\lambda+ i 0}(x, y)}{  \Im m \,  G_{\lambda+ i 0}(x, x)}=\Phi_\lambda(d(x, y))$$
 for any two points $x, y\in \mathfrak{X}$ at distance $d(x, y)$ in the tree. Here the Green function $G_z(x, y)=(z-\cA)^{-1}(x, y)$ is defined for $z\not\in \IR$, and $G_{\lambda+ i 0}(x, y)$
 is the limit of $G_{\lambda +i\eps}(x, y)$ as $\eps$ is a positive real number going to $0$.
\end{remark}

\begin{remark}\label{r:rate} What we actually prove (Proofs 2 and 3), for a fixed regular graph $G=\Gamma\backslash \mathfrak{X}=(V, E)$, is an upper bound of the form
\begin{equation}\frac1{|V|} \sum_{j=1}^{|V|} \left| \la \psi_j, {\mathbf K}\psi_j\ra - \la {\mathbf K}\ra_{\lambda_j}
\right|^2 \leq  \tilde C(D, \beta)\left(\frac{\norm{{\mathbf K}}^2_{HSN}}{R^{\alpha}} +q^{\tilde\alpha R} \frac{1}{|V|}\sharp\{x\in V, \rho(x)\leq R\}\sup_{x, y}|K(x, y)|^2\right),\label{e:rate}
\end{equation}
valid when $K$ is supported at distance $\leq D$ from the diagonal and where $R$ is arbitrary, $\alpha, \tilde \alpha$ are some fixed positive numbers (the different methods give different values of $\alpha, \tilde \alpha$). 
The norm $\norm{{\mathbf K}}_{HSN}$ is the normalized Hilbert-Schmidt norm of the operator (see \eqref{e:defHSN}), and $\sup_{x, y}|K(x, y)|$ is its $\ell^1\To \ell^\infty$ norm.
The constant $\tilde C(D, \beta)$ is an explicit function of $D$ and of the spectral gap $\beta$ (coming from the expanding condition).

For the sake of brevity, Proof 4 is written in a less explicit way, and gives a result in the form
\begin{equation}\label{e:weaker}Var(K)\leq \frac{\tilde C(D, \beta)}R \norm{{\mathbf K}}^2_{HSN} +\sup_{x, y}|K(x, y)|^2 \, \, \left(o_{ R, D}(1)_{G\To\mathfrak X } + O_D(R^{-1})\right)
\end{equation}
where $o_{ R, D}(1)_{G\To\mathfrak X }$ is a quantity going to $0$ under the condition (BST), if $D$ and $R$ are fixed.

Under assumptions (EXP) and (BST), the bound \eqref{e:rate} implies Theorem \ref{t:gen}.  In particular, if the girth goes to infinity, we can take $R$ to be the half-girth and we obtain in this case
\begin{equation}\label{e:girth}\frac1{|V|} \sum_{j=1}^{|V|} \left| \la \psi_j, {\mathbf K}\psi_j\ra - \la {\mathbf K}\ra_{\lambda_j}
\right|^2 \leq  \tilde C(D, \beta)\frac{\norm{{\mathbf K}}^2_{HSN}}{girth(G)^{\alpha}}   \end{equation}
since the other term in \eqref{e:rate} vanishes.

\end{remark}

\subsection{Anisotropic, homogeneous random walks}
Counting the original proof of \cite{ALM}, we now have four different proofs of Theorems \ref{t:main}, \ref{t:gen}.
In future work, we would like to be able to deal with regular graphs with random weights on the edges or random potentials (Anderson model), or with certain non-regular graphs. We tried to adapt our four proofs to these models but failed for the first three; however, there is hope to adapt Proof 4,
provided a certain number of assumptions are respected. To illustrate this, we treat in Section \ref{s:anis} a model that is still simple (because homogeneous), but where the general ideas start to emerge~: the anisotropic random walks on regular graphs.

Again we consider a family $G_N=(V_N, E_N)$ of $(q+1)$-regular graphs. We assume that there is a labelling of the edges
$${\mathrm c}: B_N \To \{1,\ldots, q+1\}$$
such that ${\mathrm c}(e)={\mathrm c}(\hat e)$ for all $e$ (the label does not depend on the orientation of the edge) and
such that for every $x\in V_N$, for every $j\in \{1, \cdots, q+1\}$,
$$\sharp\{ e\in B_N, x=o(e), {\mathrm c}(e)=j\}= 1.$$
In other words, among the $q+1$ edges emanating from any given vertex, exactly one has each label.
Assume we are given probability weights $p_1, \ldots, p_{q+1}>0 $ (independent of $N$) such that $\sum_{j=1}^{q+1} p_j=1$.
We consider the random walk on $V_N$ defined by the condition that the probability of jumping from $x$ to $y$ is $p_j$,
whenever $x$ and $y$ are joined by an edge labelled $j$. In other words, we consider the symmetric stochastic operator
\begin{equation}\label{e:Ap}\cA_p f(x)=\sum_{y\sim x} p({\mathrm c}(x, y)) f(y).\end{equation}
The random walk is {\em{ anisotropic}} because the transition probability is not the same in all directions, but it is still {\em{homogeneous}}; we mean by this that the local situation is the same around each point.

\begin{thm}\label{t:mainanis} The statement of Theorem \ref{t:main} remains true, without changing a single word, in the case of anisotropic homogeneous random walks.
\end{thm}
At this stage we do not see any difference between the isotropic and the anisotropic random walks~: the probability measures $|\psi^{(N)}_j|^2$ approach the uniform measure on $V_N$. However the calculations involved in the proof will be quite different, and
this difference is more visible in the generalization of Theorem \ref{t:gen}, that is, when we use general operators as ``test functions''.

\begin{thm}  \label{t:anis}  The statement of Theorem \ref{t:gen} remains true, under the same assumptions, in the case of anisotropic homogeneous random walks -- if we replace the quantity $\la {\mathbf K}\ra_\lambda$ by $\la {\mathbf K}\ra_{\lambda, p}$ defined below, which depends on the probability transitions $(p_j)$.









 \end{thm}
 
  To define $ \la {\mathbf K}\ra_{\lambda, p}$, we write $(G_N, {\mathrm c})$ (the labelled graph) as a quotient $\Gamma_N\backslash (\mathfrak X, {\mathrm c})$ of the labelled tree $(\mathfrak X, {\mathrm c})$ by a group of automorphisms $\Gamma_N$ preserving the labelling. The operator $\cA_p$ can be defined on the tree as in \eqref{e:Ap}.

Consider for $\gamma\not\in \IR$, and $x, y\in \mathfrak X$, the Green function
$$G_\gamma(x, y)=(z-\cA_p)^{-1}(x, y)$$
of the operator $\cA_p$ on the labelled tree. Then, if $K(x, y)$ is supported at distance $\leq D$ from the diagonal,
\begin{equation}\label{e:klambdap} \la {\mathbf K}\ra_{\lambda, p}=\frac{1}{N}\sum_{x, y\in V_N}K(x, y) \frac{\Im m \,  G_{\lambda+ i 0}(x, y)}{  \Im m \,  G_{\lambda+ i 0}(x, x)}.
\end{equation}
To define $\Im m \,  G_{\lambda+ i 0}(x, y)$ for $x, y\in V_N$, take its value at any lifts $\tilde x, \tilde y$ of $x, y$ in $\mathfrak X$ such that
$d_{\mathfrak X}(\tilde x, \tilde y)\leq D$. This gives an unambiguous definition except for the points $x$ where the injectivity radius $\rho(x)$ is $<D$,
but because of (BST), this ambiguity in the definition is only of order $o(N^{-1})$. When $D=0$ we obtain
$ \la {\mathbf K}\ra_{\lambda, p}=\frac{1}{N}\sum_{x\in V_N}K(x, x)$ and Theorem \ref{t:anis}
implies Theorem \ref{t:mainanis}.

Another way to understand this formula is the following. Let $P_{(+\infty, \lambda]}$ be the spectral projector on $(+\infty, \lambda]$ for $\cA_p$
on $\ell^2(\mathfrak X)$. Seeing $P$ as an operator-valued measure, write
$$P_{(+\infty, \lambda]}=\int_{-\infty}^{\lambda} P_t \, {\mathrm m}(t) dt,$$
that is, $P_t$ is the density of $P$ with respect to the measure ${\mathrm m}(t) dt$, defined as in \eqref{e:planch} with $\cA$ replaces by $\cA_p$.
Since
$$P_t(x, y)=\frac{\Im m \,  G_{t+ i 0}(x, y)}{  \Im m \,  G_{t+ i 0}(x, x)},$$
we see that
$$ \la {\mathbf K}\ra_{\lambda, p}=\frac{1}{N}\sum_{x, y\in V_N}K(x, y)P_{\lambda}(x, y).$$

\bigskip

{\bf{Outline of the paper~:}} Sections \ref{s:notation} and \ref{s:prep} gather notation and facts that are needed in the three proofs. The proofs themselves can in principle be read independently of each other, and are given in Sections \ref{s:short} (shortest proof), \ref{s:ergflav} (ergodic-theoretic flavoured proof) and \ref{s:nbt} (proof using the non-backtracking random walk). The adaptation of the latter to the anisotropic case is given in Section \ref{s:anis}; notation and ideas from Section \ref{s:nbt} are re-used in Section \ref{s:anis}.
 
\bigskip

{\bf{Acknowledgements~:}} The author is supported by the Labex IRMIA and USIAS of Universit\'e de Strasbourg.
This material is based upon work supported by the Agence Nationale de la Recherche under grant No. ANR-13-BS01-0007-01,
and by the National Science Foundation under grant No. DMS-1440140 while the author was in residence at the MSRI in Berkeley, Ca., during the spring 2015 semester.

I am very thankful to Mostafa Sabri for his careful reading and numerous useful comments on the manuscript.

\section{Notation and basic identities \label{s:notation}}

\subsection{Operators on the tree and on quotients of the tree\label{s:ops}}
Let $\mathfrak{X}$ be the $(q+1)$-regular tree, let $\Gamma$ be a subgroup of the automorphism group of $\mathfrak{X}$ such that the quotient graph $G=\Gamma\backslash \mathfrak{X}=(V, E)$ is finite. We will always assume that $G$ is connected. The adjacency operator on the tree is
$\cA f(x)=\sum_{y\sim x} f(y)$. It preserves the $\Gamma$-invariant functions, which may be identified with functions on $V$. We assume for simplicity that $\Gamma$ acts without fixed vertices; then  $\cA$ is self-adjoint on $\ell^2(V)$ (with $V$ endowed with the uniform measure). 
In this section, $\Gamma$ is fixed, so we drop the index $N$ (in Theorems \ref{t:main} and \ref{t:gen}, $\Gamma=\Gamma_N$ is a sequence assumed to satisfy (EXP) and (BST)).

Call $\cD\subset V(\mathfrak{X})$ a fundamental domain for the action of $\Gamma$ on $V(\mathfrak{X})$. Its cardinality is $|\cD|=|V|$, the number of vertices of $G$.

We introduce the space $\cH$ (depending on $\Gamma$ although this does not appear in the notation) of 
functions $K : V(\mathfrak{X})\times V(\mathfrak{X})\To \IC$, satisfying the condition that $K(\gamma\cdot x, \gamma \cdot y)=K(x, y)$ for all $\gamma \in \Gamma$, and such that
\begin{equation}\frac{1}{|\cD|}\sum_{x\in \cD, y\in\mathfrak{X}} |K(x, y)|^2<+\infty.\label{e:hilbert}\end{equation}
Defining the scalar product $\la K_1, K_2\ra_{\cH}= \frac{1}{|\cD|}\sum_{x\in \cD, y\in\mathfrak{X}} \overline{K_1(x, y)}K_2(x, y) $ turns $\cH$ into a Hilbert space, and we denote by $\norm{K}_{\cH}$ the corresponding norm.
The space $\cH$ may be decomposed into the Hilbertian sum $\bigoplus_{k=0}^{+\infty}\cH_k$,
where $\cH_k\subset \cH$ is the finite-dimensional subspace of functions such that
$$ K(x, y)\not=0\Rightarrow d_{\mathfrak X}(x, y)=k.$$
When using this decomposition we will write $K=(K_k)_{k=0}^{+\infty}$. We will denote $\cH_{\leq k}= \bigoplus_{j=0}^{k}\cH_j$, and if $K\in \cH_{\leq k}$ we say it is {\em{supported at distance $\leq k$ from the diagonal}}.

An element $K \in \cH$ naturally defines an operator $\widehat{K}$ acting on $C_c(V(\mathfrak{X}))$~: the operator with kernel
$K(x, y)=\la \delta_x, \widehat{K}\delta_y\ra_{\ell^2(\mathfrak X)}$.
Note that $\norm{K}_{\cH}$ is the Hilbert-Schmidt norm of the operator $|\cD|^{-1/2} \bbbone_{\cD}\widehat{K}$ on $\ell^2(\mathfrak{X})$.

We will also use the norm $\norm{K}_{\sup}\defeq \sup_{x, y}|K(x, y)|$. Of course $|V|^{-1/2}\norm{K}_{\sup}\leq \norm{K}_\cH\leq (q+1)q^{k-1}\norm{K}_{\sup}$. Remark that there exists $c(k)>1$, independent of the group $\Gamma$, such that
$c(k)^{-1}\norm{\widehat{K}}_{\ell^2(\mathfrak{X})\To \ell^2(\mathfrak{X})}\leq \norm{K}_{\sup}\leq  c(k)\norm{\widehat{K}}_{\ell^2(\mathfrak{X})\To \ell^2(\mathfrak{X})} $ for $K\in \cH_k$
(throughout the paper we denote by $\norm{\cdot}_{F\To G}$ the norm of bounded operators from a normed space $F$ to a normed space $G$).
 
Denote by $B_0(\mathfrak{X})= V(\mathfrak{X})$ the set of vertices of $\mathfrak{X}$; $B_1(\mathfrak{X})=B(\mathfrak{X})$ the set of oriented edges, and more generally $B_k(\mathfrak{X})$ the set of non-backtracking paths of length $k$ in $\mathfrak{X}$. The set $B_k(\mathfrak{X})$ is in natural bijection
with the pairs $(x, y)$ of vertices of $\mathfrak{X}$ at distance $k$ from each other; we will denote by $(x;y)$ the element of $B_k(\mathfrak{X})$ determined by two such $x$ and $y$.
Thus, a non-backtracking path $\omega=(x_0,\cdots, x_k)\in B_k(\mathfrak{X})$ can also be denoted by $(x_0; x_k)$. For
$\omega'=(x'_0,\cdots, x'_k)\in B_k$, we write $\omega \rightsquigarrow \omega'$ if $x'_0=x_1,\ldots, x'_j=x_{j+1}$ and $(x_0,\cdots, x_k, x'_{k})$ belongs to $B_{k+1}(\mathfrak{X})$.

In the quotient, denote by $B_0=V=\Gamma\backslash V(\mathfrak{X})$ the set of vertices of $G$; $B_1=B=\Gamma\backslash B(\mathfrak{X})$ the set of oriented edges, and more generally $B_k=\Gamma\backslash B_k(\mathfrak{X})$ the set of non-backtracking paths of length $k$ in $G$. 
In many places of the paper, elements of $\cH_k$ will be seen as functions on $B_k$~: an element $K$ defines a function $(x;y)\mapsto K(x, y)$ on $B(\mathfrak{X})$, and the equivariance property $K(\gamma \cdot x, \gamma\cdot y)=K(x, y)$ shows that this actually descends to a function on the quotient $B_k$. Depending on the context, we will use the notation $K(x, y)$, $K(x;y)$ or $K(x_0, x_1, \ldots, x_k)$ for the value of the function $K$ on a element of $B_k$. For $l\leq k$, we denote by $j_{l, k}$ is the isometric injection of $\cH_{l}$ into $\cH_{k}$, obtained by identifiying $\cH_{l}$ with the subspace of functions $K: B_k\To \IC$, such that $K(x_0, x_{1}, \ldots, x_k)$ depends only on $(x_0, x_{1}, \ldots, x_{l})$.

If $\omega=(x_0, \ldots, x_k)\in B_k$ (with $k\geq 2$), we define $mil(\omega)\in B_{k-2}$ by
$mil(\omega)=(x_1, \ldots, x_{k-1})$, in other words we remove the first and last segments of the path.

We introduce several operators :
an operator
$\cM : \cH_{k}\To  \cH_{k-2}$ defined by
$$\cM  K(\omega)=\frac1q \sum_{mil(\omega')=\omega} K(\omega')$$
for $k\geq 2$
(on $\cH_0$ or $\cH_1$ we let $\cM=0$).
Its adjoint $\cM ^* : \cH_{k-2}\To \cH_{k}$ ($k\geq 2$) has the expression
$$\cM ^* K(\omega)=\frac{1}q K(mil(\omega)).$$
We note that $  \cM \cM^*= I$ on $\cH_k$ for $k\geq 1$, which implies that 
\begin{equation}\norm{\cM}_{\cH_{k+2}\To\cH_{k+2}}=\norm{\cM^*}_{\cH_{k}\To\cH_{k}}=1.\label{e:norm1}\end{equation}
On $\cH_0$ we have $  \cM \cM^*=\frac{q+1}q I$.

Also, we introduce $\nabla:  \cH_{k-1}\To \cH_{k}$ by
$$\nabla f(x_0, \ldots, x_k)=f(x_1,\ldots, x_k)-f(x_0,\ldots, x_{k-1}).$$
The adjoint of $\nabla$ is $\nabla^*=-\cM \nabla$.

The central object of our argument is the  self-adjoint operator:
$\cL : \cH\To\cH$ defined so that $\widehat{\cL(K)}= [\cA, \widehat{K}]$.
A crucial point for proofs 2 and 3 of Theorem \ref{t:gen}, specific to the case of regular graphs, is that it factorizes as follows:
\begin{equation}\label{e:factL}
\cL=\nabla+\nabla^*=(I-\cM)\nabla.
\end{equation}

We will use repeatedly the Pythagorean theorem, in the following form: if $K^{(1)}, \ldots, K^{(n)}\in \cH$ are such that $K^{(j)}\in \cH_{i_j}$ with $j\not= k\Rightarrow i_j\not= i_k$, then
\begin{equation}\left\| \frac{K^{(1)}+\ldots+K^{(n)}}n\right\|_{\cH}\leq \frac{1}{\sqrt{n}}\max \left\|K^{(j)}\right\|_{\cH}.\label{e:pyth}\end{equation}

The last easy thing we will use, again specific to regular graphs, is the following lemma (the proof of which is obvious)\footnote{This is the only place where in the notation we do not distinguish between an operator and its kernel, i.e. its matrix in the canonical basis $(\delta_x)$. In this lemma the operators $\cA^m$ are seen as an element of $\cH$.}~:
\begin{lemma}\label{l:F}Let $\cF=\overline{Vect\{\cA^m, m\in \IN\}}\subset \cH$. Then $\cF$ is exactly the set of elements $K$ of $\cH$  such that
$$d_{\mathfrak X}(x, y)=d_{\mathfrak X}(x', y')\Rightarrow K(x, y)=K(x', y').$$
\end{lemma}

Finally, we will need two slightly more elaborate facts.
If $K\in \cH_k$ for some $k$, let us denote 
$$\la K\ra=\frac{1}{|V| \tau(k)}\sum_{x\in\cD, y\in\mathfrak{X}}K(x, y)$$
(recall that $\tau(k)=(q+1) q^{k-1}$ is the cardinality of a sphere of radius $k$ in the tree).
Let us introduce the element $\bbbone_k \in \cH_k$,
\begin{equation}\label{e:bbone}\bbbone_k(x, y)=  \bbbone_{d(x, y)=k},\end{equation}
so that $\tau(k)^{-1}\widehat{\bbbone_k}$ is the operator of averaging on spheres of radius $k$ (in particular $\widehat{\bbbone_0}$ is the identity).
Denoting by $\cH_k^o\subset \cH_k$ the subset of those $K$ such that $\la K\ra=0$, we note that
$$K-\la K\ra \bbbone_k \in \cH_k^o$$
for all $K\in \cH_k$.
If $G$ is connected then for every $k$, there exists $C(k, G)>0$ such that for all $K\in \cH_{k}$
\begin{equation}\norm{K-\la K\ra \bbbone_k }_{\cH_k}\leq C(k, G) \norm{\nabla K}_{\cH_{k+1}}.\label{e:ck}\end{equation}

The notation reflects the fact that $C(k, G)$ depends on $G$. A fact specific to regular graphs is that $C(k, G)$ can be expressed in terms of the spectral gap $\beta$:
\begin{prop}\label{p:Ck}One can take $$C(k, G)= C(k, \beta)=  k+\frac{1}{\beta^{2}}. $$
 
\end{prop}
\begin{proof}
Introduce the transfer operator $\cS: \cH_{k}\To  \cH_{k}$ defined by
$$\cS K(\omega)=\frac{1}{q}\sum_{\omega'\rightsquigarrow \omega} K(\omega').$$
 
 An application of the Cauchy-Schwarz inequality shows that
 $$\norm{(I-\cS)K}_{\cH}^2 \leq \frac{1}q \norm{\nabla K}_{\cH}^2.$$
 Also by Cauchy-Schwarz,
 \begin{equation*}\label{e:norms1}\norm{\cS}_{\cH\To\cH}\leq 1.
 \end{equation*}
 The spectrum of $\cS$ on $\cH_1$ was studied in \cite{ALM}, Section 7.1 (in the notation of that paper $\cS$ corresponds to the adjoint of $M^{\sharp}$). The study of the spectrum of $\cS$ on $\cH_{k}$ ($k\geq 2$)
reduces to that on $ \cH_1$, since for $l\geq k-1$ $\cS^l$ maps $\cH_{k}$ to $j_{1, k}(\cH_{1})$, and $\cS^l \circ j_{1, k}=j_{1, k}\circ \cS^l$.

 The explicit spectral decomposition of $\cS$ done in \cite{ALM} (Sections 7.1 and 7.2) and recalled in Section \ref{s:nbt} of this paper, shows that
 \begin{equation} \norm{\cS^m}_{\cH^o_1\To\cH^o_1}\leq (m +1)(1-\beta')^{m}
 \end{equation}
with 
 $$1-\beta'= \frac{2}{(q+1)\left(1-\beta- \sqrt{(1-\beta)^2-\frac{4q}{(q+1)^2}}\right)} .$$
 Note that $\beta'\sim \beta\frac{q+1}{q-1}$ for small $\beta$ (and conversely, $\beta$ bounded away from $0$ implies the same for $\beta'$).

 For $k>1$ and $m\geq k$,
 we have
 \begin{equation}\label{e:sg}\norm{\cS^m}_{\cH^o_k\To\cH^o_k}\leq \norm{\cS^k}_{\cH^o_k\To j_{1, k}(\cH^o_1)}\norm{\cS^{m-k}}_{\cH^o_1\To\cH^o_1}\leq ((m-k)+1)(1-\beta')^{m-k}
 \end{equation}
 (see \eqref{e:sg'}).
 
 Inequalities
 \eqref{e:norms1} and \eqref{e:sg} imply that
 \begin{equation}\label{e:inverse}\norm{(I-\cS)^{-1}}_{\cH^o_k\To\cH^o_k}\leq k +\sum_{m\geq 0 } (m +1) (1-\beta')^{m}\leq k+\frac{1}{\beta^{'2}}
 \end{equation}
 and thus
 $$\norm{K-\la K\ra \bbbone_k }_{\cH_k}^2\leq  \left(k+\frac{1}{\beta^{'2}}\right)^2 \norm{(I-\cS)K}_{\cH}^2 \leq \frac{1}q \left(k+\frac{1}{\beta^{'2}}\right)^2\norm{\nabla K}_{\cH}^2.$$
\end{proof}
\begin{remark}\label{r:Ck}
We see that $C(k, \beta)$ is an upper bound on
$\sum_{l=0}^{+\infty}\norm{\cS^l}_{\cH^o_k\To\cH^o_k}.$
This will be used in a few places of the paper, and $C(k, \beta)$ will always designate the quantity introduced here.
\end{remark}

 And here is the last estimate we will need: let $K\in\cH_k$.  We have seen that $K$ defines an operator $\widehat{K}$ acting on $C_c(V(\mathfrak{X}))$. This operator preserves the space of $\Gamma$-invariant functions, hence also define an operator
 $\widehat{K}_G$ on $\ell^2(V)$. The kernel of $\widehat{K}_G$ has the expression
 \begin{equation} \label{e:kg0}
 \la \delta_x, \widehat{K}_G\delta_y\ra_{\ell^2(V)} =\sum_{\gamma\in\Gamma}K (\tilde x, \gamma\cdot \tilde y)\end{equation}
 if $\tilde x, \tilde y$ are representatives in $\mathfrak{X}$ of $x, y\in G=\Gamma\backslash \mathfrak{X}$. We let
 $K_G(x, y)= \la \delta_x, \widehat{K}_G\delta_y\ra_{\ell^2(V)}$, in other words $K_G$ is the matrix of $\widehat{K}_G$ in the canonical basis of $\ell^2(V)$.
 The normalized Hilbert-Schmidt norm of $\widehat{K}_G$ is defined by
 \begin{equation}\label{e:defHSN}\norm{\widehat{K}_G}_{HSN}^2=\frac{1}{|V|}\sum_{x, y\in V}|{K}_G(x, y)|^2.\end{equation}
 
 \begin{prop}\label{p:norms}(i) If $K\in \cH_{\leq k}$ and $k$ is less than the minimal injectivity radius of $G$, then
 $$\norm{\widehat{K}_G}_{HSN}^2 =\norm{K}^2_{\cH}.$$
 In particular, the map $K\in \cH_k \mapsto \widehat{K}_G$ is injective if $k$ is less than the minimal injectivity radius of $G$.
 
 (ii) More generally,
 $$\left|\norm{\widehat{K}_G}_{HSN}^2- \norm{K}^2_{\cH} \right|\leq \tilde \tau(k)^2 \norm{K}^2_{\sup} \frac{1}{|V|}\sharp\{x\in V, \rho(x)\leq k\}.$$
 \end{prop}

\begin{proof} Write
 $$\norm{\widehat{K}_G}_{HSN}^2=\frac{1}{|V|}\sum_{\tilde x, \tilde y\in \cD}\left|\sum_{\gamma\in\Gamma}K (\tilde x, \gamma\cdot \tilde y)\right|^2.$$
 If $k$ is less than the minimal injectivity radius of $G$, then for each $\tilde x, \tilde y\in \cD$ there is at most one $\gamma$ such that $K (\tilde x, \gamma\cdot \tilde y)\not=0$.
 In that case, we can write
 \begin{eqnarray*}\norm{\widehat{K}_G}_{HSN}^2&=&\frac{1}{|V|}\sum_{\tilde x, \tilde y\in \cD}\sum_{\gamma\in\Gamma}|K (\tilde x, \gamma\cdot \tilde y)|^2\\
 &=& \frac{1}{|V|}\sum_{\tilde x\in \cD,  \tilde y\in \mathfrak{X}} \left|K (\tilde x, \tilde y)\right|^2=\norm{K}^2_{\cH}.
 \end{eqnarray*}
More generally, this argument remains true for those $\tilde x$ such that $\rho(\tilde x)>k$.
For the other ``bad'' points $\tilde x$ (such that $\rho(\tilde x)\leq k$), the sum $\sum_{\gamma\in\Gamma}K (\tilde x, \gamma\cdot \tilde y)$
contains at most $\tilde\tau(k)$ non-zero terms (the cardinality of a ball of radius $k$ in $\mathfrak{X}$), because these would correspond to distinct\footnote{Here we use the fact that $\Gamma$ acts without fixed points.}
points in the ball centered at $\tilde x$, of radius $k$. Thus for a ``bad'' point $\tilde x$, we write
$$\left|\sum_{\gamma\in\Gamma}K (\tilde x, \gamma\cdot \tilde y)\right|^2\leq  \tilde\tau(k)^2 \norm{K}^2_{\sup}.$$
This ends the proof of (ii).
\end{proof}

\subsection{Fourier analysis on the tree \label{s:Fouriertree}}
In our last proof of Theorem \ref{t:gen} (Proof 4, Section \ref{s:nbt}) we will need some Fourier analysis on the tree to describe explicitly the spectral decomposition of $\cA$ on $\ell^2(\mathfrak{X})$.
We refer to \cite{CS99} for a detailed presentation and we just gather here the few facts we will need.
 The use of Fourier analysis may seem a drawback of Proof 4, since we announced we aimed precisely at avoiding this. However, as we shall see in Section \ref{s:anis}, the following facts can be extended to other operators than $\cA$, provided they are {\em{nearest-neighbour operators}}, and that we work in a region of the spectrum where the Green function has a finite limit on the real axis (which implies absolute continuity of the spectral measure, commonly associated with delocalization of eigenfunctions). Note also that, contrary to \cite{ALM}, we will not use any Paley-Wiener type theorem.

Fix a vertex $o$ in $\mathfrak{X}$, called the origin. Introduce the Green function,
$$g_\gamma(x)=(\gamma-\cA)^{-1}(x, o)$$
for $\gamma\in \IC\setminus \IR$.
Its explicit expression is
$$g_\gamma(x)=\frac{q^{-\alpha d(x, o)}}{q^{\alpha}-q^{-\alpha}}$$
where $\alpha$ is the only solution (modulo $2i\pi/\ln q$) of
$$\gamma=q^{1-\alpha}+q^{\alpha}$$
such that $\Re e\,\alpha > 1/2$.
We see that for every $\lambda\in\R$, $g_\gamma(x)$ has a limit when 
$\gamma=\lambda + i\eps$ and $\eps>0$ tends to $0$; similarly when $\gamma=\lambda - i\eps$ and $\eps>0$ tends to $0$. These limits are denoted by $g_{\lambda \pm i0}(x)$~:
$$g_{\lambda + i0}(x)=\frac{q^{-(1/2+is ) d(x, o)}}{q^{1/2+is }-q^{-(1/2+is)}}$$
where $s$ should be chosen such that $\sin(s\ln q)\geq 0$.

The spectral measure of the operator $\cA$ on $\ell^2( \mathfrak X)$ for the state
$\delta_o$ (Dirac mass at the origin $o$) is absolutely continuous, given by
$${\mathrm{m}}(\lambda)d\lambda=-\frac{1}{\pi} \Im m g_{\lambda + i 0}(o) d\lambda = -\frac{1}{2\pi i} (g_{\lambda + i 0}(o)-  g_{\lambda - i 0}(o))d\lambda.$$
This measure is also often referred to as the Plancherel measure for $\mathfrak{X}$, or the Kesten measure \cite{Kes59}. It is supported in $(-2\sqrt{q}, 2\sqrt{q})$, where it is smooth and positive.

Let us also introduce the ``Poisson kernel''. Denote by $\partial  \mathfrak X$ the boundary at infinity of the tree; here it can be identified with the set of infinite non-backtracking paths starting at the origin $o$. For $x\in  \mathfrak X$ and $\omega\in \partial  \mathfrak X$ and $\gamma\not\in \R$,
$$P_\gamma(x, \omega)=\lim_{y\in \mathfrak X, y\To \omega}\frac{g_\gamma(x, y)}{g_\gamma(o, y)}=q^{- \alpha h_\omega(x)}$$
where $h_{\omega}(x)=\lim_{y\To \omega}d(x, y)-d(o, y)$ is the Busemann function centered at $\omega\in \partial  \mathfrak X$.
Actually, this is also well defined for $\gamma=\lambda + i0$, by replacing $\alpha$ by $-(1/2+is)$ with $\sin(s\ln q)\geq 0$. Denote by $P_{\lambda, \omega}$
the function on $\mathfrak X$, $x\mapsto P_{\lambda + i0}(x, \omega)$.

 Let $\nu$ be the harmonic measure on $\partial  \mathfrak X$, seen from $o$.  
 If $(o, x_{i_1},\ldots, x_{i_M})$ is a non-backtracking path of length $M$ starting at $o$, and if $[o, x_{i_1}, \ldots, x_{i_M}]\subset \partial \mathfrak X$ denotes the set of infinite paths starting with $o, x_{i_1},\ldots, x_{i_M}$,
it has the simple expression
$$\nu([o, x_{i_1}, \ldots, x_{i_M}])=(q+1)^{-1} q^{-(M-1)} = \frac{1}{\tau(M)}.$$
 
The functions $P_{\lambda, \omega}$ are generalized eigenfunctions of $\cA$.
For any $f\in \ell^2(\mathfrak X)$, the spectral decomposition of $f$ for the operator $\cA$ reads~:
\begin{equation} f(x)= \int_{ \partial \mathfrak X\times \IR} \la P_{\lambda, \omega} , f\ra P_{\lambda, \omega}(x ) d\nu(\omega) {\mathrm{m}}(\lambda) d\lambda
\end{equation}
and the Plancherel formula,
\begin{equation}
\norm{f}^2_2=\int_{ \partial\mathfrak X\times \IR} | \la P_{\lambda, \omega} , f\ra |^2 d\nu(\omega) {\mathrm{m}}(\lambda) d\lambda.
\end{equation}

For $\widehat{K}$ an operator on $\ell^2( \mathfrak X)$, with kernel $K(x, y)$ (assuming for instance $K$ is compactly supported) we also have
\begin{equation}\label{e:trace}\Tr_{\ell^2( \mathfrak X)} \widehat{K}=\sum_{x\in \mathfrak X} K(x, x)=\int_{ \partial \mathfrak X\times \IR} \la   P_{\lambda, \omega}, \widehat{K}   P_{\lambda, \omega}\ra  d\nu(\omega) {\mathrm{m}}(\lambda) d\lambda,\end{equation}
and similarly
\begin{equation}\label{e:HSX}
\norm{ \widehat{K}}^2_{KS(\ell^2(\mathfrak X)}=\Tr_{\ell^2( \mathfrak X)}  \widehat{K}^*\widehat{K}=\sum_{x, y\in \mathfrak X} |K(x, y)|^2
\int_{  \partial\mathfrak X\times \IR}  \norm{ \widehat{K}   P_{\lambda, \omega}}^2_{\ell^2(\mathfrak X)}\ra  d\nu(\omega) {\mathrm{m}}(\lambda) d\lambda.\end{equation}

Note that all these formulas are independent of the choice of the origin $o$ used in the definition of $\nu$ and of the functions $P_{\lambda, \omega}$. This can be seen from the fact that if $\nu_{o'}$ is the harmonic measure seen from a point $o' \in \mathfrak X$, we have
\begin{equation}\label{e:harmonic}\frac{d\nu_{o'}}{d\nu_{o}}(\omega)=q^{h_\omega(o)-h_\omega(o')}.\end{equation}

\section{Preparatory remarks \label{s:prep}}
The three proofs of Theorem \ref{t:gen} contained in this paper will be numbered 2, 3, 4 (Proof 1 will refer to the original proof in \cite{ALM}). We are going to prove the result announced in Remark \ref{r:rate}; so from now on we work with a fixed finite graph and drop the index $N$ from our notation. The graphs will always be assumed to be connected.

Before we start proving Theorem \ref{t:gen} in \S \ref{s:short}, let us make a few simplifying remarks on the operator $\mathbf K$ it involves.
First, we note that any operator $\mathbf K$ supported at distance $\leq D$ from the diagonal, as in Theorem \ref{t:gen}, is a sum of $D+1$ operators, $\mathbf K=\sum_{m=0}^D \mathbf K_m$ where $\mathbf K_m$ is such that $  K_m(x, y)\not= 0\Longrightarrow d_G(x, y)=m$. Thus it is enough to fix $m$ and to prove the theorem for an operator $\mathbf K_m$ satisfying this assumption.

Second, the operator $\mathbf K_m$ can always be written as $\mathbf K_m=\widehat{K}_G$ for some $K\in \cH_m$. Recall from \eqref{e:kg0} that for $x, y\in V$, the kernel of $\widehat{K}_G$ is
\begin{equation} \label{e:kg}K_G(x, y)=\sum_{\gamma\in\Gamma} K(\tilde x, \gamma\cdot \tilde y)\end{equation}
where $\tilde x, \tilde y$ are lifts of $x, y$ in $\mathfrak X$. There may be several $K$ yielding the same $\widehat{K}_G$, but we may always choose $K$ so that there is always at most one non-zero term in the sum \eqref{e:kg},
corresponding to an element $\gamma$ such that $d_{\mathfrak X}(\tilde x, \gamma\cdot \tilde y)=d_G(x, y)$. For this choice, the quantity $\sum_{x, y\in |V|} K_G(x, y)\Phi_{\lambda_j}(d_G(x, y))$ is the same as
$ \sum_{x\in\cD, y\in {\mathfrak{X}}} K(x, y)\Phi_{\lambda_j}(d_{\mathfrak X}(x, y))$. Note also that in this case $\norm{\widehat{K}_G}_{HSN}=\norm{K}_{\cH_m}.$

So it is enough to prove the theorem for an operator $\mathbf K=\widehat{K}_G$ with $K\in \cH_m$, $m\leq D$ fixed.
We finally check that it is enough to prove Theorem \ref{t:gen} in the case of an operator $\mathbf K=\widehat{K}_G$ with $K\in \cH_m^o$ for some fixed $m$. In this case, $\la \mathbf K \ra_{\lambda}=0$ for all $\lambda$, and the result we have to prove simply reads 
\begin{equation}\label{e:part}\frac1{|V|} \sum_{j=1}^{|V|} \left| \la \psi_j, \widehat{K}_G\psi_j\ra  
\right|^2 \leq  
\tilde C(m, \beta)\left(\frac{{\norm{K}_{\cH_m}^2}}{R^{\alpha}} +q^{\tilde\alpha R} \frac{1}{|V|}\sharp\{x\in V, \rho(x)\leq R\}\norm{K}^2_{sup}\right)
\end{equation}
(or the weaker type of estimate \eqref{e:weaker} for Proof 4).

Assuming \eqref{e:part} holds for $K\in \cH_m^o$, let us check that the general result for $K\in \cH_m$ follows~: if $K\in \cH_m$, then we replace $K$ by $K-\la K\ra \bbbone_m$ which lies in $\cH_m^o$.
The operator $\widehat{\bbbone_m}$ is an explicit function of $\cA$~: $\widehat{\bbbone_m}=h_m(\cA)$ where
$$h_{\mathrm{m}}(\lambda)=\tau(m) q^{-m /2}\left(\frac{2}{q+1}\cos(m s\ln q) +\frac{q-1}{q+1}\frac{\sin((m+1)s\ln q)}{\sin(s\ln q)}\right)$$
for $\lambda=2\sqrt{q}\cos(s\ln q).$
We thus have
$${(\widehat{\bbbone_m})}_G  \psi_j  =  h_m(\lambda_j)  \psi_j.$$
After adjustment of the factor  $\tilde C(m, \beta)$, equation \eqref{e:part} for operators in $\cH_m^o$ implies
\begin{equation}\frac1{|V|} \sum_{j=1}^{|V|} \left| \la \psi_j, \widehat{K}_G\psi_j\ra  - h_m(\lambda_j) \la K\ra  
\right|^2 \leq  \tilde C(m, \beta)\left(\frac{{\norm{K}_{\cH_m}^2}}{R^{\alpha}} +q^{\tilde\alpha R} \frac{1}{|V|}\sharp\{x\in V, \rho(x)\leq R\}\norm{K}^2_{sup}\right).
\end{equation}
for any $K\in \cH_m$. This is the desired result~: just note that for $K\in \cH_m$ we have
$$h_m(\lambda_j) \la K\ra  =\la \mathbf K\ra_{\lambda_j}.$$

We are almost ready to begin with our proof. From now on, let us fix $m$. For $K\in \cH_m$, define the ``quantum variance'',
\begin{equation}\label{e:defvar}Var(K)=\frac1{|V|} \sum_{j=1}^{|V|} \left| \la \psi_j, \widehat{K}_G\psi_j\ra  
\right|^2.\end{equation}
The goal is to show that for $K\in \cH_m^0$ we have \eqref{e:part}.
 (or the weaker form \eqref{e:weaker} for Proof 4).

Obviously, 
$$Var(K)\leq \norm{ \widehat{K}_G}^2_{HSN}.$$
We will also use repeatedly the fact that $Var(K_1+K_2)\leq 2Var(K_1)+ 2Var(K_2).$

For Proofs 2 and 4, the following technical remark will be useful~:
\begin{lemma}\label{l:useful}
Assume we have shown an estimate of the form
$$Var((1-\cS)K)\leq \frac{C(m)}{n}\norm{K}^2_{\cH_m} +\norm{K}^2_{sup} f_{m, G}(n)$$
for all $K\in \cH_m$ ($m\geq 1$) and all $n$.

(i) Then for all $K\in \cH^0_m$, for any integer $n$, we have
\begin{multline*}
Var(K)
\leq 2\frac{C(m)C(m, \beta)^2}{n}\norm{K}^2_{\cH^0_m} +2T^2\norm{K}^2_{sup} f_{m, G}(n)+  2\frac{C(m, \beta)^2}{T^2}\norm{K}^2_{\cH_m^0}
\\+2\norm{K}^2_{sup} \tilde\tau(m)^2\frac{1}{|V|}\sharp\{
x\in V, \rho(x)\leq m\}.
\end{multline*}

(ii) Taking $n=T^2$, we obtain that for all $K\in \cH^0_m$, for any integer $n$, we have
\begin{multline*}Var(K)\leq \frac{2 C(m,\beta)^2(C(m)+ 1)}{n}\norm{K}^2_{\cH^0_m} \\
+2 \norm{K}^2_{sup} \left( n f_{m, G}(n) + \tilde\tau(m)^2\frac{1}{|V|}\sharp\{
x\in V, \rho(x)\leq m\}\right)
\end{multline*}
\end{lemma}
The constant $C(m,\beta)$ is the one defined in Proposition \ref{p:Ck}.
The same lemma holds for $m=0$, replacing $1-\cS$ by $1-\frac{\cA}{q+1}$, and taking $C(0,\beta)=\beta^{-1}$.

\begin{proof}
For $K\in \cH^0_m$, we know that there exists a unique $K'\in \cH^0_m$ such that $K=(1-\cS)K'$, and moreover $\norm{K'}^2_{\cH_m^0}\leq C(m,\beta) \norm{K}^2_{\cH_m^0}$ (see \eqref{e:inverse}).
Thus our assumption implies that
$$Var(K)\leq \frac{C(m)C(m,\beta)}{n}\norm{K}^2_{\cH_m} +\norm{K'}^2_{sup} f_{m, G}(n).$$
The problem is that the spectral gap assumption gives no information about $\norm{K'}_{sup}$.

To solve this problem, fix an integer $T$ and consider an approximation of $(1-\cS)^{-1} K$, defined by
$$K'_T=\frac{1}T\sum_{k=0}^{T-2}(T-1-k)\cS^k K.$$
It is built so that $(1-\cS) K'_T=K- \frac{1}T\sum_{k=0}^{T-1}\cS^k K$. Clearly, $\norm{K'_T}_{sup} \leq T \norm{K}_{sup}$.
By Remark \ref{r:Ck},
$$\left\Vert \frac{1}T\sum_{k=0}^{T-1}\cS^k K \right\Vert_{\cH_m^0} \leq \frac{C(m, \beta)}{T}\norm{K}_{\cH_m^0}$$
and
$$\norm{K'_T}_{\cH_m^0}\leq \sum_{k=0}^{T-2}\norm{\cS^k K}_{\cH_m^0}\leq C(m, \beta)\norm{K}_{\cH_m^0}.$$
Applying our assumption and Proposition \ref{p:norms}, we get
\begin{eqnarray*}
Var(K)&\leq& 2Var((1-\cS) K'_T) + 2Var \left(\frac{1}T\sum_{k=0}^{T-1}\cS^k K\right)\\
&\leq& 2\frac{C(m)C(m, \beta)^2}{n}\norm{K}^2_{\cH^0_m} +2T^2\norm{K}^2_{sup} f_{m, G}(n) +  2\left\Vert \frac{1}T\sum_{k=0}^{T-1}\widehat{(\cS^k K)}_G \right\Vert_{HSN}^2\\
&\leq&2\frac{C(m)C(m, \beta)^2}{n}\norm{K}^2_{\cH^0_m} +2T^2\norm{K}^2_{sup} f_{m, G}(n)+  2\frac{C(m, \beta)^2}{T^2}\norm{K}^2_{\cH_m^0}
\\&&+2\norm{K}^2_{sup} \tilde\tau(m)^2\frac{1}{|V|}\sharp\{
x\in V, \rho(x)\leq m\}.
\end{eqnarray*}
\end{proof}

{\bf{Notation~:}} In what follows, the notation $o_\Lambda(1)_{u\To v}$ will be used to mean ``a quantity tending to $0$ when $u$ tends to $v$, $\Lambda$ being fixed'' ($\Lambda, u, v$ can be any kind of object, integers, graphs, etc.).

\bigskip

{\bf{``Constants that vary from line to line''~:}} For more clarity, the notation $C(k, \beta)$ will {\em{always}} refer to the quantity introduced in Proposition \ref{p:Ck} and Remark \ref{r:Ck}.

The notation $c(k)>1$ will always mean a constant (depending on $k$) such that, for $K\in \cH_k$,
$$c(k)^{-1}\norm{\widehat{K}}_{\ell^2(\mathfrak{X})\To \ell^2(\mathfrak{X})}\leq \norm{K}_{\sup}\leq  c(k)\norm{\widehat{K}}_{\ell^2(\mathfrak{X})\To \ell^2(\mathfrak{X})}$$
and
$$\norm{\widehat{K}_G}_{\ell^2(V)\To \ell^2(V)}\leq c(k) \norm{K}_{\sup}.$$
In Proof 4 we shall also ask that
$$\norm{\widehat{K}_B}_{\ell^2(B)\To \ell^2(B)} 
\leq  c(k)\norm{K}_{sup}.$$
The quantity $c(k)$ can be taken independent of the graph $G$.

All the other factors, like $C(m), \tilde C(m, \beta),$ etc. are expressions that can, in theory, be made explicit, and whose expression can be adjusted from line to line.

\section{Proof 2: ultra-short\label{s:short}}
We start with our shortest proof, which in addition gives the best exponent $\alpha=1$ in the estimate \eqref{e:part} of the rate of decay of the variance as a function of the girth. This proof is especially short if the girth goes to infinity.
However, this proof seems less susceptible of adaptation to other models.

 Consider as before the Hilbert space
$HSN(\ell^2(V)),$
space of linear operators on $\ell^2(V)$ endowed
with the normalized Hilbert-Schmidt norm,
$$\norm{\mathbf K}^2_{HSN}=\frac1{|V|}\sum_{x, y\in V}|K(x, y)|^2$$
if $K$ is the matrix of $\mathbf K$ in the canonical basis.
As we have seen in Proposition \ref{p:norms}, there is a map $K\mapsto \widehat{K}_G$ from $\cH_m$ to $HSN(\ell^2(V))$, which is an isometry if $m$ is less than the minimal injectivity radius.

For $K \in \cH_{m}$ (or more generally $\cH_{\leq m}$ for some integer $m$), define the quantum variance :
$Var(K)=\frac{1}{|V| }\sum_{j=1}^{|V|} \left| \la \psi_j, \widehat{K}_G\psi_j\ra  
\right|^2.$
An easy but important step of the proof is to note that 
\begin{equation}\label{e:estHS}Var(K)\leq \norm{\widehat{K}_G}^2_{HSN}.\end{equation}
The quantum variance obviously has the following property :
$$Var( \cL K)=0$$
for all $K$, where, recall, $\cL$ is such that $\widehat{\cL K}=[\cA, \widehat K]$.

We now use the fact that
$$ { \cL }=(1-\cM)\nabla=\nabla+\nabla^*$$
(see \eqref{e:factL}) which implies
$$Var( (\nabla+\nabla^*)(K))=0$$
and (since $Var(K_1+K_2)=0$ implies $Var(K_1)=Var(K_2)$),
$$Var(\nabla^*K)=Var(\nabla K).$$
Let also note that $\cM \nabla \cM^*=\nabla$. Thus, applying the identity $Var( (1-\cM)\nabla K)=0$ with $K$ replaced by $\cM ^* K$, we obtain
$$Var( \nabla(I-\cM^*)K)=0$$
which implies
$$Var( \nabla K)=Var( \nabla \cM^* K).$$
Applying this with $K$ replaced by $\frac1n\sum_{k=0}^{n-2}(n-1-k)\cM^{* k}K$, we find for all $n$
$$Var( \nabla K)=Var\left(\nabla \Sigma^n K\right)$$
where
\begin{equation}\label{e:Sm}\Sigma^n \defeq\frac1n (I+\cM^*+\ldots+\cM^{* (n-1)}).\end{equation}
 Finally we have got
$$Var(\nabla^* K)=Var\left(\nabla \Sigma^n K\right)$$
for all $n$.

If $K\in \cH_m$, then $\nabla\Sigma^n K\in \cH_{ m+2n +1}$. Applying Proposition \ref{p:norms},
\begin{multline*}Var\left(\nabla \Sigma^n K\right)\leq  \norm{\widehat{\nabla \Sigma^n K}_G}^2_{HSN}
\\ \leq
\norm{\nabla \Sigma^n K}^2_{\cH} +  \tilde\tau(m+2n +1)^2 \norm{\nabla \Sigma^n K}^2_{\sup} \frac{1}{|V|}\sharp\{x\in V, \rho(x)\leq m+2n +1\}.
\end{multline*}
Note that $\norm{\nabla}_{\cH\To \cH}\leq 2\sqrt{q}$, $\norm{\nabla}_{sup\To sup}\leq 2 $, that $\norm{\Sigma^n}_{\cH\To \cH}\leq \left(\frac{q+1}{q}\right)^{1/2}$ for all $n$ (see \eqref{e:norm1}),
and that  $\norm{\Sigma^n}_{sup\To sup}\leq q^n$.

If $K\in \cH_m$, then $\nabla \cM^{* k} K\in \cH_{ m+2k +1}$, and thus the family $(\nabla \cM^{* k} K)_{k\in \IN}$ is orthogonal in $\cH$. Using Pythagoras' theorem in the form \eqref{e:pyth} to evaluate $\norm{\nabla \Sigma^n K}^2_{\cH}$,
we obtain
\begin{equation}\label{e:varnabla}Var\left(\nabla^* K\right)\leq  \frac{4(q+1)}{n}\norm{ K}^2_{{\cH}}+  8(q+1) q^n \tilde\tau(m+2n+1)^2 \norm{ K}^2_{\sup} \frac{1}{|V|}\sharp\{x\in V, \rho(x)\leq m+2n+1\}.
\end{equation}
To conclude the proof, we note that for each $K\in\cH^0_m$, we can write
$$ K=\nabla^* \nabla K'$$ with $K'\in \cH_m$ and $\norm{K'}_{\cH}\leq C(m, \beta)^2\norm{K}_{\cH}$ where $C(m, \beta)$ is as in Proposition \ref{p:Ck}.
Thus for $K\in\cH^0_m$
\begin{equation*}Var\left(K\right)\leq 16q(q+1) \frac{ C(m, \beta)^2}{n}\norm{ K}^2_{{\cH}}+ 8(q+1) q^n \tilde\tau(m+2n+2)^2\norm{ \nabla K'}^2_{\sup} \frac{1}{|V|}\sharp\{x\in V, \rho(x)\leq m+2n+2\}.
\end{equation*}
At this point we are stuck if we are not able to estimate  $\norm{  K'}^2_{\sup}$ with constants that do not depend on the graph. So let us start with the simpler case
where $\rho(x) > m+2n+2$ for all $x$. We find
\begin{equation*}Var\left(K\right)\leq 16q(q+1) C(m, \beta)^2 \frac{1}{n}\norm{ K}^2_{{\cH}}
\end{equation*}
for any $n$ such that $m+2n+2$ is less than the minimal injectivity radius. Assuming the girth goes to infinity, we thus find
\begin{equation*}Var\left(K\right)\leq \tilde C(m, \beta) \norm{ K}^2_{{\cH}} \, girth(G)^{-1},
\end{equation*}
for all $K\in\cH^0_m$, that is a statement of the form \eqref{e:girth} with $\alpha=1$.

If the girth does not grow to infinity, the proof still works at the expense of a few more lines. Start again at \eqref{e:varnabla}.
 If $j_{m-1,m}$ is the injection of $\cH_{m-1}$ into $\cH_m$ describe in \S \ref{s:notation}, we note the identity $\nabla^*\circ  j_{m-1,m}=q(I-\cS)$ on $\cH^0_{m-1}$.
 Thus, taking $K'\in \cH_{m-1}$ and applying \eqref{e:varnabla} to $K= j_{m-1,m} K'$, we obtain
\begin{equation*} Var\left((1-\cS) K'\right)\leq  \frac{4(q+1)}{n}\norm{ K}^2_{{\cH}}+  8(q+1) q^n \tilde\tau(m+2n+1)^2 \norm{ K}^2_{\sup} \frac{1}{|V|}\sharp\{x\in V, \rho(x)\leq m+2n+1\}.
\end{equation*}
Applying Lemma \ref{l:useful}, (and changing $m-1$ into $m$)
we obtain
\begin{multline*}Var(K)\leq \frac{2C(m,\beta)^2(4q+5)}{n}\norm{K}^2_{\cH_m} 
\\ +2 \norm{ K}^2_{\sup}\left( 8n(q+1) q^n \tilde\tau(m+2n+2)^2 \frac{1}{|V|}\sharp\{x\in V, \rho(x)\leq m+2n+2\}
+ \tilde\tau(m)^2\frac{1}{|V|}\sharp\{
x\in V, \rho(x)\leq m\}\right) 
\end{multline*}
This is again of the desired form \eqref{e:part}. 

\section{Proof 3 : ergodic-flavoured proof\label{s:ergflav}} 

As a little warm-up let us prove the following lemma. We use the notation from Section \ref{s:notation}.
\begin{lemma} \label{l:basic}The null-space of the operator $\cL$ in $\cH$ is the subspace $\cF$.
\end{lemma}
 

\begin{proof}  Because of \eqref{e:factL}, $\cL(K)=0$ is equivalent to
$\nabla K= \cM\nabla K$, and because $\cM $ maps $\cH_k$ to  $\cH_{k-2}$ and has norm $1$, this implies
$$\norm{(\nabla K)_{k-2}}_{\cH}\leq \norm{(\nabla K)_k}_{\cH}.$$
On the other hand, since $\nabla K\in\cH$, $\norm{(\nabla K)_k}_{\cH}$ goes to $0$ as $k\To +\infty$.
Thus we must have $\nabla K=0$.  Because of Lemma \ref{l:F}, we must have $K\in\cF$.
\end{proof}

{\bf{Question~:}} It is an interesting (and apparently difficult) question to understand if Lemma \ref{l:basic} extends to non-regular graphs. We have in mind the following~: let $\cT$ be an infinite tree (without leaves), and assume that there is a group $\Gamma$ of automorphisms, acting without fixed vertices, such that $G=\Gamma\backslash \cT$ is a finite graph.
Call $\cD\subset V( {\cT})$ a fundamental domain for the action of $\Gamma$ on $V( {\cT})$. Its cardinality is the number of vertices of $G$.
Introduce the Hilbert space $\cH$ of functions $K:V(\cT)\times V(\cT)\To \IC$ -- such that $K(\gamma\cdot x, \gamma \cdot y)=K(x, y)$ for all $\gamma \in \Gamma$, and such that
$$\frac{1}{|\cD|}\sum_{x\in \cD, y\in\mathfrak{X}} |K(x, y)|^2<+\infty.$$
As before, each element $K$ on $\cH$ defines an operator $\widehat K$ on $\ell^2(\cT)$. Define an operator $\cL: \cH\To \cH$ by the identity $\widehat{\cL K}=[\cA, \widehat K]$.

Let $\cF=\overline{Vect\{\cA^m, m\in \IN\}}\subset \cH$.

{\em{Under which conditions on $\cT$ and $\Gamma$ is it true that the null-space of $\cL: K\mapsto [\cA, K]$ in $\cH$ is $\cF$ ?}}
Although we do not have the answer, we can make a few more comments\footnote{We could ask the same questions with the Laplacian $\Delta$ instead of the adjacency matric $\cA$; for a non-regular graph these operators are no longer related by \eqref{e:qI}.}.

\begin{itemize}
\item the desired property (namely that the null-space of $\cL$ in $\cH$ is $\cF_{\cA}$) is a form of ergodicity of the unitary flow $(e^{it\cL})$ : it says that the only elements of the Hilbert space $\cH$ that are invariant under the flow are the trivial ones, that is, the elements of $\cF_{\cA}$.
\item it is quite simple to show (even for non regular trees) that, for any $k$, the elements of $\cH_{\leq k}$ that commute with $\cA$ are polynomials of degree $k$ in $\cA$. This does not  trivially imply the desired property for all elements of $\cH$. Continuing the analogy with ergodic theory, this is similar to the difference between topological transitivity (``the only smooth invariant functions are the trivial ones'') and ergodicity (``the only $L^2$ invariant functions are the trivial ones''). Topological transitivity does not automatically imply ergodicity.
\item in spite of the fact that we cannot prove Lemma \ref{l:basic} for non-regular graph, we still believe that it should hold in great generality. It is usual that ergodic properties are easier to prove for homogeneous geometries than for non-homogeneous ones, because one can use explicit algebraic calculations that are not available otherwise. For instance, ergodicity of hyperbolic {\em{linear}} torus aumorphisms is extremely simple to prove using Fourier series, whereas for the non-algebraic ones one needs to use the general theory of hyperbolic dynamical systems, in particular the existence and absolute continuity of stable/unstable foliations. Recall that ergodicity of the geodesic flow was proven by Hopf and Hedlund in the 30s for compact hyperbolic manifolds \cite{Hedl, Hopf} but only at the end of the 60s by Anosov for manifolds of variable negative curvature \cite{Anosov}. The explicit rate of ergodicity (or mixing) was obtained in the 70s by Ratner for hyperbolic surfaces using representation theory \cite{Rat} but only in 98 by Dolgopyat \cite{Dolgo} for surfaces
of variable negative curvature (see also \cite{Liv}). This lapse of time indicates the difficulty of treating {{\em non-homogeneous}} geometries.
\end{itemize}

\bigskip

Back to the case of regular graphs, we want a quantitative version of Lemma \ref{l:basic}:
\begin{lemma}  
Let $K\in\cH$ be such that $\norm{\cL(K)}_{\cH}\leq \delta <1$ and such that $K\perp \cF$ in $\cH$. Then, for all $m$, $$ \norm{K_m}_{\cH}\leq  C(m, \beta)\frac{\delta^{1/3}}{1-\delta^{1/3}}$$
where $C(m, \beta)$ is the constant appearing in Proposition \ref{p:Ck}.
\end{lemma}

\begin{proof} We have $\norm{(I-\cM)\nabla K}\leq \delta$. Since $\cM$ has norm $\leq 1$, applying $(I+\cM+\ldots+\cM^{n-1})$ to $(I-\cM)\nabla K$ we obtain the inequality
$\norm{(I-\cM^n)\nabla K}_{\cH}\leq n\delta$ for any integer $n$. Thus, for any integers $m$ and $n$, for every $u\in \cH_m$,
$$\left|\la \nabla K, u\ra_{\cH}\right|\leq \left|\la \nabla K, \Sigma^n u\ra\right| +n\delta\norm{u}_{\cH}$$
where $\Sigma^n$ was defined in \eqref{e:Sm}.

By Pythagoras in the form \eqref{e:pyth}, we have $\norm{\Sigma^n u}_{\cH}\leq    n^{-1/2}\norm{u}_{\cH}$.
Taking $n$ of order $\delta^{-2/3}$ we obtain
$$\left|\la \nabla K, u\ra_{\cH}\right|\leq \delta^{1/3}\norm{u}_{\cH}(1+\norm{\nabla K}_{\cH})$$
(for every $u\in \cH_m$) and thus
$$\norm{(\nabla K)_m}_{\cH}(1-\delta^{1/3})\leq \delta^{1/3}$$
for all $m$.
From Proposition \ref{p:Ck} this implies that
$$\norm{K_{m} - \la K_{m}\ra \bbbone_m  }_{\cH}\leq     C(m,\beta)\frac{\delta^{1/3}}{1-\delta^{1/3}}$$ and because we assumed $\la K_{m}\ra=0$ (for all $m$) we get the announced result.
\end{proof}

This lemma implies the following:

\begin{lemma} If $K\in \cH^0_m$ then
$$\left\|\frac1T\int_{0}^T e^{it\cL}K dt\right\|_{\cH}\leq \frac{( C(m,\beta)^{1/2} +16)\norm{K}}{T^{1/7}}.$$
\end{lemma}

\begin{proof} We may assume that $\norm{K}_{\cH}\leq 1$. Use the spectral decomposition of $\cL$ to write
$$K=\, \bbbone_{[-\delta, \delta]}(\cL)(K) + \bbbone_{[-\delta, \delta]^c}(\cL)(K) =: K_{(1)}+K_{(2)}$$
where $\bbbone_{I}(\cL)$ stands for the spectral projector associated to $\cL$ on the interval $I$.
Obviously,
$$\left\|\frac1T\int_{0}^T e^{it\cL}K_{(2)} dt\right\|_{\cH}\leq \frac{2 }{\delta T}.$$
On the other hand, by the previous lemma, 
$$|\la K_{(1)}, K\ra|\leq   C(m,\beta)\frac{\delta^{1/3}}{1-\delta^{1/3}} $$
and thus
$$\norm{K_{(1)}}^2\leq    C(m,\beta)\frac{\delta^{1/3}}{1-\delta^{1/3}}  .$$
Thus,
$$\left\|\frac1T\int_{0}^T e^{it\cL}K dt\right\|_{\cH}\leq  C(m, \beta)^{1/2} \frac{\delta^{1/6}}{(1-\delta^{1/3})^{1/2}}   + \frac{2}{\delta T}.$$
Here $\delta>0$ was arbitrary : we now take $\delta= \frac{T^{-6/7}}{(1+T^{-2/7})^3}\sim T^{-6/7}$ to get the announced result.
\end{proof}
We are now ready to prove Theorem \ref{t:gen}.
 
\begin{proof} Define the quantum variance as in \eqref{e:defvar}.
Since $Var(\cL(K))=0$ for all $K$, the quantum variance has the following invariance property
$$Var(K)=Var\left(\frac1T\int_0^Te^{it { \cL }} K dt\right),$$
for any $T$, or in fact we will use the approximate version of the exponential,
$$Var(K)=Var\left(\frac1T\int_0^T \sum_{j\leq M}\frac{(it { \cL })^j K}{j!} dt\right),$$
for any $T$ and any $M$.


Let us assume from now on that $K\in \cH_m^0$.
Let $1>\eps>0$. Choose $T=T(\eps)$ such that $\frac{( C(m,\beta)^{1/2} +16) }{T^{1/7}}\leq \eps^{1/2}$ and thus
\begin{equation}\label{e:ergo} \left\|\frac1T\int_{0}^T e^{it\cL}K dt \right\|_{\cH}\leq \eps \norm{K}.\end{equation}
Since 
$\cL$
is a bounded operator (with norms independent of our graph, actually just a function of $q$),
we can find $M =M(T(\eps))$ such that
$$ \left\|e^{it\cL}-\sum_{j\leq M}\frac{(it\cL)^j}{j!}  \right\|_{\cH}\leq \eps$$
for all $t\leq T(\eps)$. Actually $M =5 \norm{ { \cL }}_{\cH\To \cH} {T(\eps)}$ amply does the job.

It follows in particular that
 \begin{equation}\label{e:est1} \left\|\frac1T \int_0^T \sum_{j\leq M}\frac{(it{ \cL })^j K}{j!} dt \right\|^2_{{ \cH}} \leq \left\|\frac1T\int_0^Te^{it { \cL }} K dt \right\|^2_{ \cH  } +\eps(2+\eps)\norm{K}^2_{\cH}
 \end{equation}
 and 
\begin{equation} \label{e:est2} \left\|\frac1T \int_0^T \sum_{j\leq M}\frac{(it{ \cL })^j K}{j!} dt \right\|^2_{{ \cH}} \leq (1+\eps)^2\norm{K}^2_{\cH}\leq 4\norm{K}^2_{\cH}.\end{equation}

Taking $R\geq m+2M$, we now use Proposition \ref{p:norms} to write
\begin{multline}\label{e:estP1}\left\| \frac1T\int_0^T \sum_{j\leq M}\frac{(it\widehat{{ \cL })^j K}_G}{j!} dt \right\|^2_{HSN}\leq 
\left\|\frac1T\int_0^T \sum_{j\leq M}\frac{(it\cL)^j K}{j!} dt \right\|_{\cH}^2\\
+4\tilde\tau(R)^2 \norm{K}^2_{\sup}\frac{1}{|V|}\sharp\{x\in V, \rho(x)\leq R\}.\end{multline}

All in all, using successively \eqref{e:estHS}, \eqref{e:estP1}, \eqref{e:est1} and \eqref{e:ergo},
\begin{multline}
Var(K)=Var\left(\frac1T\int_0^T \sum_{j\leq M}\frac{(it { \cL })^j K}{j!} dt\right)\\
\leq \left\|\frac1T\int_0^T \sum_{j\leq M}\frac{(it\cL)^j K}{j!} dt \right\|_{\cH}^2+4\tilde\tau(R)^2 \norm{K}^2_{\sup}\frac{1}{|V|}\sharp\{x\in V, \rho(x)\leq R\}\\
\leq \left\|\frac1T\int_0^Te^{it { \cL }} K dt \right\|^2_{ \cH  } +3 \eps\norm{K}_{{ \cH}}^2+4\tilde\tau(R)^2 \norm{K}^2_{\sup}\frac{1}{|V|}\sharp\{x\in V, \rho(x)\leq R\} \\
\leq 4\eps\norm{K}_{{ \cH}}^2 +4\tilde\tau(R)^2 \norm{K}^2_{\sup}\frac{1}{|V|}\sharp\{x\in V, \rho(x)\leq R\} .
\end{multline}

Note that for fixed $m$, $R$ is of order $\eps^{-7/2}$, in other words $\eps$ is of order $R^{-2/7}$. This gives an estimate of the form \eqref{e:part}  (with $\alpha=2/7$) and thus proves Theorem \ref{t:gen}.

\end{proof}

\begin{remark} If $K\in \cH_m^0$, we know that $K\in \cF^\perp=(Ker\cL)^\perp = \overline{Im\, \cL}$. If we had $K\in Im \,\cL$ we would have the equality $Var (K)=0$. We can summarize this proof by saying that
we approximate $K$ by explicit elements of $Im \, \cL$, namely
$$K_T= K -\frac1T\int_0^Te^{it { \cL }} K dt = \frac1T\int_0^T (I- e^{it { \cL }}) K\in Im \, \cL,$$
for which we can prove that $Var(K-K_T)$ is small, and thus $Var(K)$ is close to $Var(K_T)=0$.
\end{remark}


\section{Proof 4, using the non-backtracking random walk \label{s:nbt}}
In this proof we work with the eigenfunctions of the non-backtracking random walk. In \S \ref{s:spnbt}
 we recall how these are related to the eigenfunctions of the laplacian. In \S \ref{s:nbtqv} we define a new quantum variance, defined in terms of the non-backtracking eigenfunctions; we show that it vanishes asymptotically, under assumptions (EXP) and (BST).
Finally we transfer the result to a result about the eigenfunctions of the laplacian, proving Theorem \ref{t:main}.
 
  \subsection{Eigenfunctions of the non-backtracking operator\label{s:spnbt}}
  The facts given in this section are true for all regular graphs, without the need of Assumptions (EXP) and (BST).
  
  Recall that $B({\mathfrak{X}})$ is the set of oriented edges of ${\mathfrak{X}}$.
 If $e$ is an element of $B({\mathfrak{X}})$, we shall denote by $o(e)\in {\mathfrak{X}}$ its origin, $t(e)\in {\mathfrak{X}}$ its terminus, and $\hat e\in B({\mathfrak{X}})$ the reversed bond.
Let $\cA^\sharp$ be the positive matrix indexed by $B({\mathfrak{X}})\times B({\mathfrak{X}})$, defined by
$$\cA^\sharp(e, e')=1$$
if $o(e')=t(e)$ and $e'\not=\hat e$; and $\cA^\sharp(e, e')=0$ otherwise. Remark that $\cA^\sharp/q$ is bistochastic. If we define $\iota : \ell^2(B({\mathfrak{X}}))\To \ell^2(B({\mathfrak{X}}))$ by $\iota f(e)=f(\hat e)$, we have an explicit conjugation between $\cA^\sharp$ and its adjoint~: $\cA^{\sharp *}=\iota  \cA^\sharp \iota$.

Denote by $B=\Gamma\backslash B({\mathfrak{X}})$ the set of directed bonds of $G=\Gamma\backslash {\mathfrak{X}}$. Note that $B$ has cardinality $|V|(q+1)$ (recall that $V$ is the set of vertices of $G$). We can let $\cA^\sharp$ act on the space of $\Gamma$-invariant functions on $B({\mathfrak{X}})$; hence on $\ell^2(B)$. The spectral radius
of $\cA^\sharp$ is $q$.
There is an explicit relation (described for instance in \cite{ALM}, where $\cA^\sharp/q$ is denoted $M^\sharp$) between the spectrum of $\cA^\sharp$ and the spectrum of the original adjacency operator $\cA$ on $\ell^2(V)$~:

\begin{itemize}
\item[(o)] the matrix $\cA^\sharp$ has $q$ in its spectrum, corresponding to the constant eigenfunction. The matrix $\cA^\sharp$ has $-q$ in its spectrum iff $\cA$ has $-(q+1)$ in its spectrum, iff the graph $G$ is bi-partite.

\item[(i)] each eigenvalue $\lambda\not=\pm (q+1)$ of $\cA$, written as $\lambda=q^{1/2 +is}+ q^{1/2 -is}$ with $s\in \IR\cup i\IR$, gives rise to the two eigenvalues 
$$\frac{2q}{\left(\lambda\pm \sqrt{\lambda^2-4q}\right)}=q^{1/2\pm is}$$
of $\cA^\sharp$;

\item[(ii)] in addition, $\cA^\sharp$ admits the eigenvalue $ 1$ with multiplicity $b\defeq|E|-|V|+1$ (the rank of the fundamental group of $G$); and the eigenvalue $ -1$ with multiplicity $b-1 $ if $-(q+1)$ is not an eigenvalue of $\cA$, or $b $ if $-(q+1)$ is an eigenvalue of $\cA$ (equivalently, if the graph is bi-partite).
\end{itemize}

In particular, the eigenvalue $q$ of $\cA^\sharp$ has multiplicity $1$. The tempered spectrum of $\cA$ corresponds to eigenvalues of $\cA^\sharp$ of modulus $\sqrt{q}$; the untempered spectrum of $\cA$ contained in $(q+1)[-1, 1-\beta]$ gives rise to real eigenvalues of $\cA^\sharp$ contained in $q[-1, 1-\beta']$
with
\begin{equation}\label{e:beta'}1-\beta'= \frac{2}{(q+1)\left(1-\beta- \sqrt{(1-\beta)^2-\frac{4q}{(q+1)^2}}\right)}.\end{equation}

\bigskip
This statement can be made more precise by relating the eigenvectors of the two operators.
The eigenvectors of $\cA^\sharp$ are related to those of $\cA$ as follows:

\begin{itemize}
\item[(o), (i)] an eigenfunction $\phi$ of $\cA$ for the eigenvalue $\lambda\not=\pm 1$ gives rise to the two eigenfunctions of $\cA^{\sharp}$,
$$f_1(e)=\phi(t(e))-\eps_1\phi(o(e));\qquad f_2(e)=\phi(t(e))-\eps_2\phi(o(e)),$$
where $\eps_1, \eps_2$ are the two roots of $q\eps^2-\lambda\eps+1=0$. If $\lambda=q^{1/2 +is}+ q^{1/2 -is}$ then $\eps_1, \eps_2=q^{-1/2\pm is}$. The two eigenvalues of $\cA^{\sharp}$ are $q\eps_1, q\eps_2= q^{1/2\pm is}$. In the particular case of the eigenvalues $\lambda=\pm 2\sqrt{q}$ of $\cA$, we get Jordan blocks of size $2$ for $\cA^{\sharp}$.

Note that these functions (together with the trivial eigenfunctions, i.e. constant and alternate) generate the subspace $F$ of $\ell^2(B)$ generated by functions that depend only on the origin and functions that depend only on the terminus. This space is $2|V|-1$-dimensional in the non-bipartite case, and $2|V|-2$-dimensional in the bipartite case.
\item[(ii)] the eigenvalues $\pm 1$ of $\cA^\sharp$ correspond, respectively, to odd and even\footnote{Odd means $f(\hat e)=-f(e)$ and even means $f(\hat e)=f(e)$, for every bond $e$.} solutions of $\sum_{o(e)=x} f(e)=0$ (for every vertex $x$).
 \end{itemize}
 
 \bigskip
The eigenfunctions of the family (ii) are automatically orthogonal to those of the family (o), (i). 
In (i), eigenfunctions of $\cA^\sharp$ stemming from different eigenvalues $\lambda$ of $\cA$ are orthogonal; the two eigenfunctions $f_1, f_2$ stemming from the same $\lambda$ are {\em {not}} orthogonal. However,  in an orthonormal basis $(f_1, f'_2)$ of $Vect(f_1, f_2)$, the matrix of $\cA^\sharp$ is triangular, with $q\eps_1, q\eps_2$ on the diagonal, and in the upper corner, a complex number which is an explicit continuous function of $q, \eps_1, \eps_2$, and thus is bounded (independently of the graph).

\subsection{Notational point\label{s:issue}}
To complement the definitions introduced in Section \ref{s:notation}, for $k\geq 1$, any element $K\in \cH_k$ can now be used to define an operator $\widehat{K}_{B({\mathfrak{X}})}$ on $\ell^2(B({\mathfrak{X}}))$, with kernel
\begin{equation}\label{e:defK}\la \delta_{b_1}, \widehat{K}_{B({\mathfrak{X}})}\delta_{b_2}\ra_{\ell^2(B({\mathfrak{X}}))}\defeq K(o(b_1), t(b_2))\end{equation}
if $b_1, b_2\in B({\mathfrak{X}})$ are such that $\cA^{\sharp (k-1)}(b_1, b_2)\not=0$ and $\la \delta_{b_1}, \widehat{K}_{B({\mathfrak{X}})}\delta_{b_2}\ra_{\ell^2(B({\mathfrak{X}}))}= 0$ otherwise. It also defines an operator $\widehat{K}_B$ on $\ell^2(B)$, with kernel
\begin{equation}\label{e:defKB}\la \delta_{b_1}, \widehat{K}_{B}\delta_{b_2}\ra_{\ell^2(V)}=\sum_{\gamma\in\Gamma}K(o(\tilde b_1), \gamma\cdot  t(\tilde b_2))\defeq K_B(b_1, b_2)\end{equation}
if $b_1, b_2\in B= \Gamma\backslash B({\mathfrak{X}})$ and $\tilde b_1,  \tilde b_2 \in B({\mathfrak{X}})$ are lifts of $b_1, b_2$.
By linearity, these definitions extend to $K\in \cH_{\leq k}$.

Recall that we see an element $K\in \cH_k$ as a function on $B_k$,  the set of non-backtracking paths of length $k$ on $G$. Remember that we have introduced the action of the transfer operator $\cS$ in the proof of Proposition \ref{p:Ck}.
If $K\in \cH_k$ then $\cS K\in \cH_k$. So it defines operators $\widehat{\cS K}$ on $\ell^2({\mathfrak{X}})$, $\widehat{\cS K}_{G}$ on $\ell^2(V)$, as well as $\widehat{\cS K}_{B({\mathfrak{X}})}$ on $\ell^2(B({\mathfrak{X}}))$, and $\widehat{\cS K}_{B}$ on $\ell^2(B)$, according to \eqref{e:defK}, \eqref{e:defKB}.

Note that $\cH_1$ coincides with $\ell^2(B)$ (up to the factor $|V|$ in the definition of the norm).
With this identification, the action of $\cS$ on $\cH_1$ coincides with the action of $q^{-1}\cA^{\sharp *}$ on $\ell^2(B)$.  We want to caution here that the operator $\widehat{\cS K}_{B}= q^{-1}\widehat{\cA^{\sharp *} K}_{B}$ should not be mistaken for the composition $
q^{-1}\cA^{\sharp *} \circ \widehat {K}_{B}$ of two operators acting on $\ell^2(B)$.

On the other hand, we will work with $\cA^\sharp \widehat{K}_{B}$ , that is the composition $\cA^\sharp \circ \widehat{K}_{B }$. It coincides with $\widehat{\sigma K}_{B }$, where $\sigma :\cH \To \cH$ is the ``left-shift'' defined by
\begin{equation}\label{e:lshift}\sigma K (x,y)= K(x', y)\end{equation}
if $(x;y)=(x, x', \ldots, y)$.  
Similarly,  $ \widehat{K}_{B}\cA^\sharp=  \widehat{K}_{B}\circ \cA^\sharp$ coincides with $\widehat{\rho K}_{B }$, where $\rho :\cH \To \cH$ is the ``right-shift'' defined by
\begin{equation}\label{e:rshift}\rho K (x,y)= K(x, y')\end{equation}
if $(x;y)=(x, x', \ldots, y', y)$. Note that both $\sigma$ and $\rho$ map $\cH_k$ to $\cH_{k+1}$, and that they commute.

 
\subsection{The non-backtracking quantum variance\label{s:nbtqv}}Now fix $(\psi_j)$ an orthonormal family of $\ell^2(V)$ consisting of {\em{tempered}} eigenfunctions of $\cA$, for the eigenvalues $\lambda_j \in (-2\sqrt{q}, 2\sqrt{q})$ (to simplify the discussion we simply ignore the untempered spectrum, this is a harmless simplification if (BST) is satisfied);
define
$$f_j(e)= \psi_j(t(e))-\eps(j)\psi_j(o(e))$$
where $\eps(j)$ is a root of $q\eps^2-\lambda_j\eps+1=0$. If $\lambda_j=q^{1/2+ is_j}+ q^{1/2-is_j}$, then $\eps(j)=q^{-1/2\pm is_j}$; To fix the notation, say $\eps(j)=q^{-1/2-is_j}$. Then $f_j$ is an eigenfunction of $\cA^\sharp$
for the eigenvalue
$\mu_j= q^{1/2+is_j}$. Now define
$$f^*_j(e)=\iota f_j(e)= \psi_j(o(e))-{\eps(j)}\psi_j(t(e))$$
(recall that $\iota$ consists in reversing the edges).
Now $f^*_j$ is an eigenfunction of $\cA^{\sharp *}$
for the eigenvalue $\mu_j$. As soon as $\mu_j\not = \overline{\mu_j}$, i.e. $\lambda_j\not\in \left\{ \pm 2\sqrt{q}\right\}$, then $f^*_j$ and $f_j$ are automatically orthogonal.

For $K\in \cH_{\leq m}$ ($m\geq 1$ fixed) we now introduce the ``non-backtracking quantum variance''
$$Var_{nb}(K)=\frac{1}{|V|}\sum_{j} |\la f_j^*, \widehat{K}_B f_j\ra|^2.$$
Note that it has the funny property that
$Var_{nb}(\bbbone)=0$
(usually we expect quantum variances to satisfy $Var(\bbbone)=1$). Here we will prove that $Var_{nb}(K)$ tends to zero when the size of the graph increases, {\em{for all $K$ without any restriction}}. We shall see later what it implies for the usual quantum variance (for which, of course, the best we can hope for is to prove that it decays for all $K\in\cH_m^0$).

From the previous proofs, we recall that the general philosophy is to use the eigenfunction property to note, trivially, that the quantum variance vanishes on a certain subspace of operators; then to transform this trivial statement into a less trivial one, using the assumptions (BST) and (EXP).
Here, the eigenfunction property implies that 
\begin{equation}\label{e:invnbt}\la f_j^*, (\bar\eps_j\cA^{\sharp} \widehat{K}_B- \widehat{K}_B\cA^{\sharp}\eps_j) f_j\ra =0
\end{equation} for all $j$. Since $\eps_j$ does depend on $j$, we see that
this identity depends on $j$, so that we cannot convert this into an identity for $Var_{nb}$.
To solve that problem, we restrict the definition of the variance
to eigenvalues in a small interval $I$, where $\eps_j$ may be regarded as almost constant.

To that end, let us fix a value $E_0\in (-2\sqrt{q}, 2\sqrt{q})$, and $I=(E_0-\delta, E_0+\delta)$ where $\delta$ is small enough so that $I\subset (-2\sqrt{q}, 2\sqrt{q})$ ($\delta$ is fixed, but can be taken arbitrarily small). When needed, we shall write $E_0=q^{1/2+is_0}+ q^{1/2-is_0}$ and $\eps_0=q^{-1/2-is_0}$ with, say, $\Im m \eps_0\leq 0$.
Define the local quantum variance~:
$$Var^I_{nb}(K)=\frac{1}{N(I)}\sum_{j, \lambda_j\in I} |\la f_j^*,  \widehat{K}_B f_j\ra|^2$$
where $N(I)=|\{ j, \lambda_j\in I\}|$.

\begin{remark} Note that we might also want to extend this definition to families of operators, $\lambda\in\IR\mapsto K_\lambda\in \cH_{\leq m}$
depending on $\lambda$ in a $C^1$ (or even $C^0$) fashion. The quantum variance is then defined as
$$Var^I_{nb}(K)=\frac{1}{N(I)}\sum_{j, \lambda_j\in I} |\la f_j^*,  \widehat{K}_{\lambda_j,B} f_j\ra|^2.$$
The proof below extends to this case without any additional difficulty.
\end{remark}

Because $\eps_j$ is a $C^1$ function of $\lambda_j$, we have $|\eps_j - \eps_0|\leq C\delta$ for $|\lambda_j -E_0|\leq \delta$. Equation \eqref{e:invnbt} implies
\begin{eqnarray*}Var^I_{nb}(\bar\eps_0 \sigma K-\rho K \eps_0)
&\leq &\norm{\widehat{K}_B}_{\ell^2(B)\To \ell^2(B)}^2\,O( \delta)\\
&\leq & c(m)^2 \norm{K}_{sup}^2 O( \delta)
\end{eqnarray*}
where $\sigma, \rho$ are the left- and right-shifts introduced in \S \ref{s:issue}.
More generally, for any $n$, 
$$Var_{nb}(K)=Var_{nb}\left(\frac{1}{n+1}\sum_{k=0}^n  \bar\eps_0^{n-k}   \eps_0^{k}\sigma^{(n-k)}\rho^k K \right) +nc(m)^2 \norm{K}_{sup}^2 O( \delta).$$

Let us assume that $K \in \cH_m$ -- then $\sigma^{(n-k)}\rho^k K  \in  \cH_{m+n}$, and by Proposition \ref{p:crucialthing} below, there is a constant $C$ independent of $n, m$ and of $I$ such that
\begin{multline*} Var^I_{nb}\left(\frac{1}{n+1}\sum_{k=0}^n \bar\eps_0^{n-k}   \eps_0^{k}\sigma^{(n-k)}\rho^k K\right)\leq  C \left\Vert \frac{1}{n+1}\sum_{k=0}^n\bar\eps_0^{n-k}   \eps_0^{k}\sigma^{(n-k)}\rho^k K\right\Vert^2_{\cH} \\ + \norm{K}_{sup}^2\, o_{n, m, \delta}(1)_{G\To\mathfrak X }.
\end{multline*}
The remainder $o_{n, m, \delta}(1)_{G\To\mathfrak X }$ may be expressed in terms of the number of small loops, but we prefer to keep it in this non-explicit form.

To conclude, we need to analyse the middle term, and we note the two identities~:
\begin{lemma}\label{l:lem}
$$  \norm{\bar\eps_0^{n-k}   \eps_0^{k}\sigma^{(n-k)}\rho^k K}^2_{\cH}= q^{-n}\norm{\sigma^{(n-k)}\rho^k K}^2_{\cH}=\norm{K}^2_{\cH},$$
and more generally for $k\geq k'$,
$$ \la\bar\eps_0^{n-k}   \eps_0^{k}\sigma^{(n-k)}\rho^k K,\bar\eps_0^{n-k'}   \eps_0^{k'}\sigma^{(n-k')}\rho^{k'} K\ra_{\cH} =
  q^{2is_0(k-k')}\la  \cS^{k-k'}K  ,  K   \ra_{\cH}.$$
  \end{lemma}
  \begin{proof}
 There are obtained by direct calculation of the scalar products~:
 \begin{multline*}
  \la\bar\eps_0^{n-k}   \eps_0^{k}\sigma^{(n-k)}\rho^k K,\bar\eps_0^{n-k'}   \eps_0^{k'}\sigma^{(n-k')}\rho^{k'} K\ra_{\cH}=  q^{-n}q^{2is_0(k-k')}
   \la \sigma^{(n-k)}\rho^k K, \sigma^{(n-k')}\rho^{k'} K\ra_{\cH}\\
   =  q^{-n}q^{2is_0(k-k')}\frac{1}{|V|}\sum_{(x_0, x_1, \ldots, x_{m+n})\in B_{m+n}} \overline{K(x_{n-k},x_{m+n-k})}K(x_{n-k'},x_{m+n-k'})\\
   =q^{-(k-k')}q^{2is_0(k-k')}\frac{1}{|V|}\sum_{(y_{k'-k},y_{k'-k+1}\ldots, y_0, \ldots, y_m)\in B_{m+k-k'}} \overline{K(y_{0},y_m)}K(y_{k'-k},y_{k'-k+m})\\
   =q^{2is_0(k-k')}\la  \cS^{k-k'}K  ,  K   \ra_{\cH}.
  \end{multline*}
 \end{proof}

The action of $\cS$ on $\cH_1$ coincides with that of $q^{-1}\cA^{\sharp *}=q^{-1}\iota\cA^{\sharp}\iota$, and thus its spectrum
and eigenfunctions were entirely described in \S \ref{s:spnbt}. In particular the spectrum of $\cS$ on $\cH_1$ is the union of a real spectrum and a spectrum contained in the disc of radius $q^{-1/2}$.
The operator $\cS$ is not diagonalizable in an orthonormal basis, but there is an orthonormal basis where $\cS$ is represented by triangular blocks on the diagonal of size at most 2 , with the eigenvalues $(\eps_i, q \eps_i^{-1})$ (in addition there is an orthonormal family of eigenvectors associated to $\pm 1/q$). The $\eps_i$ are either real, contained in $\{1\}\cup [-1+\beta', 1-\beta']$, or have modulus $q^{1/2}$.

Recall also that $\cS^m$ sends $\cH_m$ to $j_{1, m}\cH_1$\footnote{The injection $j_{m', m}$ of $\cH_{m'}$ into $\cH_{m}$ ($m'\leq m$)
was defined in \S \ref{s:ops}.}
and has norm $1$.
We have already used this in \eqref{e:sg}, to see that for $m>1$ and $k\geq m$,
 \begin{multline}\label{e:sg'}\norm{\cS^k}_{\cH^o_m\To\cH^o_m}\leq \norm{\cS^k}_{\cH^o_m\To j_{1, m}(\cH^o_1)}\norm{\cS^{k-m}}_{\cH^o_1\To\cH^o_1}\\
 \leq \norm{\cS^{k-m}}_{\cH^o_1\To\cH^o_1}
 \leq ((k-m)+1)(1-\beta')^{k-m}.
 \end{multline}
 The factor $(1-\beta')^{k-m}$ comes from the bound on the eigenvalues of $\cS$, and the polynomial correction $(k-m)+1$
 comes from the fact that $\cS$ is not diagonalizable in an orthonormal basis, but is represented by triangular blocks of size at most 2 as described above.

 Actually, the norm of $\cS^k$ restricted to the spectrum in the disc of radius $q^{-1/2}$
 is $\leq ((k-m)+1)q^{-1/2(k-m)}$. On the other hand, on the real spectrum, the norm of $\sum_{n'\leq k\leq n}q^{2is_0k}\cS^{k}$ is less than
 $\sup_{t\in  \{1\}\cup [-1+\beta', 1-\beta'] }|\sum_{n'\leq k\leq n} t^k q^{2is_0k}|$, for any $n'\leq n$. 
 
If $q^{ 2is_0}\not= \pm 1$, this can be bounded above by  $\frac{C}{|\Im m q^{2is_0}|}\leq  \frac{C}{|\sin(2s_0\ln q)|}$ for any $n'\leq n$. We see that for all
$K\in  \cH_m$,
\begin{multline*} \left\Vert\frac{1}{n+1}\sum_{k=0}^n \bar\eps_0^{n-k}   \eps_0^{k}\sigma^{(n-k)}\rho^k K \right\Vert^2_{\cH}\\
\leq \frac{2}{(n+1)^2}\sum_{ k'\leq k< k'+m}q^{2is_0(k-k')}\la  \cS^{k-k'}K  ,  K   \ra_{\cH} + \frac{2}{(n+1)^2}\sum_{0\leq k'+m\leq k\leq n} q^{2is_0(k-k')}\la  \cS^{k-k'}K  ,  K   \ra_{\cH}
\\ \leq\frac{C}{(n+1)^2}\sum_{ k'\leq k< k'+m}  \norm{K}_{\cH}^2 +
\frac{C}{(n+1)^2}\sum_{0\leq k'+m\leq k\leq n}  (1+(k-k'-m)) q^{-1/2(k-k'-m)} \norm{K}_{\cH}^2 
 \\+\frac{C}{(n+1)| \sin(2s_0\ln q)|}  \norm{K}_{\cH}^2
\\ \leq \frac{C}{n+1} \left(1 + \frac{1}{|\sin(2s_0\ln q)|}\right)\norm{K}_{\cH}^2.\end{multline*}

If 
$q^{2is_0}= -1$, we see that $\sup_{t\in  \{1\}\cup [-1+\beta', 1-\beta'] }|\sum_{n'\leq k\leq n} t^k q^{2is_0k}|\leq C \beta^{' -1}$ for any $n'\leq n$.
We now find
$$\left\Vert\frac{1}{n+1}\sum_{k=0}^n \bar\eps_0^{n-k}   \eps_0^{k}\sigma^{(n-k)}\rho^k K\right\Vert^2_{\cH}\\
\leq \frac{C}{n+1} \left(1 + \frac{1}{\beta'}\right)\norm{K}_{\cH}^2.$$

Note that the case $q^{2is_0}= 1$ is
of no interest to us, because it would correspond to $E_0=\pm 2\sqrt{q}$, at the edge of the spectrum.

 Finally, we obtain, for all $K\in \cH_m$, and $q^{2is_0}$ away from $1$,
 an estimate of the form
 \begin{multline}Var^I_{nb}(K )\leq \frac{\tilde C(m, s_0, \beta)}{ (n+1)}\norm{K}_{\cH}^2 +\norm{K}_{sup}^2 \, (o_{n, m, \delta}(1)_{G\To\mathfrak X } + n  c(m)^2\delta)\label{e:estvarnb}
\end{multline}

Before proving the technical Proposition \ref{p:crucialthing}, let us clarify what \eqref{e:estvarnb} implies for the full variance $Var_{nb}(K )$. For $\varepsilon >0$ let $J_\varepsilon=(-2\sqrt{q}+\varepsilon, 2\sqrt{q}-\varepsilon)$. We have
\begin{eqnarray*}|Var_{nb}(K ) -Var^{J_\varepsilon}_{nb}(K )|& \leq& 2\norm{\widehat{K}_B}_{\ell^2(V)\To \ell^2(V)}^2\frac{|\{j, \lambda_j\not\in J_\varepsilon\} |}{|V|}\\
& \leq & 2\norm{\widehat{K}_B}_{\ell^2(V)\To \ell^2(V)}^2 \left(\int_{J_\varepsilon^c}{\mathrm{m}}(\lambda)d\lambda + o_\varepsilon(1)_{G\To\mathfrak X }\right) \\
&=& 2\norm{\widehat{K}_B}_{\ell^2(V)\To \ell^2(V)}^2( O(\varepsilon) + o_\varepsilon(1)_{G\To\mathfrak X }).
\end{eqnarray*}
Then, we subdivide $J_\varepsilon$ into $\lfloor \delta^{-1} \rfloor$ disjoint intervals of size $|J_\varepsilon| / \lfloor \delta^{-1} \rfloor\simeq C\delta$~: $J_\varepsilon=\sqcup_{l=1}^{\lfloor \delta^{-1} \rfloor} I_l$. We see that $Var^{J_\varepsilon}_{nb}(K )$ is a convex combination of all the $Var^{I_l}_{nb}(K )$, for which the estimate \eqref{e:estvarnb} holds. We get the final estimate for  $Var_{nb}(K )$,
\begin{equation}
\label{e:estvarnbglob}
Var_{nb}(K )\leq \frac{\tilde C(m, J_{\varepsilon}, \beta)}{ (n+1)}\norm{K}_{\cH}^2 +\norm{K}_{sup}^2 \, (o_{n, m, \delta, \varepsilon}(1)_{G\To\mathfrak X } + n c(m)^2O(\delta)+ c(m)^2O(\varepsilon)) .
\end{equation}
where $\tilde C(m, J_{\varepsilon}, \beta)=\sup_{2\sqrt{q}\cos(s_0\ln q)\in J_\varepsilon}\tilde C(m, s_0, \beta)$.

There remains to prove the domination of the quantum variance by the Hilbert-Schmidt norm, which is easy for $Var_{nb}(K)$ but much less obvious for the local version $Var^I_{nb}(K)$.
\begin{prop}\label{p:crucialthing}Assume that $K\in \cH_m$. Then
\begin{equation} \label{e:}
Var^I_{nb}(K) \leq C\norm{ K}^2_{\cH} + o_{m, I}(1)_{G\To \mathfrak{X}}\norm{K}_{sup}^2
\end{equation}
where $C$ is some universal constant.
\end{prop}
It will be important for us in the sequel that $C$ does not depend on $m$ nor $I$; whereas the term $o(1)$ is allowed to depend on both. Remember, $o_{m, I}(1)_{G\To \mathfrak{X}}$
means a quantity that goes to $0$ when $G$ converges to $\mathfrak{X}$ (in the sense of Benjamini-Schramm), $m$ and $I$ being fixed.

\begin{proof}  Introduce a continuous function $\chi$ such that $\chi\equiv 1$ on $I$, $\chi\equiv 0$ outside $I_\delta\defeq I+[-\delta, \delta]$ and $|\chi|\leq 1$. Then
$$Var^I_{nb}(K)\leq \frac{1}{N(I)}\sum_{j=1}^{|V|}  \chi(\lambda_j)^2 |\la f_j^*, \widehat{K}_B f_j\ra|^2.$$

(a) We will need the Kesten-McKay law, saying that 
$$N(I)\sim |V| \int_I {\mathrm{m}}(\lambda)d\lambda$$
when $G$ approaches the infinite $(q+1)$-regular tree in the Benjamini-Schramm topology.
More generally, for any polynomial $Q$ of degree $d$, and $K\in \cH_m$, 
the statement is that
\begin{multline*}\left|\sum_{j}\la \psi_j, Q(\cA) \widehat{K}_G \psi_j\ra_{\ell^2(V)} -  \int \la P_{\lambda,\omega}, \bbbone_{\cD}Q(\cA) \widehat{K}  P_{\lambda,\omega}\ra_{\ell^2(\mathfrak{X})} d\nu(\omega){\mathrm{m}}(\lambda)d\lambda  \right| \\
\leq   \sup_{x, y} |Q(\cA)K(x, y)| \tilde\tau(d+m) \sharp\{x, \rho(x)\leq d+m\}.
\end{multline*}
This is obtained by writing 
\begin{eqnarray*}\sum_j \la  \psi_j, Q(\cA) K_G \psi_j\ra &= &\Tr_{\ell^2(V)} Q(\cA)\widehat{K}_G\\
&=& \sum_{x\in V }[Q(\cA) {K}]_G (x,x)\\
&=&\sum_{\tilde x\in \cD, \rho(x) > d+m }[Q(\cA) {K}] (\tilde x,\tilde x) + O(\sup_{x} |[Q(\cA) {K}]_G(x, x)|)\sharp\{x, \rho(x)\leq d+m\})\\
&=& \Tr_{\ell^2(\mathfrak X)} \bbbone_{\cD}Q(\cA)  \widehat{K} +O(\sup_{x, y} |[Q(\cA)K](x, y)| \tilde\tau(d+m) \sharp\{x, \rho(x)\leq d+m\}),
\end{eqnarray*}
and using \eqref{e:trace}.

Now, for any continuous function $\chi$, for any polynomial $Q$ of degree $d$, and $K\in \cH_m$, 
we have
\begin{multline*}\left|\sum_{j}\la \psi_j, \chi(\cA) \widehat{K}_G \psi_j\ra_{\ell^2(V)} -  \int \la P_{\lambda,\omega}, \bbbone_{\cD}\chi(\cA) \widehat{K}  P_{\lambda,\omega}\ra_{\ell^2(\mathfrak{X})} d\nu(\omega){\mathrm{m}}(\lambda)d\lambda  \right| \\
\leq  \norm{Q-\chi}_{L^\infty(-(q+1), q+1))}\left(|V| \norm{\widehat{K}_G}_{\ell^2(V)\To \ell^2(V)}+  \int |\la P_{\lambda,\omega}, \bbbone_{\cD}  \widehat{K} P_{\lambda,\omega}\ra| d\nu(\omega){\mathrm{m}}(\lambda)d\lambda  \right)\\+ \sup_{x, y} |[Q(\cA)K](x, y)| \tilde\tau(d+m) \sharp\{x, \rho(x)\leq d+m\}.
\end{multline*}
This is obtained from the previous case by approximating $\chi(\cA)$ by $Q(\cA)$.

By similar arguments, if we have two elements $K_1\in \cH_{m_1}$ and $K_2\in \cH_{m_2}$, 
\begin{multline*}\left|\sum_{j}\la \psi_j, \chi(\cA) \widehat{K}^*_{1G} \widehat{K}_{2G} \psi_j\ra_{\ell^2(V)} -  \int \la P_{\lambda,\omega}, \chi(\cA) \widehat{K_1}^*\bbbone_{\cD}  \widehat{K_2} P_{\lambda,\omega}\ra d\nu(\omega){\mathrm{m}}(\lambda)d\lambda  \right| \\
\leq  \norm{Q-\chi}_{L^\infty(-(q+1), q+1))}\Bigg(|V|   \norm{\widehat{K}_{1G}}_{\ell^2(V)\To \ell^2(V)} \norm{\widehat{K}_{2G}}_{\ell^2(V)\To \ell^2(V)}\\+  \int \norm{ \bbbone_{\cD}   \widehat{K_1} P_{\lambda,\omega}}_{\ell^2( \mathfrak X)}  \norm{ \bbbone_{\cD} \widehat{K_2} P_{\lambda,\omega}}_{\ell^2( \mathfrak X)}d\nu(\omega){\mathrm{m}}(\lambda)d\lambda  \Bigg)\\+\sup |(Q(\cA)K_1^*K_2)(x, y)| \tilde\tau(d+m_1+m_2) \sharp\{x, \rho(x)\leq d+m_1+m_2\}
\end{multline*}
Using the Cauchy-Schwarz inequality and \eqref{e:HSX}, this is
\begin{multline*} \leq 
\norm{Q-\chi}_{L^\infty(-(q+1), q+1))}\Bigg(|V|  \norm{\widehat{K_1}_G}_{\ell^2(V)\To \ell^2(V)}\norm{\widehat{K_2}_G}_{\ell^2(V)\To \ell^2(V)}\\
+ \norm{\bbbone_\cD\widehat{K_1}}_{HS(\ell^2(\mathfrak X))}\norm{\bbbone_\cD\widehat{K_2}}_{HS(\ell^2(\mathfrak X))}\Bigg)\\+\sup |(Q(\cA)K_1^*K_2)(x, y)|\tilde \tau(d+m_1+m_2) \sharp\{x, \rho(x)\leq d+m_1+m_2\}
\\ \leq  \norm{Q-\chi}_{L^\infty(-(q+1), q+1))}\Bigg(|V|  \norm{\widehat{K_1}_G}_{\ell^2(V)\To \ell^2(V)}\norm{\widehat{K_2}_G}_{\ell^2(V)\To \ell^2(V)}\\
+ ( \sum_{x\in \cD, y\in \mathfrak{X}} |K_1(x, y)|^2)^{1/2}( \sum_{x\in \cD, y\in \mathfrak{X}} |K_2(x, y)|^2)^{1/2}\Bigg)\\+\sup |(Q(\cA)K_1^*K_2)(x, y)| \tilde\tau(d+m_1+m_2) \sharp\{x, \rho(x)\leq d+m_1+m_2\} 
\\ \leq c(m_1)c(m_2)\norm{Q-\chi}_{L^\infty(-(q+1), q+1))} |V| \norm{K_1}_{sup}\norm{K_2}_{sup}\\
+\sup |(Q(\cA)K_1^*K_2)(x, y)| \tilde\tau(d+m_1+m_2) \sharp\{x, \rho(x)\leq d+m_1+m_2\} .
\end{multline*}
Choosing a sequence of polynomial $(Q_n)$ (of degree $d_n$) that approximates $\chi$ in $L^\infty(-(q+1), q+1))$, and noting that for any given $Q_n$ 
we have $$\sup |Q_n(\cA)K_1^*K_2(x, y)| \tilde\tau(d_n+m_1+m_2) \sharp\{x, \rho(x)\leq d_n+m_1+m_2\} =o_{Q_n, K_1, K_2}(|V|)_{G\To\mathfrak X },$$ we see that
\begin{multline}\label{e:approxhs}\left|\sum_{j}\la \psi_j, \chi(\cA) \widehat{K}^*_{1G} \widehat{K}_{2G} \psi_j\ra_{\ell^2(V)} -  \int \la P_{\lambda,\omega}, \chi(\cA) \widehat{K_1}^*\bbbone_{\cD}  \widehat{K_2} P_{\lambda,\omega}\ra d\nu(\omega){\mathrm{m}}(\lambda)d\lambda  \right| \\ =o_{\chi, m_1, m_2}(|V|)_{G\To\mathfrak X } \norm{K_1}_{sup}\norm{K_2}_{sup}.
\end{multline}
The precise rate $o(|V|)$ depends both on how fast $\chi$ can be approximated by polynomials (the regularity of $\chi$) and on the number of small loops, so we don't try to make this more explicit.

\bigskip

(b) Choose a continuous branch for the map $$\Theta : \lambda=\sqrt{q} (q^{is}+ q^{-is})\mapsto \eps= q^{-1/2-is}$$ (this map is actually analytic for $\lambda\not= \pm 2\sqrt{q}$). To fix ideas, let us say that $\Im m \Theta \leq 0$.

We introduce the operator $\widetilde U:\ell^2(V(\mathfrak{X}))\To \ell^2(B(\mathfrak{X}))$ defined by
$$\widetilde U\psi (e)=\psi(t(e))-(\Theta(\cA)\psi )(o(e)).$$ We can define similarly $U:\ell^2(V )\To \ell^2(B)$; note that $U$ is built so that
$ f_j =U\psi_j$. Obviously, $U$ is bounded independently of the graph $G$ (in fact, we will only use the fact that $\norm{U\psi_j}_{\ell^2(V)}$ is bounded).  
\begin{multline*}Var^I_{nb}(K)\leq \frac{1}{N(I)}\sum_{j=1}^{|V|}    |\la f_j^*, \widehat{K}_B U \chi(\cA)\psi_j\ra|^2  
\\ \leq C  \frac{1}{N(I)}\sum_{j=1}^{|V|}    \norm{\widehat{K}_B U\chi(\cA)\psi_j}_{\ell^2(B)}^2  \\
= \frac{C}{N(I)}\sum_{j=1}^{|V|}    \la \chi(\cA)\psi_j , U \widehat{K}_B ^* \widehat{K}_B U\chi(\cA)\psi_j\ra_{\ell^2(V)} 
= \frac{C}{N(I)}\sum_{j=1}^{|V|}    \la \chi^2(\cA)\psi_j , U \widehat{K}_B ^* \widehat{K}_B U \psi_j\ra_{\ell^2(V)} 
\end{multline*}
From the explicit expression of $U$, we see that $\la \chi^2(\cA)\psi_j , U \widehat{K}_B ^* \widehat{K}_B U \psi_j\ra_{\ell^2(V)}$ can be written as a linear combination of four terms
of the form $\la \psi_j, \chi'(\cA) \widehat{K}^*_{1G} \widehat{K}_{2G} \psi_j\ra_{\ell^2(V)}$, for some explicit functions $\chi'\in \{\chi^2, \chi^2 \Theta, \chi^2 \Theta^2\}$ and $K_1, K_2\in  \cH_{\leq 2m}$ if $K\in \cH_m$.

Applying \eqref{e:approxhs} to each of these terms, we obtain
\begin{multline}Var^I_{nb}(K)\leq
 \frac{C}{N(I)} \int \la \chi(\cA)P_{\lambda,\omega},  \widetilde U^* \widehat{K}_{B(\mathfrak X)}^* \bbbone_{\cD} \widehat{K}_{B(\mathfrak X)}  \widetilde U \chi(\cA) P_{\lambda,\omega}\ra d\nu(\omega){\mathrm{m}}(\lambda)d\lambda \\+  \norm{K}_{sup}^2\, o_{I, \chi, m}(1)_{G\To\mathfrak X }\\
= \frac{C}{N(I)} \int \chi^2(\lambda) \norm{ \bbbone_{\cD} \widehat{K}_{B(\mathfrak X)} \widetilde U P_{\lambda,\omega}}^2_{\ell^2(B(\mathfrak X))} d\nu(\omega){\mathrm{m}}(\lambda)d\lambda +\norm{K}_{sup}^2\, o_{I, \chi, m}(1)_{G\To\mathfrak X }  \label{e:bound}
 \end{multline}
 
It turns out that the calculation of $ \widetilde U P_{\lambda,\omega}$ is particularly simple. For $e\in B(\mathfrak X)$, we have $h_\omega(o(e))= h_\omega(t(e))\pm 1$. Say that $e$ goes towards $\omega$ if
$h_\omega(o(e))= h_\omega(t(e)) + 1$, and that $e$ goes away from $\omega$ if $h_\omega(o(e))= h_\omega(t(e)) - 1$. We see that $ \widetilde U P_{\lambda, \omega}(e) =0 $ if $e$ goes away from $\omega$, and that
$$ \widetilde U P_{\lambda, \omega}(e)= q^{-(1/2+is) h_\omega(t(e))}(1- q^{-1-2is})$$
if $e$ goes towards $\omega$. When we calculate
$$ \widehat{K}_{B(\mathfrak X)}(\widetilde U P_{\lambda, \omega})(e)=\sum_{e'\in B(\mathfrak X), \cA^{\sharp (m-1)}(e, e')\not= 0}K(o(e), t(e')) \widetilde U P_{\lambda, \omega}(e')$$
we see this gives $0$ if $e$ goes away from $\omega$. If $e$ goes towards $\omega$, only one term in the sum can be non-zero, namely the only $e'$ going towards $\omega$ and such that $\cA^{\sharp (m-1)}(e, e')\not= 0$. We denote this $e'$ by $e'_{e, \omega}$. In this case
$$|\widehat{K}_{B(\mathfrak X)}(\widetilde U P_{\lambda, \omega})(e)|^2 = |K(o(e), t(e'_{e, \omega}))|^2 q^{-h_\omega(t(e'_{e, \omega}))}  |1- q^{-1-2is}|^2.$$
Of course $|1- q^{-1-2is}|^2\leq 4$. To estimate $ \int |\widehat{K}_{B(\mathfrak X)}(\widetilde U P_{\lambda, \omega})(e)|^2 d\nu(\omega)$ we first note that this expression is independent of the choice of the origin $o$ used in the definition of $\nu$ and $P_{\lambda, \omega}$ (see \eqref{e:harmonic}); so that it is enough to treat the case where $o(e)=o$, where the calculations will be slightly simpler.
For any real $\lambda$ and any edge $e$,
\begin{multline*} \int |\widehat{K}_{B(\mathfrak X)}(\widetilde U P_{\lambda, \omega})(e)|^2 d\nu(\omega)\\
\leq 4 \sum_{e',  \cA^{\sharp (m-1)}(e, e')\not= 0}|K(o(e), t(e'))|^2  \int_{\{ \omega \mbox{ s.t. } e'_{e, \omega} =e' \}}q^{-h_\omega(t(e'))}  d\nu(\omega) .
\end{multline*}
In the case $o(e)=o$, the calculation is simple, it suffices to note that $ \nu\{ \omega \mbox{ s.t. } e'_{e, \omega} =e' \} =\frac{q^{-(m-1)}}{q+1}$, and that $h_\omega(t(e'))=-m$ if $e'_{e, \omega} =e' $.
Thus 
$$ \int |\widehat{K}_{B(\mathfrak X)}(\widetilde U P_{\lambda, \omega})(e)|^2 d\nu(\omega)\leq 4 \sum_{e'}|K(o(e), t(e'))|^2$$
and summing over $e$, 
$$ \int \norm{\bbbone_{\cD} \widehat{K}_{B(\mathfrak X)}(\widetilde U P_{\lambda, \omega}}_{\ell^2(B(\mathfrak X))}^2 d\nu(\omega)\leq C |V| \norm{K}^2_{\cH}$$
for all $\lambda$. Going back to \eqref{e:bound}, integrating over $\lambda$, and using the fact that $N(I)\sim |V| \int_I {\mathrm{m}}(\lambda) d\lambda$, we obtain
$$Var^I_{nb}(K) \leq
 C\frac{\int \chi^2(\lambda){\mathrm{m}}(\lambda)d\lambda }{\int_I {\mathrm{m}}(\lambda)d\lambda} \norm{K}^2_{\cH} + \norm{K}_{sup}^2 o_{I, \chi, m}(1)_{G\To\mathfrak X }.$$
 Note that $\frac{\int \chi^2(\lambda){\mathrm{m}}(\lambda)d\lambda }{\int_I {\mathrm{m}}(\lambda)d\lambda} $ is bounded independently of $I$ and $\chi$.

\end{proof}

\subsection{Back to the original quantum variance \label{s:back1}}
There remains to establish the link with our original eigenfunctions, those of $\cA$. 

\begin{lemma}  \label{l:formulas}
(i) Let $a\in \ell^2(V)=\cH_0$. If $K'\in \cH_1$ is defined by $K'(x, y)= a(x)$ for $y\sim x$, then
$$\la f_j^*, \widehat{K'}_B f_j\ra_{\ell^2(B)}=\overline{\eps_j} \la \psi_j, \widehat{(((q+1)-\cA) a)} \psi_j\ra_{\ell^2(X)}$$
where $\widehat{((q+1)-\cA) a}$ is the operator of multiplication by $((q+1)-\cA)  a$.

(ii) Let $K\in \cH_m$ ($m\geq 1$). We have the formula
$$\la f_j^*, \widehat{K}_B f_j\ra_{\ell^2(B)}=  \la \psi_j, \widehat{((1-\cS) K)}_G \psi_j\ra_{\ell^2(V)} + \eps_j \la \psi_j,  \widehat{\nabla^* K}_G \psi_j\ra_{\ell^2(X)}.$$
The notation $\widehat{((1-\cS) K)}_G$ is the one explained in \S \ref{s:issue}.
 \end{lemma}
 
  
 The proof is just a calculation.

 Lemma \ref{l:formulas} (i) implies that
 $$Var(((q+1)-\cA) a)  = q Var_{nb}(K')$$
 and thus  \eqref{e:estvarnb} implies that
 \begin{multline*}Var((I-(q+1)^{-1}\cA) a) \\
 \leq  \frac{q}{(q+1)} \frac{\tilde C(1, J_{\varepsilon}, \beta)}{ (n+1)}\norm{a}_{\cH}^2   +\norm{a}_{sup}^2 \, \left(o_{n, \delta, \varepsilon}(1)_{G\To\mathfrak X } + n c(1)^2O(\delta)+ c(1)^2O(\varepsilon)\right) .
\end{multline*}
 for $\delta, \varepsilon >0$ arbitrary.

By Lemma \ref{l:useful}, for all $a\in \cH_0^0$, for any $n$ and $T$,
 \begin{multline*}Var(  a) \leq \frac{C(0,\beta)^2 \tilde C(1, J_{\varepsilon}, \beta)}{n}\norm{a}^2_{\cH^0_0} 
+  2\frac{C(1, \beta)^2}{T^2}\norm{a}^2_{\cH_0^0}
\\+2T^2\norm{a}_{sup}^2 \, \left(o_{n, \delta, \varepsilon}(1)_{G\To\mathfrak X } + n c(1)^2O(\delta)+ c(1)^2O(\varepsilon)\right) .
\end{multline*}
We can take $\varepsilon =\delta=n^{-1}T^{-2-k}$ ($k$ arbitrary), and $n=n(T)$ large enough so that
$\frac{C(0,\beta)^2 \tilde C(1, J_{\varepsilon}, \beta)}{n}\norm{a}^2_{\cH^0_0} 
\leq 2\frac{C(1, \beta)^2}{T^2}\norm{a}^2_{\cH_0^0},$ we obtain
 \begin{equation*}Var(  a) \leq 4\frac{C(1, \beta)^2}{T^2}\norm{a}^2_{\cH_0^0}
+2\norm{a}_{sup}^2 \, \left(o_{ T}(1)_{G\To\mathfrak X } + c(1)^2O(T^{-k}))\right) .
\end{equation*}

 We can now proceed by induction on $m$~: suppose we have proved, for all $K\in \cH_{m-1}^0$, an estimate of the form
 \begin{equation}\label{e:finalfinal}Var(K)\leq \frac{\tilde C(m-1, \beta)}n \norm{K}^2_{\cH} + \norm{K}_{sup}^2 \, \, \left(o_{ n, m-1}(1)_{G\To\mathfrak X } + \tilde C(m-1)O(n^{-k})\right)\end{equation}
 valid for all $k$.

 Let now $K\in \cH_m$. Using Lemma \ref{l:formulas} (ii), we have (by the induction hypothesis and by \eqref{e:estvarnbglob},
 \begin{eqnarray*}Var((1-\cS)K)&\leq & 2 Var(\nabla^* K) + 2Var_{nb}(K)\\
 &\leq & \frac{8 (q+1)\tilde C(m-1, \beta)}n \norm{K}^2_{\cH} + 8q^2\norm{K}_{sup}^2\, \left(o_{ n,  m-1}(1)_{G\To\mathfrak X } + \tilde C(m-1)O(n^{-k}))\right)\\
 && +  2\frac{\tilde C(m, J_{\varepsilon}, \beta)}{ (n+1)}\norm{K}_{\cH}^2 +2\norm{K}_{sup}^2 \, (o_{n, m, \delta, \varepsilon}(1)_{G\To\mathfrak X } + n c(m)^2O(\delta)+ c(m)^2O(\varepsilon)).
  \end{eqnarray*}
 
 By Lemma \ref{l:useful}, we have for all $K\in \cH_m^0$, for all $n$ and $T$,
 \begin{multline*}
 Var(K)\leq  \frac{2 C(m,\beta)^2(8 (q+1)\tilde C(m-1, \beta)+ 2q\tilde C(m, J_{\varepsilon}, \beta))}{n}\norm{K}^2_{\cH^0_m}
 +  2\frac{C(m, \beta)^2}{T^2}\norm{K}^2_{\cH_m^0}\\
  +4 T^2 \norm{K}_{sup}^2 \, \left(o_{n, m, \delta, \varepsilon}(1)_{G\To\mathfrak X } + n c(m)^2O(\delta)+  c(m)^2O(\varepsilon) +    o_{ n,  m-1}(1)_{G\To\mathfrak X } + 8q^2\tilde C(m-1)O(n^{-k})\right)
 \\ +2\norm{K}^2_{sup} \tilde\tau(m)^2\frac{1}{|V|}\sharp\{
x\in V, \rho(x)\leq m\}
  \end{multline*}
 We can again take $\varepsilon =\delta=n^{-1}T^{-2-k}$ ($k$ arbitrary), and $n=n(T) \geq T$ large enough so that the first term is smaller than the second one. This gives
 \begin{multline*}
 Var(K)\leq   4\frac{C(m, \beta)^2}{T^2}\norm{K}^2_{\cH_m^0}
  \\ +4 \norm{K}_{sup}^2 \, \left(o_{T, m }(1)_{G\To\mathfrak X } +  c(m)^2O(T^{-k})  + 8 q^2 \norm{K}_{sup}^2 \,\tilde C(m-1)O(T^{-k+2})\right).
  \end{multline*}

 We see that an estimate of the form \eqref{e:finalfinal} will hold for all $K\in \cH_{m}^0$. This completes the induction and proves Theorem \ref{t:gen} in its form \eqref{e:weaker}.
 
\section{Anisotropic random walks on regular graphs \label{s:anis}}

Ideally, we would like to be able to deal with disordered systems, such as regular graphs with random weights on the edges or random potentials (Anderson model), or with certain non-regular graphs. Here we show that Proof 4 is adaptable to other models,
by treating a model that is still simple and without disorder (because it is homogeneous), but where the general ideas start to emerge~: the anisotropic random walks on regular graphs.

\subsection{The set-up\label{s:label}}
Let $G=(V, E)$ be a $(q+1)$-regular graph. We keep the notation introduced in \S \ref{s:notation} and \S \ref{s:nbt} -- in particular, $B$ is the set of directed edges -- and in addition we assume that there is a labelling of the edges
$${\mathrm c}: B \To \{1,\ldots, q+1\}$$
such that ${\mathrm c}(e)={\mathrm c}(\hat e)$ for all $e$, and such that for every $x\in V$, for every $j\in \{1, \cdots, q+1\}$,
$$\sharp\{ e\in B, x=o(e), {\mathrm c}(e)=j\}= 1.$$
In other words, among the $q+1$ edges containing any given vertex, exactly one has each label.

\begin{example} Let $\fG$ be a group acting on a set $V$. If $S$ is a symmetric subset of $\fG$ which acts freely on $V$,
the associated Schreier graph $G$ is obtained by drawing an edge between $x\in V$ and $y\in V$ iff $y=xs$ for some $s\in S$ (in this example, $q+1=|S\cup S^{-1}|$). Now every map $c :S\To \{1, \ldots, q+1\}$ such that ${\mathrm c}(s)={\mathrm c}(s^{-1})$ gives rise to a labelling of the edges of $G$.
It turns out that every regular graph of even degree can be realized as a Schreier graph \cite{Gross}.
\end{example}

Now assume we are given $p_1, \ldots, p_{q+1}>0 $ such that $\sum_{j=1}^{q+1} p_j=1$.
We consider the random walk on $V$ defined by saying that the probability of jumping from $x$ to $y$ is $p_j$ if $x$ and $y$ are joined by an edge labelled $j$.

We will be interested in the eigenfunctions of the stochastic symmetric operator
$$\cA_p f(x)=\sum_{y\sim x} p_{{\mathrm c}(x, y)} f(y),$$
where ${\mathrm c}(x, y)$ is the label of the edge $(x, y)$.
This operator is self-adjoint on $\ell^2(V)$, $V$ being endowed with the uniform measure.

In the case when $G$ is the infinite $(q+1)$-regular tree, the spectral analysis of the operator $\cA_p$ has been fully described by Aomoto \cite{Ao} and Fig\`a-Talamanca--Steger \cite{FTS}.
In \S \ref{s:green} and \ref{s:density} we will recall some of their results concerning the Green function and the spectral density.

\subsection{Non-backtracking operator}
The discussion of this section holds for any regular graph, labelled as in \S \ref{s:label}. Later on, when we prove Quantum Ergodicity, we will
restrict our attention to large finite graphs.

Let $\psi: V\To \IC$ be an eigenfunction of $\cA_p$,
$$\cA_p\psi=\lambda \psi.$$
Inspired by \S \ref{s:nbt}, we want to find a method to transform the eigenfunction $\psi$ into a function $f:B\To \IC$
of the form
\begin{equation}\label{e:link}f(e) = \psi(t(e))-\zeta({\mathrm c}(e)) \psi(o(e))\end{equation}
and which is an eigenfunction for some non-backtracking operator (a variant of the operator $\cB$ used in \S \ref{s:nbt}). Here $\zeta$ is a function from $\{1,\ldots, q+1\}$ to $\IC$, but we will also see it as
a function from $B$ to $\IC$, depending only on the labeling, by writing $\zeta(e)$ instead of $\zeta({\mathrm c}(e))$, for $e\in B$. Similarly, we will write $p(e)$ instead of 
$p_{{\mathrm c}(e)}$ (the probability assigned to the edge $e$).
If our graph looks locally like a tree, the non-backtracking trajectories are much simpler to study than
the trajectories of the simple random walk; this idea was already used in \S \ref{s:nbt} and is a motivation for what follows.

Denote by $\cB_p$ the non-backtracking operator acting on $\ell^2(B)$, defined by
$$\cB_p f(e)=\sum_{e', o(e')=t(e), e'\not= \hat e} p(e')f(e')$$
for $f$ a function on $B$.

If $f$ is related to an eigenfunction $\psi$ by formula \eqref{e:link},
we see that
$$\cB_p f (e)=-p(e) \psi(o(e))+\lambda \psi(t(e)) -\cB_p\zeta (e)\,\,\psi(t(e)).$$
A sufficient condition for $f$ to be solution of an equation of the form $\cB_p f =\beta f$, with $\beta$ a function on $B$, is that the function $\zeta$ satisfy
$$\cB_p\zeta (e)=\lambda -\frac{p(e)}{\zeta(e)}.$$
In that case, and if $\zeta$ does not vanish, we can take
$$\beta(e)=\frac{p(e)}{\zeta(e)}$$
Thus we should look for $\zeta : \{1, \ldots,q+1\}\To \IC$ satisfying :
$$\lambda - \frac{p_j}{\zeta(j)}=\sum_{k\not= j}  p_k \zeta(k) .$$
Equivalently
\begin{equation}\label{e:green}\lambda - \frac{p_j}{\zeta(j)}+p_j\zeta(j)=\sum_{k=1}^{q+1} p_k \zeta(k) .\end{equation}
Although this is not visible in our notation, we actually impose that $\zeta(k)$ must be a function of $p_k$ (that is, $\zeta(k)=\zeta(j)$ whenever $p_k=p_j$).


\subsection{Green functions on the weighted tree\label{s:green}}
It turns out that the system of equations \eqref{e:green} is the same that appears in the calculation of the Green function $(\gamma - \cA_p)^{-1}$ when $G=\mathfrak{X}$, the $(q+1)$-regular tree. These calculations were first done in \cite{Ao}, but we use the notation of \cite{FTS}. There,  the $(q+1)$-regular tree $\mathfrak X$ is seen as the Cayley graph of the free product $\cG=\Z/ 2\Z \star \Z/ 2\Z\star \ldots \star \Z/ 2\Z$ ($q+1$ times) with its canonical system of generators 
$a_1=a_1^{-1}$,..., $a_{q+1}=a_{q+1}^{-1}$. We will denote by $o$ the identity element of $\cG$, and it will serve as an origin on the Cayley tree $\mathfrak X$. On $\mathfrak X$, $\cA_p$ is the generator of a $\cG$-invariant random walk with probability transitions $\cA_p(x, xa_j)=p_j$.

For $\gamma\not\in \R$, the operator $\gamma-\cA_p$ is invertible in $\ell^2(\mathfrak X)$. Let us recall the formulas found in \cite{FTS, Ao} (analogous formulas also appear in other contexts, under the name ``resolvent identity'' or ``Schur complement formula''). In the special case $p_1=\ldots=p_{q+1}$, note that we are back to the Fourier analysis of \S \ref{s:Fouriertree}.

\begin{proposition} \label{p:green} There are complex numbers $w, \zeta(1),\ldots, \zeta(q+1)$ (which are algebraic functions of $\gamma \in \IC\setminus \IR$ and of the $p_j$s)
such that
$$(\gamma - \cA_p)^{-1} (y, y)=\frac{1}{2w}$$
$$(\gamma - \cA_p)^{-1} (y, yg)=\frac{ \zeta(i_1)\ldots \zeta(i_M)}{2w}$$
if $g=a_{i_1}\ldots a_{i_M}$ is the reduced words expression for $g$.

In addition, the interpretation of $\zeta(j)$ is that
$$\zeta(j)= p_j (\gamma -\cA_j)^{-1}(o,o)$$
where $\cA_j(x, y)=\cA_p(x, y)$ except for $x=o$, $y=a_j$, for which $\cA_j(o, a_j)=0$ (i.e. we cut the edge between $o$ and $a_j$).

In addition, $w$ and the $\zeta(j)$s satisfy the equations
\begin{itemize} 
\item[(a)] $\gamma=\sum_{j=1}^{q+1} p_j\zeta(j) +2w$
\item[(b)] $p_j(\zeta(j)^{-1}-\zeta(j))=2w$
\item[(c)] among all the solutions of the previous system of polynomial equations, $w$ and $\zeta(j)$ are characterized by the fact that $\Im m \zeta(j) <0$ if $\Im m \gamma >0$ (and $\Im m \zeta(j) >0$ if $\Im m \gamma <0$).
\end{itemize}
\end{proposition}
We kept the notation of \cite{FTS} -- which somehow hide the fact that $\zeta(k)=\zeta(j)$ whenever $p_k=p_j$.

The last item is not proven in \cite{FTS, Ao} but is a consequence of the results in \cite{Froese}. The probability $p_j$ will be fixed throughout the paper, however $\gamma$ will vary, so we will denote $\zeta_{\gamma}(j)$ when we want to emphasize the dependence of $\zeta$ on $\gamma$. As above, $\zeta_{\gamma}: \{1,\ldots, q+1\}\To \IC$ will also be seen as a function on $B$, by allowing the abuse of notation $\zeta_{\gamma}(e)=\zeta_{\gamma}({\mathrm c}(e))$ for $e\in B$.

\subsection{The spectral density and the spectral theorem for $\cA_p$ on $\mathfrak X$\label{s:density}}
The results of this section are borrowed from \cite{FTS, Ao} and generalize the Fourier analysis of \S \ref{s:Fouriertree}.

For $x\in \mathfrak X$ (seen as the Cayley graph of $\cG=\Z/ 2\Z \star \Z/ 2\Z\star \ldots \star \Z/ 2\Z$), $o$ the identity element of $\cG$ and $\gamma\not\in \R$, define 
$$g_\gamma(x)=(\gamma-\cA_p)^{-1}(x, o).$$
Lemma (3.1) in Chapter 2 of \cite{FTS} tells us that for every $x\in \mathfrak X$, for every $\lambda\in\R$, $g_\gamma(x)$ has a limit when 
$\gamma=\lambda + i\eps$ and $\eps>0$ tends to $0$; similarly when $\gamma=\lambda - i\eps$ and $\eps>0$ tends to $0$. These limits will be denoted by $g_{\lambda \pm i0}(x)$.

The spectral measure for the operator $\cA_p$ on $\ell^2( \mathfrak X)$ and the state
$\delta_o$ (Dirac mass at the identity of $\cG$) is absolutely continuous, the density being given by
$${\mathrm{m}}(\lambda)=-\frac{1}{\pi} \Im m \, g_{\lambda + i 0}(o)  = -\frac{1}{2i\pi} (g_{\lambda + i 0}(o)-  g_{\lambda - i 0}(o)) .$$

Moreover, one can introduce the ``Poisson kernel''~: for $x\in  \mathfrak X$ and $\omega\in \partial  \mathfrak X$ and $\gamma\not\in \R$,
$$P_\gamma(x, \omega)=\lim_{y\in \mathfrak{X}, y\To \omega}\frac{g_\gamma(y^{-1}x)}{g_\gamma(y^{-1})}.$$
The existence of the limit is trivial because of Proposition \ref{p:green}. Actually, this is also well defined for $\gamma=\lambda + i0$ as soon as the spectral density at $\lambda$, $g_{\lambda + i 0}(o)-  g_{\lambda - i 0}(o)$ does not vanish (in this case, $w_{\lambda+i0}\not= \infty$ and from the algebraic relations given in Proposition \ref{p:green} we see that the function $g_{\lambda + i 0}$ does not vanish). For $\lambda\in \IR$ and $\omega\in \partial  \mathfrak X$ we let $P_{\lambda, \omega}$ be the function $x\mapsto P_{\lambda +i0}(x, \omega)$ defined on $\mathfrak X$.

For $\lambda$ in the $\ell^2$-spectrum of $\cA_p$ and such that the spectral density $g_{\lambda + i 0}(o)-  g_{\lambda - i 0}(o)$ is non-zero, there exists a unique probability measure $\nu_{\lambda}$ on $\partial \mathfrak X$ satisfying the quasi-invariance property for the action of $\cG$,
$$\frac{d\nu_{\lambda}(g\omega)}{d\nu_{\lambda}(\omega)}=| P_\lambda(g^{-1}, \omega) |^2 .$$
If $a_{i_1}\ldots a_{i_M}$ is a reduced word, and if $[ a_{i_1}, \ldots, a_{i_M}]\subset \partial \mathfrak X$ denotes the set of infinite reduced words starting with the letters $a_{i_1}\ldots a_{i_M}$,
we have the simple expression
$$\nu_{\lambda}([ a_{i_1}, \ldots, a_{i_M}])=| \zeta_{\lambda}(i_1)|^2\ldots |\zeta_{\lambda}(i_{M-1})|^2 \frac{|\zeta_{\lambda}(i_M)|^2}{1+ |\zeta_{\lambda}(i_M)|^2}.$$
Here $\zeta_{\lambda}(j)$ is the limit at $\gamma=\lambda + i0$ of the function $\zeta_{\gamma}(j)$ defined in Proposition \ref{p:green}.
The fact that $\nu_{\lambda}$ is a well-defined probability measure, in particular that it satisfies the Kolmogorov compatibility condition, comes from the relation
\begin{equation}\label{e:kolmo}\sum_{j=1}^{q+1}\frac{|\zeta_{\lambda}(j)|^2}{1+|\zeta_{\lambda}(j)|^2}=1,\end{equation}
valid for real $\lambda$ as soon as $\Im m( w_{\lambda+i0})\not=0$, that is, as soon as the spectral density at $\lambda$ does not vanish (as a matter of fact, taking the imaginary part of
relations (a) and (b) in Proposition \ref{p:green} yields the relation $0= \Im m (w_{\lambda+i0})\left(\sum_{j=1}^{q+1}\frac{|\zeta_{\lambda}(j)|^2}{1+|\zeta_{\lambda}(j)|^2}-1\right)$).

Note also that
\begin{equation}\label{e:funny}\nu_{\lambda}([ a_{i_1}, \ldots, a_{i_M}])= |P_{\lambda, \omega}(a_{i_1} \ldots a_{i_M})|^{-2} \frac{1}{1+ |\zeta_{\lambda}(i_M)|^2}\end{equation}
for any $\omega \in [ a_{i_1}, \ldots, a_{i_M}]$, this will be useful later.

For any $f\in \ell^2(\mathfrak X)$, the spectral decomposition of $f$ for the operator $\cA_p$ reads~:
\begin{equation} f(x)= \int_{\partial  \mathfrak X\times \IR} \la P_{\lambda, \omega}, f\ra_{\ell^2(\mathfrak X)} P_\lambda(x, \omega) d\nu_{\lambda}(\omega) {\mathrm{m}}(\lambda) d\lambda
\end{equation}
and there is a Plancherel formula,
\begin{equation}
\norm{f}^2_2=\int_{\partial  \mathfrak X\times \IR} | \la P_{\lambda, \omega}, f\ra |^2 d\nu_{\lambda}(\omega) {\mathrm{m}}(\lambda) d\lambda.
\end{equation}

If $\widehat{K}$ is an operator on $\ell^2(\mathfrak X)$ with a compactly supported kernel $K(x, y)$, we also have
\begin{equation}\label{e:traceanis}\Tr_{\ell^2(\mathfrak X)} \widehat{K}=\sum_{x\in \mathfrak X} K(x, x)=\int_{\partial  \mathfrak X\times \IR} \la  P_{\lambda, \omega}, K  P_{\lambda, \omega}\ra  d\nu_{\lambda}(\omega) {\mathrm{m}}(\lambda) d\lambda.\end{equation}

This implies, for any $\gamma\not\in \IR$,
$$ \Tr_{\ell^2(\mathfrak X)} \widehat{K} ( (\gamma-\cA_p)^{-1}- (\overline{\gamma}-\cA_p)^{-1}) 
=\int_{\partial  \mathfrak X\times \IR} \frac{-2i\Im m \gamma}{|\gamma-\lambda|^{2}}\la  P_{\lambda, \omega}, K  P_{\lambda, \omega}\ra  d\nu_{\lambda}(\omega) {\mathrm{m}}(\lambda) d\lambda.$$
Letting $\gamma=\lambda+i\eps$ with $\eps>0$ and $\eps\To 0$, we obtain
\begin{equation}\label{e:partialtrace} \sum K(x, y)\Im m g_{\lambda + i0}(x^{-1}y)=
  \Im m \, g_{\lambda + i 0}(o)
 \int_{\partial  \mathfrak X }  \la  P_{\lambda, \omega}, K  P_{\lambda, \omega}\ra  d\nu_{\lambda}(\omega)
 \end{equation} 
(this of course could also be retrieved from Proposition (2.2) in \cite{FTS} Chapter 3).

\subsection{Non-backtracking quantum variance, definition and estimates}
We are now ready to start with the study of quantum ergodicity. We now assume that $(G, {\mathrm c})$ is a finite regular labelled graph (belonging to a family satisfying (EXP) and (BST)). We can write $(G, {\mathrm c})$ (the labelled graph) as a quotient $\Gamma\backslash (\mathfrak X, {\mathrm c})$ of the labelled tree $(\mathfrak X, {\mathrm c})$ by a group of automorphisms $\Gamma$ preserving the labelling.

Without knowing yet what the final result should look like, we proceed by analogy with what was done in \S \ref{s:nbt}. We start by defining a ``non-backtracking quantum variance'', and use the eigenfunction property to notice (trivially) that this quantum variance possesses a certain invariance property. From that, we hope to show (non-trivially) that the quantum variance is small for all operators. For the same reasons as in \S \ref{s:nbtqv}, we work in a small spectral interval around a value $E_0$ where the spectral density $-\frac{1}{2i\pi} (g_{E_0 + i 0}(o)-  g_{E_0 - i 0}(o))$ is positive. Let $I=(E_0-\delta, E_0+\delta)$, where $\delta$ is fixed (but can be taken arbitrarily small).

We start with an orthonormal basis $(\psi_j)$ of eigenfunctions of $\cA_p$, with associated eigenvalues $\lambda_j$. For each $j$ we build a function
$f_j$ on $B$ thanks to formula \eqref{e:link}, taking $\zeta=\zeta_{\lambda_j}$
(the value at $\gamma=\lambda_j +i 0$ of the function $\zeta$ of Proposition \ref{p:green}). We define a function $g_j$ on $B$ by letting $g_j (e)\defeq  \frac{p(e)}{\zeta_{\lambda_j}(e)}f_j (e)$ for $e\in B$. Recall that $p$ and $\zeta_{\lambda_j}$, originally defined as functions on $\{1, \ldots, q+1\}$, may also be seen as functions on $B$ depending only on the labels.
Recall the notation $\iota :\ell^2(B)\To \ell^2(B)$ for the edge-reversal involution.

We now define the non-backtracking quantum variance
\begin{equation*}Var_{nb}(K)
=\frac1{N}\sum_{j=1}^N | \la  { \iota  g_j}, \widehat{K}_B   g_j \ra_{\ell^2(B)}|^2,
\end{equation*}
and the local quantum variance
on $I$ by
\begin{equation*}Var^I_{nb}(K)
=\frac1{N(I)}\sum_{\lambda_j\in I} | \la  { \iota  g_j}, \widehat{K}_B   g_j \ra_{\ell^2(B)}|^2
\end{equation*}
for $K\in \cH_{\leq m}$. Here $N(I)=|\{ j, \lambda_j\in I\}$. If a sequence of labelled graphs $(G, {\mathrm c})$ satisfies (BST), we have
\begin{equation}\label{e:consBST}N(I)\sim |V| \int_I {\mathrm{m}}(\lambda) d\lambda.\end{equation}

\begin{remark}\label{r:extension} Note that we may want to extend this definition to families of operators, $\lambda\in\IR\mapsto K_\lambda\in \cH_{\leq m}$
depending on $\lambda$ in a $C^1$ (or even $C^0$) fashion. The quantum variance is then defined as
\begin{equation*}Var^I_{nb}(K)
=\frac1{N(I)}\sum_{\lambda_j\in I} | \la  { \iota  g_j}, \widehat{K}_{\lambda_j,B}   g_j \ra_{\ell^2(B)}|^2
\end{equation*}
The proof below extends to this case without any additional difficulty.
\end{remark}

Let us study the invariance properties of $Var^I_{nb}(K)$, coming from the eigenfunction equation. Recall that $f_j$ has been designed to satisfy
$$\cB_p f_j=\frac{p}{\zeta_{\lambda_j}} f_j$$
and so
$$\cB_p \frac{\zeta_{\lambda_j}}pg_j= g_j.$$
Since
$$\cB_p^*=p(\iota \cB_p \iota)p^{-1},$$
we also have
$$\frac{1}p \cB_p^* \zeta_{\lambda_j} {\iota g_j}=  {\iota g_j}$$
(note that $\iota \zeta_{\lambda_j}=\zeta_{\lambda_j}$).

This implies that, for any given $j$,
$$ \la  { \iota  g_j}, \widehat{K}_B   g_j \ra_{\ell^2(B)}=\left\la  { \iota  g_j}, \widehat{K}_B  \left( \cB_p \frac{\zeta_{\lambda_j}}p\right) g_j \right\ra_{\ell^2(B)}= \left\la  { \iota  g_j}, \left(\bar\zeta_{\lambda_j}\cB_p \frac{1}p\right)\widehat{K}_B   g_j \right\ra_{\ell^2(B)},$$
or more generally, for any integers $k, l$,
$$ \la  { \iota  g_j}, \widehat{K}_B   g_j \ra_{\ell^2(B)}=\left\la  { \iota  g_j}, \left(\bar\zeta_{\lambda_j}\cB_p \frac{1}p\right)^k\widehat{K}_B  \left( \cB_p \frac{\zeta_{\lambda_j}}p\right)^l g_j \right\ra_{\ell^2(B)}.$$
Since $\zeta_\lambda$ is an analytic function of $\lambda$, we have
$|\zeta_{\lambda_j}-\zeta_{E_0}|\leq C\delta$ for all $\lambda_j\in I$. It is clear also that $\norm{g_j}_{\ell^2(B)}$ is bounded independently of the graph,
so the previous equality implies
$$ \la  { \iota  g_j}, \widehat{K}_B   g_j \ra_{\ell^2(B)}=\left\la  { \iota  g_j}, \left(\bar\zeta_{E_0}\cB_p \frac{1}p\right)^k\widehat{K}_B  \left( \cB_p \frac{\zeta_{E_0}}p\right)^l g_j \right\ra_{\ell^2(B)}+
(k+l)O(\delta)\norm{\widehat{K}_B}_{\ell^2(B)\To \ell^2(B)}.$$

Similarly to what was done in \eqref{e:lshift} and \eqref{e:rshift}, we introduce weighted left- and right-shifts,
$$\sigma_{E_0 }K(x, y)= \bar\zeta_{E_0}(x, x') K(x', y)$$
$$\rho_{E_0 }K(x, y)= K(x, y') \zeta_{E_0}(y', y)$$
where $(x;y)=(x, x', \ldots, y', y)$. Note that both $\sigma_{E_0 }$ and $\rho_{E_0 }$ map $\cH_k$ to $\cH_{k+1}$, and that they commute.
They are defined so that
$$ \widehat{K}_B \circ \left( \cB_p \frac{\zeta_{E_0}}p\right) =  \widehat{\rho_{E_0 }K}_B  $$
and
$$ \left(\bar\zeta_{E_0}\cB_p \frac{1}p\right)\circ \widehat{K}_B=  \widehat{\sigma_{E_0 }K}_B.$$

 For any fixed integer $T$, we find that
$$Var^I_{nb}(K)=Var^I_{nb}\left(\sigma_{E_0 }^T K \right) +TO(\delta)= Var^I_{nb}\left( \rho_{E_0 }^TK\right) +TO(\delta)c(m)^2\norm{K}^2_{sup}$$
and more generally
$$Var^I_{nb}(K)=Var^I_{nb}\left( \frac{1}T\sum_{k+l=T-1}\sigma_{E_0 }^k \rho_{E_0 }^l K  \right) +TO(\delta)c(m)^2\norm{K}^2_{sup}.$$

 Let us from now on assume that $K\in \cH_m$ ($m\geq 1$). Then  $\sigma_{E_0 }^k \rho_{E_0 }^l K $
is in $\cH_{m+k+l}=\cH_{m+T-1}$, and we have the expression
\begin{multline*}
\sigma_{E_0 }^k \rho_{E_0 }^l K (x_0; x_{m+T-1})
\\=\bar \zeta_{E_0}(x_0, x_1)\ldots\bar\zeta_{E_0}(x_{k-1}, x_k) K(x_k; x_{k+m})\zeta_{E_0}(x_{k+m}, x_{k+m+1})\ldots \zeta_{E_0}(x_{m+T-2},x_{m+T-1})
\end{multline*}
if $(x_0; x_{m+T-1})=(x_0, x_1, \ldots, x_{m+T-1})$ is the non-backtracking path in the tree $\mathfrak X$ from $x_0$ to $x_{m+T-1}$.
To make the calculation a little nicer, we consider
\begin{equation}\label{e:tildek}\tilde K(x_0, x_m)\defeq \zeta_{E_0}(x_0, x_1)\ldots \zeta_{E_0}(x_{m-1}, x_{m})K(x_0, x_m),\end{equation}
so that we have
$$
\sigma_{E_0 }^k \rho_{E_0 }^l \tilde K (x_0; x_{m+T-1})
=\zeta_{E_0}(x_0, x_1)\ldots \zeta_{E_0}(x_{m+T-2},x_{m+T-1})K(x_k, x_{k+m}) \bar u(x_0, x_1) \ldots \bar u(x_{k-1}, x_k)
$$
where $u$ is the function of modulus one defined by
$$u(e)=\frac{\zeta_{E_0}(e)}{\bar\zeta_{E_0}(e)}=\frac{\zeta_{E_0}(x, y)}{\bar\zeta_{E_0}(x, y)}$$
for every edge $e=(x;y)\in B$.

\begin{prop}\label{p:HS2}
 Proposition \ref{p:crucialthing} still holds.
  \end{prop}

Let us postpone the proof of this proposition to the end of the section, and see what it implies.
Recall that we see $\cH_m$ as the set of functions on $B_m$, and introduce the (stochastic) operator $\cS_{E_0}: \cH_m\To \cH_m$
$$\cS_{E_0} K(\omega)=\sum_{\omega'\rightsquigarrow \omega}\frac{1+|\zeta(\omega_0, \omega_1)|^2}{1+|\zeta(\omega'_0, \omega'_1)|^2}|\zeta(\omega'_0, \omega'_1)|^2 K(\omega')$$
with $\omega=(\omega_0, \ldots, \omega_m)$ and $\omega'=(\omega'_0, \ldots, \omega'_m)$ (and of course $\omega'_1=\omega_0$). Introduce also
$$\cS^u_{E_0} K(\omega)=\sum_{\omega'\rightsquigarrow \omega}\frac{1+|\zeta(\omega_0, \omega_1)|^2}{1+|\zeta(\omega'_0, \omega'_1)|^2}|\zeta(\omega'_0, \omega'_1)|^2 u(\omega'_0, \omega'_1) K(\omega').$$
Applying Proposition \ref{p:HS2}, we obtain
\begin{equation}\label{e:majhs1}Var^I_{nb}(\tilde K)\leq C \left\Vert \frac{1}T\sum_{k+l=T-1}\sigma_{E_0 }^k \rho_{E_0 }^l \tilde K\right\Vert^2_{\cH} + o_{T, E_0, K}(1)_{G\To \mathfrak{X}}\norm{K}_{sup}^2
\end{equation}
with $C$ independent of $T$, $E_0$ and $\delta$.

\begin{remark}At this point, it is convenient to modify the Hilbert space structure on $\cH$. We introduce $\cH_{E_0}=\oplus_{m\geq 0} \cH_{m, E_0}$, where the decomposition is orthogonal,
the vector space $\cH_{m, E_0}$ is the same as $\cH_m$, but for $m\geq 1$ the norm is modified by
$$\norm{K}^2_{\cH_{m, E_0}} =\frac{1}{|\cD|}\sum_{x\in \cD, y\in\mathfrak{X}} \frac{1}{1+ |\zeta_{E_0}(x, x')|^2} |K(x, y)|^2 \frac{1}{1+ |\zeta_{E_0}(y', y)|^2}$$
if $(x;y)=(x, x', \ldots, y', y)$ is the non-backtracking path between $x$ and $y$. We also introduce
 $\widetilde \cH_{E_0}=\oplus_{m\geq 0} \widetilde\cH_{m, E_0}$, where the decomposition is orthogonal,
the vector space $\widetilde\cH_{m, E_0}$ is the same as $\cH_m$, but for $m\geq 1$ the norm is modified by
$$\norm{K}^2_{\widetilde\cH_{m, E_0}} =\frac{1}{|\cD|}\sum_{x\in \cD, y\in\mathfrak{X}} \frac{1}{1+ |\zeta_{E_0}(x, x')|^2} |\tilde K(x, y)|^2 \frac{1}{1+ |\zeta_{E_0}(y', y)|^2},$$
$\tilde K$ and $K$ being related by \eqref{e:tildek}.
Obviously the norms on $\cH_{E_0}$ is equivalent to the norm on $\cH$; however it is not equivalent to the norm on $\widetilde\cH_{E_0}$.
\end{remark}

Possibly modifying the factor $C$, we can write
\begin{equation}\label{e:majhs2}Var^I_{nb}(\tilde K)\leq C \left\Vert \frac{1}T\sum_{k+l=T-1}\sigma_{E_0 }^k \rho_{E_0 }^l \tilde K\right\Vert^2_{\cH_{E_0}} + o_{T, E_0, K}(1)_{G\To \mathfrak{X}}\norm{K}_{sup}^2.
\end{equation}
instead of \eqref{e:majhs1}. Like in Lemma \ref{l:lem}, we see by explicit development of the scalar product that
\begin{equation*}\la \sigma_{E_0 }^{k'} \rho_{E_0 }^{l'} \tilde K,\sigma_{E_0 }^k \rho_{E_0 }^l \tilde K\ra_{\cH_{E_0}} =  \la( \cS^u_{E_0})^{k-k'} K, K\ra_{\widetilde\cH_{E_0}}
\end{equation*}
for $k'\leq k$, $k+l=k'+l'=T-1$. Here we start seeing the reason why introducing the new Hilbert space $\widetilde\cH_{E_0}$~: the property
\eqref{e:kolmo} implies that
$$ \la\cS_{E_0}K, \bbbone_m\ra_{\widetilde\cH_{m, E_0}}= \la K, \bbbone_m\ra_{\widetilde\cH_{m, E_0}},$$
in other words the linear form on $\widetilde\cH_{m, E_0}$
$$\mu : K\mapsto \la K, \bbbone_m\ra_{\widetilde\cH_{E_0}}$$
is invariant under the action of $\cS_{E_0}$ (if $\bbbone_m$ is the constant function as in \eqref{e:bbone}). Since $\cS_{E_0} \bbbone_m=  \bbbone_m$, we also have $\cS^*_{E_0} \bbbone_m=  \bbbone_m$ if the adjoint
is taken in $\widetilde\cH_{m, E_0}$. The fact that $\cS^*_{E_0}$ is stochastic will be used in a crucial way in the proof of Lemma \ref{l:decay}, which states decay of correlations for the iterates of the operator $\cS^u_{E_0}$~:
\begin{lemma} \label{l:decay} Unless $u$ is a constant function, for each $m$ there exists $\delta>0$ such that
$$\norm{(\cS^u_{E_0})^{m+1}}_{\widetilde\cH_{m, E_0}\To \widetilde\cH_{m, E_0}} \leq 1-\delta.$$
\end{lemma}
Actually $\delta$ can be probably taken independent of $m$, but we won't use this fact; on the other hand, $\delta$ does depend on the $p_j$ and on $E_0$.

Again we postpone the proof to the end of the section to see first what the lemma implies. We note that $u$ being a constant is completely independent of the graph, it depends only on the $p_j$ and on $E_0$. If we are in the conditions of application of Lemma \ref{l:decay}, then
$$|\la( \cS^u_{E_0})^{k-k'} K, K\ra_{\widetilde\cH_{E_0}}|\leq (1-\delta)^{(k-k')/(m+1)} \norm{K}_{\widetilde\cH_{E_0}}.$$
In this case we conclude as in \eqref{e:estvarnb} that
 \begin{equation*}Var^I_{nb}(\tilde K )\leq \frac{C(m, E_0)}{ T+1}\norm{K}_{\cH}^2 +\norm{K}_{sup}^2 \, (o_{T, m, \delta}(1)_{G\To\mathfrak X } + T  c(m)^2\delta) 
\end{equation*}
for all $K$, and thus also, after modifying the factors $C(m, E_0)$ and $c(m)$,
\begin{equation}\label{e:nbtanis}Var^I_{nb}( K )\leq \frac{C(m, E_0)}{ T+1}\norm{K}_{\cH}^2 +\norm{K}_{sup}^2 \, (o_{T, m, \delta}(1)_{G\To\mathfrak X } + T  c(m)^2\delta) .\end{equation}

Suppose now that $u$ is constant; by Lemma \ref{l:fnu} we must have $u=-1$ and $E_0=0$ (unless we are in the isotropic case, which has already been dealt with). In Lemma \ref{l:cheeger} we will use the (EXP) condition and show that there exists $\rho(\beta, m) <1$ such that
$$\norm{ \cS_{E_0}^{k} K}_{\widetilde\cH_{E_0}} \leq \rho(\beta, m)^k \norm{K}_{\widetilde\cH_{E_0}}$$
if $K\in \cH_m$ is orthogonal to $\bbbone_m$ for the scalar product of $\widetilde\cH_{E_0}$.
On the other hand, if $K=\bbbone_m$, we have for any $t_1< t_2$
$$\norm{\sum_{k=t_1}^{t_2} (\cS^u_{E_0})^{k} K}= \norm{\sum_{k=t_1}^{t_2} (-1)^k\cS_{E_0}^{k} \bbbone_m}
= \norm{\sum_{k=t_1}^{t_2} (-1)^k  \bbbone_m}\leq \norm{ \bbbone_m},$$
so arguing as in the proof of \eqref{e:estvarnb} we find that
 \begin{equation*}Var^I_{nb}(\tilde K )\leq \frac{\tilde C(m, \beta)}{ T+1}\norm{K}_{\cH}^2 +\norm{K}_{sup}^2 \, (o_{T, m, \delta}(1)_{G\To\mathfrak X } + T  c(m)^2\delta) 
\end{equation*}
for all $K$, and thus also, after modifying the factors $\tilde C(m, \beta)$ and $c(m)$,
\begin{equation}\label{e:nbtanis2}Var^I_{nb}( K )\leq \frac{\tilde C(m, \beta)}{ T+1}\norm{K}_{\cH}^2 +\norm{K}_{sup}^2 \, (o_{T, m, \delta}(1)_{G\To\mathfrak X } + T  c(m)^2\delta) .\end{equation}
 
From this, one can argue as in \eqref{e:estvarnbglob} to obtain an estimate over the full quantum variance $Var_{nb}(K)$~:
\begin{equation}
\label{e:estvarnbanisglob}
Var_{nb}(K )\leq \frac{\tilde C(m, {\varepsilon}, \beta)}{ (n+1)}\norm{K}_{\cH}^2 +\norm{K}_{sup}^2 \, (o_{n, m, \delta, \varepsilon}(1)_{G\To\mathfrak X } + n c(m)^2O(\delta)+ c(m)^2O(\varepsilon)) .
\end{equation}
Thus we know that the ``non-backtracking'' quantum variance vanishes in the limit $G\To \mathfrak X$. We do not know yet
what it implies for the ``simple'' quantum variance -- actually, this has not even been defined yet. This will be the object of the next and final section. Before that, we still have to prove Proposition \ref{p:HS2}, Lemma \ref{l:decay} and to state and prove Lemmas \ref{l:cheeger} and \ref{l:fnu}.

\begin{proof} (of Lemma \ref{l:decay})
To simplify the notation let us consider the case $m=1$. Note that on $\widetilde\cH_{1, E_0}$ we have 
$$\cS^u_{E_0}= \cS_{E_0}\circ M_u,$$
where $M_u$ is the multiplication by the function $u$, so 
$$(\cS^u_{E_0})^*=M_{\bar u}\circ \cS^*_{E_0} =M_{ u^{-1}}\circ \cS^*_{E_0}.$$
The adjoint is always taken in the Hilbert space $\widetilde\cH_{1, E_0}$. We have the explicit expression,
$$\cS^*_{E_0} K(\omega)=\sum_{\omega\rightsquigarrow \omega'}\frac{1+|\zeta(\omega_0, \omega_1)|^2}{1+|\zeta(\omega'_0, \omega'_1)|^2}|\zeta(\omega'_0, \omega'_1)|^2 K(\omega').$$
Note that
$$\norm{\cS_{E_0}}_{\widetilde\cH_{1, E_0}\To \widetilde\cH_{1, E_0}} = 1$$
(this holds because $\cS^*_{E_0}\cS_{E_0}$ is a stochastic operator, so it has spectral radius $1$; the spectral radius coincides with the norm because the operator is symmetric). As a consequence, we also have
$$\norm{\cS^u_{E_0}}_{\widetilde\cH_{1, E_0}\To \widetilde\cH_{1, E_0}} = 1,$$
but we want to prove that
$\norm{(\cS^u_{E_0})^2}_{\widetilde\cH_{1, E_0}\To \widetilde\cH_{1, E_0}} <1.$
This is equivalent to proving that
$\norm{(\cS^{u, *}_{E_0})^2(\cS^u_{E_0})^2}_{\widetilde\cH_{1, E_0}\To \widetilde\cH_{1, E_0}} <1,$ or to proving that
$$\norm{\cS^*_{E_0} M_{u^{-1}}\cS^*_{E_0}\cS_{E_0}M_u  \cS_{E_0}}_{\widetilde\cH_{1, E_0}\To \widetilde\cH_{1, E_0}} <1.$$

\begin{lemma} If $P = (P_{ij}), Q= (Q_{ij})$ are two $N\times N$ stochastic matrices and $D = (\delta_{ij}D_j), D'=  (\delta_{ij}D'_j)$ are two $N\times N$ diagonal unitary matrix, let
$R= P DP Q D' Q$. Then we have, for all $i$,
$$\sum_{m}|R_{im}| \leq 1,$$
with equality only if, for all $m$, there exists a $c$ such that $D_j D'_l =c$ for all $j, l$ such that $Q_{l m}\not= 0$, $P_{i j}\not= 0$ and $(PQ)_{jl}\not= 0$.
\end{lemma}
To check this simple lemma, write
$$R_{im}=\sum_{j, l} P_{ij } D_j (PQ)_{jl} D'_l Q_{lm}.$$
By the triangle inequality we have $|R_{im}| \leq (P^2Q^2)_{im}$, and $P^2 Q^2$ is a stochastic matrix. For $\sum_{m}|R_{im}|$ to equal $1$ we should have equality
in the triangle inequality, for all $m$. The case of equality happens if and only if there exists a $c$ such that $D_j D'_l =c$ for all $j, l$ such that $Q_{l m}\not= 0$, $P_{i j}\not= 0$ and $(PQ)_{jl}\not= 0$.

Now having $\sum_{m}|R_{im}| \leq 1-\delta$ for all $i$ implies that $\sum_{m}|R^k_{im}| \leq (1-\delta)^k$ for all $k$ and all $i$, and thus the spectral radius of $R$ is $\leq 1-\delta$.

To end the proof of Lemma \ref{l:decay}, we apply this considerations to the $B\times B$ matrices $P=\cS^*_{E_0}$, $Q=\cS_{E_0}$, $D=M_{u^{-1}}$ and $D'=M_u$. For $\sum_{e'\in B}|R_{e, e'}|$
to equal $1$ for some $e\in B$, we must have that, for all $e_1\rightsquigarrow e$ and for all $e_2$ with the same origin as $e_1$, 
$u(e_1)/u(e_2)$ is a constant, i.e. does not depend on $e_1$ and $e_2$. But clearly this implies that $u$ is a constant.

All this discussion is of course independent of the global shape of the graph (it involves only paths of length $2$ in the tree $\mathfrak X$), so we see that if $u$ is not constant, there exists $\delta$ independent of $G$ such that the spectral radius of $\cS^*_{E_0} M_{u^{-1}}\cS^*_{E_0}\cS_{E_0}M_u  \cS_{E_0}$ is less than $1-\delta$. Since this operator is symmetric, this is also its norm.
 
\end{proof}

\begin{proof} (of Proposition \ref{p:HS2}) (Quite miraculously) the proof is almost identical to that of Proposition \ref{p:crucialthing}.
We start again at step (b), and introduce the operator $\widetilde U:\ell^2(V(\mathfrak{X}))\To \ell^2(B(\mathfrak{X}))$ defined by
$$\widetilde U\psi (e)=p(e) \zeta_{ \cA_p + i 0}(e)^{-1}\psi(t(e))- p(e)\psi(o(e)).$$
The meaning of this expression is the following~: for frozen $e$, consider the function $ \zeta_{ \cA_p + i 0}(e)^{-1}\psi$, defined by
continuous functional calculus (the map $\lambda\mapsto \zeta_{\lambda+ i 0}(e)^{-1}$ is continuous). Then consider the map $e\mapsto \left( \zeta_{ \cA_p + i 0}(e)^{-1}\psi\right)(t(e)).$

We can already remark that $\widetilde U P_{\lambda, \omega}(e) =0 $ if $e$ goes away from $\omega$, and that
\begin{eqnarray}\label{e:UP}\widetilde U P_{\lambda, \omega}(e)&=&  p(e) \zeta_{\lambda}(e)^{-1}P_{\lambda, \omega}(t(e))- p(e)P_{\lambda, \omega}(o(e))\\
&=& 2w_\lambda P_{\lambda, \omega}(t(e))
\end{eqnarray}
if $e$ goes towards $\omega$. This comes from the fact that $ \zeta_{\lambda}(e)^{-1} P_{\lambda, \omega}(t(e))=  P_{\lambda, \omega}(o(e))$
if $e$ goes away from $\omega$, and $ \zeta_{\lambda}(e)P_{\lambda, \omega}(t(e))=  P_{\lambda, \omega}(o(e))$ otherwise; together with item (b) in Proposition \ref{p:green}.

We can define similarly $U:\ell^2(V )\To \ell^2(B)$, built so that
$ g_j =U\psi_j$.
Adapting \eqref{e:bound} to the anisotropic case, we obtain
\begin{multline}Var^I_{nb}(K)\leq \frac{1}{N(I)}\sum_{j=1}^{|V|}    |\la \iota g_j, \widehat{K}_B U \chi(\cA_p)\psi_j\ra|^2  
\\ \leq  \frac{C}{N(I)} \int \chi^2(\lambda) \norm{ \bbbone_{\cD}  \widehat{K}_{B(\mathfrak X)} \widetilde U P_{\lambda,\omega}}^2_{\ell^2(B(\mathfrak X))} d\nu_{\lambda}(\omega){\mathrm{m}}(\lambda)d\lambda + \norm{K}^2_{sup}\, o_{I, \chi, m}(1)_{G\To\mathfrak X } . \label{e:bound2}
 \end{multline}
 
When we calculate
$$ \widehat{K}_{B(\mathfrak X)} (\widetilde U P_{\lambda, \omega})(e)=\sum_{e', \cA^{\sharp (m-1)}(e, e')\not= 0}K(o(e), t(e')) \widetilde U P_{\lambda, \omega}(e')$$
we see that only one term in the sum can be non-zero, namely the only $e'$ going towards $\omega$ and such that $ \cA^{\sharp (m-1)}(e, e')\not= 0$. We denote this edge $e'$ by $e'_{e, \omega}$. Using \eqref{e:UP} we find
$$| \widehat{K}_{B(\mathfrak X)}(\widetilde U P_{\lambda, \omega})(e)|^2 = |K(o(e), t(e'_{e, \omega}))|^2 |2w_\lambda |^2 |P_{\lambda, \omega}(t(e))|^2.$$
Of course $ |2w_\lambda |^2$ stays bounded when $\lambda$ is in $I$, with a bound that does not depend on the graph. To estimate $ \int | \widehat{K}_{B(\mathfrak X)}(\widetilde U P_{\lambda, \omega})(e)|^2 d\nu_\lambda(\omega)$ we first note as above that it is enough to treat the case where $o(e)=o$, where the calculations will be simpler.
For any given $\lambda$ and $e$,
\begin{multline*} \int | \widehat{K}_{B(\mathfrak X)}(\widetilde U P_{\lambda, \omega})(e)|^2 d\nu_\lambda(\omega)\\ 
\leq |2w_\lambda |^2 \sum_{e', \cA^{\sharp (m-1)}(e, e')\not= 0}|K(o(e), t(e'))|^2  |P_{\lambda, \omega}(t(e))|^2 \nu_\lambda\left( \{ \omega \mbox{ s.t. }e'_{e, \omega} =e' \}\right).
\end{multline*}
In the case $o(e)=o$, the calculation is simpler, since in this case we can use \eqref{e:funny} to write that $ |P_{\lambda, \omega}(t(e))|^2 \nu_\lambda\left( \{ \omega \mbox{ s.t. }e'_{e, \omega} =e' \}\right)=\frac{1}{1+|\zeta_\lambda(e')|^2}\leq 1$.  Thus 
$$ \int |\widehat{K}_{B(\mathfrak X)}(\widetilde U P_{\lambda, \omega})(e)|^2 d\nu_\lambda(\omega)\leq   |2w_\lambda |^2 \sum_{e'}|K(o(e), t(e'))|^2$$
and we conclude as in the proof of Proposition \ref{p:crucialthing} .
\end{proof}

The following lemma is based on the results of Section \ref{s:spnbt} and extends them. If a regular graph has a spectral gap, we saw in \S \ref{s:spnbt} that the non-backtracking random walk (and more generally for any given $m\geq 1$ the transfer operator $\cS :  \ell^2(B_m)\To  \ell^2(B_m)$) have a spectral gap as well; the following lemma extends this to {\em{weighted}} non-backtracking random walks.
\begin{lemma} \label{l:cheeger} Let $G$ be a $(q+1)$-regular graph such that the adjacency matrix $\cA$ has spectral gap $\beta$. Let $\mu$ be a probability measure on $B_m$ and let $\cS_P : \ell^2(B_m, \mu)\To \ell^2(B_m, \mu)$ be a stochastic operator of the form
$$\cS_P K(\omega)=\sum_{\omega'\rightsquigarrow \omega} P(\omega, \omega') K(\omega').$$
Assume also that $\cS_P^*$ is stochastic (when the adjoint is taken in $\ell^2(B_m, \mu)$).

Then for all $k_0$ \begin{multline*}
\left(  1- 4c \left(  1- (1-\beta')^{2(k_0-m+1)}\right)^{1/2}\right)^{1/2}\leq \norm{\cS^{k_0}_P}_{\ell_0^2(B_m, \mu)\To \ell_0^2(B_m, \mu)}\\
\leq \left(  1-\frac{c^{-2}}2\left(1- (k_0-m+1)^2 (1-\beta')^{2(k_0-m+1)} \right)^2\right)^{1/2}
\end{multline*}
where $\ell_0^2(B_m, \mu)$ is the subspace of functions orthogonal to the constants.
The factor $c$ is an explicit function of $q^{-2 k_0}$, $\max_{\omega'\rightsquigarrow \omega} P(\omega, \omega')^{-1}$, $\max_{\omega}\mu(\omega)^{-1}$; and $\beta'$ is related to $\beta$ by \eqref{e:beta'}.

In particular, there exists $k_0$ depending only on $\beta$, and $\rho<1$ an explicit function of $\beta'$, $\max_{\omega'\rightsquigarrow \omega} P(\omega, \omega')^{-1}$, $\max_{\omega}\mu(\omega)^{-1}$, such that
$$\norm{\cS^{k_0}_P}_{\ell_0^2(B_m, \mu)\To \ell_0^2(B_m, \mu)}
\leq \rho.$$
\end{lemma}

\begin{proof}(of Lemma \ref{l:cheeger}).

Note that $\norm{\cS^{k_0}_P}_{\ell_0^2(B_m, \mu)\To \ell_0^2(B_m, \mu)}= \norm{\cS^{k_0 *}_P\cS^{k_0}_P}_{\ell_0^2(B_m, \mu)\To \ell_0^2(B_m, \mu)}^{1/2},$
so that we can decide to work with the stochastic symmetric operator $\cQ=\cS^{k_0 *}_P\cS^{k_0}_P$.  

For a symmetric operator, recall the link between the spectral gap and the isoperimetric constant. Let
$$\lambda=\max_{v\in\ell_0^2(B_m, \mu) } \frac{\la v, \cQ v\ra_{\ell_0^2(B_m, \mu)}}{\norm{v}^2_{\ell_0^2(B_m, \mu)}}$$
be the second eigenvalue of $\cQ$. Let
$$h_\cQ=\inf_{A\subset B_m}\frac{ \sum_{\omega\in A, \omega'\not\in A} \cQ(\omega, \omega')\mu(\omega)}{\min (\mu(A), \mu(A^c))}$$
be the ``Cheeger constant'' (isoperimetric constant). Note that $\cQ$ being symmetric as an operator on $\ell^2(B_m, \mu)$ is equivalent to the reversibility condition~: $
 \cQ(\omega, \omega')\mu(\omega)= \cQ(\omega', \omega)\mu(\omega')$.

It is not very difficult to prove that
$$1-\lambda \leq 2h_\cQ$$
and the Cheeger inequality for graphs says that
$$1-\lambda \geq \frac{h_\cQ^2}2,$$
see for instance \cite{Diac91}.
Thus $\norm{\cQ}_{\ell_0^2(B_m, \mu)\To \ell_0^2(B_m, \mu)}\leq 1-h_\cQ^2/2$.

Now, we compare with what we already know in the ``isotropic'' case, i.e. when the probability transitions
are uniform~: $P_{iso}(\omega, \omega')=1/q$. In this case, $\mu$ is the uniform measure, and 
$\cQ_{iso}(\omega, \omega')$ takes only the two values $0$ or $q^{-2 k_0}$. In addition $\cQ_{iso}(\omega, \omega')\not= 0$ iff $\cQ_{}(\omega, \omega')\not= 0$. Based on that, it is not difficult to see that there exists $c >1$ such that the two corresponding Cheeger constants are related by
$$c^{-1}  h_{iso} \leq h_\cQ \leq c h_{iso}.$$
The factor $c$ has an explicit expression as a function of $q^{-2 k_0}$, $\max_{\omega'\rightsquigarrow \omega} P(\omega, \omega')^{-1}$, $\max_{\omega}\mu(\omega)^{-1}$.

 On the other hand, we have already seen, when we calculated explicitly the spectrum on the non-backtracking random walk in \S \ref{s:nbt}, that
 $$(1-\beta')^{2(k_0-m+1)}\leq \norm{\cQ_{iso}}_{\ell_0^2(B_m)\To \ell_0^2(B_m)}\leq (k_0-m+1)^2 (1-\beta')^{2(k_0-m+1)}$$
 where $\beta'$ is related to the spectral gap by \eqref{e:beta'}.
 
 Thus
$$\norm{\cQ}_{\ell_0^2(B_m, \mu)\To \ell_0^2(B_m, \mu)}\leq 1-c^{-2}h_{iso}^2/2\leq 1-\frac{c^{-2}}2\left(1- (k_0-m+1)^2 (1-\beta')^{2(k_0-m+1)} \right)^2.$$
In the other direction (less interesting to us)
$$\norm{\cQ}_{\ell_0^2(B_m, \mu)\To \ell_0^2(B_m, \mu)}\geq 1-2c h_{iso}\geq 1- 4c \left(  1- (1-\beta')^{2(k_0-m+1)}\right)^{1/2}.$$

\end{proof}
 
 \begin{lemma}\label{l:fnu} The function $u=\frac{\zeta_{E_0}}{\bar\zeta_{E_0}} : \{1, \ldots, q+1\}\To \IS^1$ is constant only if the $p_j$ are all equal, or $E_0=0$. If $E_0=0$ and $u$ is constant, we must have $u= -1$.
\end{lemma}

\begin{proof} If $u$ is constant, then there exists $r:  \{1, \ldots, q+1\}\To\IR$, and $\theta$ such that
$$\zeta_{E_0}(j)= r(j) e^{i\theta}.$$
Remember from Proposition \ref{p:green} that $p_j(\zeta_{E_0}(j)^{-1} - \zeta_{E_0}(j))= 2w_{E_0}$ does not depend on $j$. If $\theta \not\in \{ 0, \pi, \pm \frac{\pi}2\}$, the two
complex numbers $e^{i\theta}$ and $e^{-i\theta}$ are independent over $\IR$
so we see
that this is only possible if $p_j\zeta_{E_0}(j)^{-1}$ does not depend on $j$ and $p_j\zeta_{E_0}(j)$ does not depend on $j$. Hence $p_j$ does not depend on $j$

If the $\zeta_{E_0}(j)$ are all real we are at a point $E_0$ where the spectral density vanishes, which we excluded from the start. If the $\zeta_{E_0}(j)$ are purely imaginary then $u=-1$. Moreover, from the relations given in Proposition \ref{p:green} we see that $E_0$ must be purely imaginary, hence equal to $0$. 
\end{proof}

\subsection{The original quantum variance \label{s:back2}}
There remains to do as in \S \ref{s:back1}, that is, interprete our inequality \eqref{e:estvarnbanisglob} in terms of the original eigenfunctions of $\cA_p$, the functions $\psi_j$.
We thus define for $K\in \cH_{\leq m}$,
$$Var(K)=\frac{1}N\sum_{j=1}^N |\la \psi_j, \widehat K_G\psi_j\ra|^2.$$
To understand the asymptotics of this quantity, it is actually enough to sum over $\lambda_j \in \supp \mathrm m$. By \eqref{e:consBST},
the other eigenvalues give a negligible contribution.

\begin{remark}\label{r:extension2} Note again that we may want to extend this definition to families of operators, $\lambda\in\IR\mapsto K_\lambda\in \cH_{\leq m}$
depending on $\lambda$ in a $C^1$ (or even $C^0$) fashion. The quantum variance is then defined as
\begin{equation*}Var_{nb}(K)
=\frac1{N}\sum_{\lambda_j}  |\la \psi_j, \widehat K_{\lambda_j,G}\psi_j\ra|^2.
\end{equation*}
We will sometimes allow ourselves to write $Var_{nb}(K_\lambda)$ in this situation.
 \end{remark}

Remember that the non-backtracking quantum variance deals with the quantities
$\la  { \iota  g_j}, \widehat{K}_B   g_j \ra_{\ell^2(B)}$
and that the relation between $g_j$ and $\psi_j$ is as follows~:
$$g_j(e)= \frac{p(e)}{\zeta_{\lambda_j}(e)}\left(\psi_j(t(e))-\zeta_{\lambda_j+i0}(e)\psi_j(o(e))\right).$$
Call $U_{\lambda_j}$ the operator $\ell^2(V)\To \ell^2(B)$, defined by
$$(U_{\lambda_j}\phi)(e)=\frac{p(e)}{\zeta_{\lambda_j}(e)}\left(\phi(t(e))-\zeta_{\lambda_j+i0}(e)\phi(o(e))\right)$$
for $\phi\in \ell^2(V)$. Define $\widetilde U_{\lambda_j}$ the operator defined similarly, from $\ell^2(V(\mathfrak{X}))\To \ell^2(B(\mathfrak{X}))$.

We have
\begin{eqnarray*}\la  { \iota  g_j}, \widehat{K}_B   g_j \ra_{\ell^2(B)}&=&\la \psi_j, (\iota U_{\lambda_j})^* \widehat{K}_B U_{\lambda_j} \psi_j  \ra_{\ell^2(V)}\\
&=&\la \psi_j, ((\iota \widetilde U_{\lambda_j})^* \widehat{K}_{B(\mathfrak{X})} \widetilde U_{\lambda_j})_G \psi_j  \ra_{\ell^2(V)};
\end{eqnarray*}
the last equality just means that $(\iota U_{\lambda_j})^* \widehat{K}_B U_{\lambda_j}$ is obtained by taking $(\iota \widetilde U_{\lambda_j})^* \widehat{K}_{B(\mathfrak{X})} \widetilde U_{\lambda_j}$ defined on $\ell^2(V(\mathfrak{X}))$, down to the quotient graph $G$. What we need is to identify the space of operators of the form 
$(\iota \widetilde U_{\lambda_j})^* \widehat{K}_{B(\mathfrak{X})} \widetilde U_{\lambda_j}$.

If $K\in \cH_m$ it is clear that $(\iota \widetilde U_{\lambda})^* \widehat{K}_{B(\mathfrak{X})} \widetilde U_{\lambda}$
is of the form $\widehat{\tilde K}$ with $\tilde K\in \cH_m\oplus \cH_{m-2}\oplus \cH_{m-1}$, and
$$\tilde K_m(x_0, x_1, \ldots, x_{m-1}, x_m)=\frac{p(x_0, x_1)}{\overline{\zeta_\lambda(x_0, x_1)}}K(x_0, x_1, \ldots, x_{m-1}, x_m)
\frac{p(x_{m-1}, x_m)}{ {\zeta_\lambda(x_{m-1}, x_m)}}$$
The terms $\tilde K_{m-1}, \tilde K_{m-2}$ can similarly be expressed as explicit linear functions of $K$, or equivalently of $\tilde K_m$~:
\begin{equation}\label{e:m-1}\tilde K_{m-1}= \cL^{\lambda}_{m-1}\tilde K_m
\end{equation}
\begin{equation}\label{e:m-2}
\tilde K_{m-2}= \cL^{\lambda}_{m-2}\tilde K_m
\end{equation}
where $\cL^{\lambda}_{m-i}$ are linear operators $\cH_{m}\To \cH_{m-i}$ ($i=1, 2$).
We won't need their explicit expressions (although it is not difficult to find them).
Inequality \eqref{e:estvarnbanisglob} can then be rewritten as
\begin{multline}\label{e:recursion}
Var(K + \cL_{m-1} K  + \cL_{m-2} K)\\
\leq \frac{\tilde C(m, {\varepsilon}, \beta)}{ (n+1)}\norm{K}_{\cH}^2 +\norm{K}_{sup}^2 \, (o_{n, m, \delta, \varepsilon}(1)_{G\To\mathfrak X } + n c(m)^2O(\delta)+ c(m)^2O(\varepsilon)) . \end{multline}
for all $K\in \cH_m$. Here $K + \cL_{m-1} K  + \cL_{m-2} K$ means the $\lambda$-dependent family
$$\lambda\mapsto K + \cL^\lambda_{m-1} K  + \cL^\lambda_{m-2} K,$$
and we work in the context of Remarks \ref{r:extension} and  \ref{r:extension2}. Equation \eqref{e:recursion} will be used to prove, by induction on $m$, the following Proposition (which implies Theorem \ref{t:mainanis}).

\begin{prop} \label{p:recursion}There exists $\tilde c(m) >0$, and for every $\varepsilon>0$ some $\tilde C(m, {\varepsilon}, \beta)>0$, such that~: for all $K\in \cH_m$, for all $\varepsilon, \delta>0$, for all $n$,
\begin{multline*}Var(K -\la K\ra_{\lambda, p} \bbbone_0)\\ \leq 
\frac{\tilde C(m, {\varepsilon}, \beta)}{ (n+1)}\norm{K}_{\cH}^2 +\norm{K}_{sup}^2 \, (o_{n, m, \delta, \varepsilon}(1)_{G\To\mathfrak X } + n \tilde c(m)O(\delta)+ \tilde c(m)O(\varepsilon) +\tilde c(m) O(n^{-1}))
\end{multline*}
where $\la K\ra_{\lambda,p}$ was defined in \eqref{e:klambdap}, and $\bbbone_0$ was defined in \eqref{e:bbone} and is such that $\widehat{\bbbone_0}$ is the identity operator.
\end{prop}
The remainder of the paper is devoted to proving Proposition \ref{p:recursion}. The next lemma will show that
\begin{equation}\label{e:recurKK}\la K\ra_{\lambda,p}=-\la \cL^\lambda_{m-1} K  + \cL^\lambda_{m-2} K\ra_{\lambda,p},
\end{equation}
which, together with \eqref{e:recursion}, will allow to work by induction on $m$. To initiate the induction we will have to prove Proposition \ref{p:recursion}
for $m=0$, this will be done in Lemma \ref{l:m=0}.

In what follows, recall that the labelled graph is written as a quotient $(G, \mathrm{c})=\Gamma\backslash (\mathfrak X,  \mathrm{c})$, and $\cD\subset V(\mathfrak X)$ is a fundamental domain for the action of $\Gamma$.

\begin{lemma} For any $\lambda$, the operator  
 $(\iota \widetilde U_{\lambda})^* \widehat{K}_{B(\mathfrak{X})} \widetilde U_{\lambda}$ satisfies
 $$ \sum_{x, y\in \mathfrak X} \bbbone_{\cD}(x)\left[(\iota \widetilde U_{\lambda})^* \widehat{K}_{B(\mathfrak{X})} \widetilde U_{\lambda}\right](x, y)\,\,\Im m g_{\lambda + i0}(x^{-1}y)=
 0.$$
 
 In other words, for $K\in \cH_m$,
  $$ \sum_{x, y\in \mathfrak X} \bbbone_{\cD}(x) K (x, y)\,\,\Im m g_{\lambda + i0}(x^{-1}y)=-
   \sum_{x, y\in \mathfrak X} \bbbone_{\cD}(x) ( \cL^\lambda_{m-1} K  +  \cL^\lambda_{m-2} K)(x, y)\,\,\Im m g_{\lambda + i0}(x^{-1}y)
 .$$
  
 \end{lemma}
 \begin{proof}
 Let us call $\cD'\subset B(\mathfrak X)$ a fundamental domain for the action of $\Gamma$ on the directed edges.
 We have
  \begin{multline*}\sum_{x, y\in \mathfrak X}  \bbbone_{\cD}(x)\left[(\iota \widetilde U_{\lambda})^* \widehat{K}_{B(\mathfrak{X})} \widetilde U_{\lambda}\right](x, y)\,\,\Im m g_{\lambda + i0}(x^{-1}y)\\=
\sum_{x, y\in \mathfrak X}\left[(\iota \widetilde U_{\lambda})^* \bbbone_{\cD'}\widehat{K}_{B(\mathfrak{X})} \widetilde U_{\lambda}\right](x, y)\,\,\Im m g_{\lambda + i0}(x^{-1}y) .\end{multline*}
By \eqref{e:partialtrace}, what we have to show is that
  $$\int_{\partial \mathfrak{X}} \la P_{\lambda, \omega}, (\iota \widetilde U_{\lambda})^* \bbbone_{\cD'  }\widehat{K}_{B(\mathfrak{X})} \widetilde U_{\lambda}P_{\lambda, \omega}\ra d\nu_{\lambda}(\omega)=0.$$
 
 In fact, we even have the stronger property that
 $$\la P_{\lambda, \omega}, (\iota \widetilde U_{\lambda})^*  \bbbone_{\cD' } \widehat{K}_{B(\mathfrak{X})} \widetilde U_{\lambda}P_{\lambda, \omega}\ra =0$$
 for all $\omega$. Indeed,  $\widehat{K}_{B(\mathfrak{X})} \widetilde U_{\lambda}P_{\lambda, \omega}$ is supported on the set of edges going towards $\omega$, whereas
 $\iota \widetilde U_{\lambda} P_{\lambda, \omega}$ is supported on the set of edges going away from $\omega$.
 \end{proof}

\begin{lemma} \label{l:m=0}

 If $a\in \cH_0^0$ we have for all $T$
$$
Var(a)\leq   4\frac{C(0,\beta)^2}{T^2}\norm{a}^2_{\cH_0^0}
+\norm{K}^2_{sup}\, (o_{T}(1)_{G\To\mathfrak X } + O(T^{-2})) .
 $$

 \end{lemma}
 This lemma, in particular, proves Theorem \ref{t:mainanis}.
\begin{proof}(of Lemma \ref{l:m=0})

In the definition of $Var^I_{nb}$ we used the functions
$$g_j(e)=\frac{p(e)}{\zeta_{\lambda_j+i0}(e)}\left(\psi_j(t(e))-\zeta_{\lambda_j+i0}(e)\psi_j(o(e))\right).$$
Needless to say, we would have obtained the same bound \eqref{e:nbtanis} for $\widetilde{Var}^I_{nb}$, defined using
$$\widetilde{g}_j(e)=\frac{p(e)}{\zeta_{\lambda_j-i0}(e)}\left(\psi_j(t(e))-\zeta_{\lambda_j-i0}(e)\psi_j(o(e))\right)$$
(if the $\psi_j$ are real-valued then $\widetilde{g}_j$ is the complex conjugate of $g_j$ but we do not necessarily want to impose that).

If $a\in \cH_0$, define $K_\lambda\in \cH_1$ by
$$K_\lambda(x, y)=\frac{a(x) |\zeta_\lambda (x, y)|^2}{p(x, y)}.$$
A direct calculation yields the formula
$$\la \iota g_j, \hat{K}_{\lambda_j, B} g_j\ra -\la \iota \widetilde g_j, \hat{K}_{\lambda_j, B} \widetilde g_j\ra= 2
\sum_{x\in V}|\psi_j(x)|^2 [(\cA_{\lambda_j}- W_{\lambda_j})a](x)$$
where $\cA_{\lambda_j}$ is the operator 
$$\cA_{\lambda_j}f (x)=\sum_{y\sim x}\Im m \zeta_{\lambda_j}(x, y) f(y)$$
and $W_{\lambda_j}$ is the constant
$W_{\lambda_j}= \sum_{y\sim x}\Im m \zeta_{\lambda_j}(x, y)$ (note that this is independent of $x$).
Thus \eqref{e:estvarnbanisglob} implies that
\begin{equation}
Var((\cA_\lambda-W_\lambda)a)\leq \frac{\tilde C(1, {\varepsilon}, \beta)}{ (n+1)}\norm{a}_{\cH}^2 +\norm{a}_{sup}^2 \, (o_{n, \delta, \varepsilon}(1)_{G\To\mathfrak X } + n c(1)^2O(\delta)+ c(1)^2O(\varepsilon)) .
\end{equation}
For every $\lambda$, the image of $\cA_\lambda-W_\lambda$ is $\cH^0_0$. Besides, arguing like in Lemma \ref{l:cheeger}, we see that under the (EXP) condition, all the $\cA_\lambda-W_\lambda$ have a uniform spectral gap. 
 
 Arguing like in Lemma \ref{l:useful}, we obtain that for all $n$ and $T$,
 for $a\in \cH_0^0$, 
 \begin{multline*}
Var(a)\leq  2\frac{C(0,\beta)^2 \tilde C(1, {\varepsilon}, \beta) }{n}\norm{a}^2_{\cH^0_0}\\ +2\frac{C(0,\beta)^2}{T^2}\norm{a}^2_{\cH_0^0}
+2T^2\norm{a}^2_{sup}\, (o_{n, \delta, \varepsilon}(1)_{G\To\mathfrak X } + n c(1)^2O(\delta)+ c(1)^2O(\varepsilon)) .
 \end{multline*}
 
To prove the lemma we take $\varepsilon =T^{-4}$, $\delta=n^{-1 }T^{-4}$, and $n=n(T)$ large enough so that 
$\frac{  \tilde C(1, {\varepsilon}, \beta) }{n}\leq \frac{ 1}{T^2}$. \end{proof}

\begin{proof} (of Proposition \ref{p:recursion})

We use \eqref{e:recursion} and \eqref{e:recurKK} to write
\begin{multline*}Var(K -\la K\ra_{\lambda, p} \bbbone_0)\\=Var\left(K+\cL_{m-1} K  + \cL_{m-2} K -(\cL_{m-1} K  + \cL_{m-2} K) +\la \cL_{m-1} K  + \cL_{m-2} K\ra_{\lambda, p} \bbbone_0\right)\\
\leq 2Var \left(K+\cL_{m-1} K  + \cL_{m-2} K\right) + 2Var\left(\cL_{m-1} K  + \cL_{m-2} K-\la \cL_{m-1} K  + \cL_{m-2} K\ra_p \bbbone_0\right)\\
\leq \frac{\tilde C(m, {\varepsilon}, \beta)}{ (n+1)}\norm{K}_{\cH}^2 +\norm{K}_{sup}^2 \, (o_{n, m, \delta, \varepsilon}(1)_{G\To\mathfrak X } + n c(m)^2O(\delta)+ c(m)^2O(\varepsilon))\\
+ 2Var\left(\cL_{m-1} K  + \cL_{m-2} K-\la \cL_{m-1} K  + \cL_{m-2} K\ra_p \bbbone_0\right).\end{multline*}
This allows to prove Proposition \ref{p:recursion} by induction on $m$.
\end{proof} 
 
\bibliographystyle{plain}
\bibliography{biblio-lemasson} 

\end{document}

%% file: def.tex
\newcommand{\nwc}{\newcommand}
\nwc{\nwt}{\newtheorem}
\nwt{coro}{Corollary}
\nwt{ex}{Example}
\nwt{prop}{Proposition}
\nwt{defin}{Definition}


\nwc{\mf}{\mathbf} 
\nwc{\blds}{\boldsymbol} 
\nwc{\ml}{\mathcal} 


\nwc{\lam}{\lambda}
\nwc{\del}{\delta}
\nwc{\Del}{\Delta}
\nwc{\Lam}{\Lambda}
\nwc{\elll}{\ell}

\nwc{\IA}{\mathbb{A}} 
\nwc{\IB}{\mathbb{B}} 
\nwc{\IC}{\mathbb{C}} 
\nwc{\ID}{\mathbb{D}} 
\nwc{\IE}{\mathbb{E}} 
\nwc{\IF}{\mathbb{F}} 
\nwc{\IG}{\mathbb{G}} 
\nwc{\IH}{\mathbb{H}} 
\nwc{\IN}{\mathbb{N}} 
\nwc{\IP}{\mathbb{P}} 
\nwc{\IQ}{\mathbb{Q}} 
\nwc{\IR}{\mathbb{R}} 
\nwc{\IS}{\mathbb{S}} 
\nwc{\IT}{\mathbb{T}} 
\nwc{\IZ}{\mathbb{Z}} 
\def\bbbone{{\mathchoice {1\mskip-4mu {\rm{l}}} {1\mskip-4mu {\rm{l}}}
{ 1\mskip-4.5mu {\rm{l}}} { 1\mskip-5mu {\rm{l}}}}}
\def\bbleft{{\mathchoice {[\mskip-3mu {[}} {[\mskip-3mu {[}}{[\mskip-4mu {[}}{[\mskip-5mu {[}}}}
\def\bbright{{\mathchoice {]\mskip-3mu {]}} {]\mskip-3mu {]}}{]\mskip-4mu {]}}{]\mskip-5mu {]}}}}
\nwc{\setK}{\bbleft 1,K \bbright}
\nwc{\setN}{\bbleft 1,\cN \bbright}
 \newcommand{\Lim}{\mathop{\longrightarrow}\limits}


\nwc{\va}{{\bf a}}
\nwc{\vb}{{\bf b}}
\nwc{\vc}{{\bf c}}
\nwc{\vd}{{\bf d}}
\nwc{\ve}{{\bf e}}
\nwc{\vf}{{\bf f}}
\nwc{\vg}{{\bf g}}
\nwc{\vh}{{\bf h}}
\nwc{\vi}{{\bf i}}
\nwc{\vI}{{\bf I}}
\nwc{\vj}{{\bf j}}
\nwc{\vk}{{\bf k}}
\nwc{\vl}{{\bf l}}
\nwc{\vm}{{\bf m}}
\nwc{\vM}{{\bf M}}
\nwc{\vn}{{\bf n}}
\nwc{\vo}{{\it o}}
\nwc{\vp}{{\bf p}}
\nwc{\vq}{{\bf q}}
\nwc{\vr}{{\bf r}}
\nwc{\vs}{{\bf s}}
\nwc{\vt}{{\bf t}}
\nwc{\vu}{{\bf u}}
\nwc{\vv}{{\bf v}}
\nwc{\vw}{{\bf w}}
\nwc{\vx}{{\bf x}}
\nwc{\vy}{{\bf y}}
\nwc{\vz}{{\bf z}}
\nwc{\bal}{\blds{\alpha}}
\nwc{\bep}{\blds{\epsilon}}
\nwc{\barbep}{\overline{\blds{\epsilon}}}
\nwc{\bnu}{\blds{\nu}}
\nwc{\bmu}{\blds{\mu}}
\nwc{\bet}{\blds{\eta}}



\nwc{\bk}{\blds{k}}
\nwc{\bm}{\blds{m}}
\nwc{\bM}{\blds{M}}
\nwc{\bp}{\blds{p}}
\nwc{\bq}{\blds{q}}
\nwc{\bn}{\blds{n}}
\nwc{\bv}{\blds{v}}
\nwc{\bw}{\blds{w}}
\nwc{\bx}{\blds{x}}
\nwc{\bxi}{\blds{\xi}}
\nwc{\by}{\blds{y}}
\nwc{\bz}{\blds{z}}


\nwc{\cA}{\ml{A}}
\nwc{\cB}{\ml{B}}
\nwc{\cC}{\ml{C}}
\nwc{\cD}{\ml{D}}
\nwc{\cE}{\ml{E}}
\nwc{\cF}{\ml{F}}
\nwc{\cG}{\ml{G}}
\nwc{\cH}{\ml{H}}
\nwc{\cI}{\ml{I}}
\nwc{\cJ}{\ml{J}}
\nwc{\cK}{\ml{K}}
\nwc{\cL}{\ml{L}}
\nwc{\cM}{\ml{M}}
\nwc{\cN}{\ml{N}}
\nwc{\cO}{\ml{O}}
\nwc{\cP}{\ml{P}}
\nwc{\cQ}{\ml{Q}}
\nwc{\cR}{\ml{R}}
\nwc{\cS}{\ml{S}}
\nwc{\cT}{\ml{T}}
\nwc{\cU}{\ml{U}}
\nwc{\cV}{\ml{V}}
\nwc{\cW}{\ml{W}}
\nwc{\cX}{\ml{X}}
\nwc{\cY}{\ml{Y}}
\nwc{\cZ}{\ml{Z}}

\nwc{\fA}{\mathfrak{a}}
\nwc{\fB}{\mathfrak{b}}
\nwc{\fC}{\mathfrak{c}}
\nwc{\fD}{\mathfrak{d}}
\nwc{\fE}{\mathfrak{e}}
\nwc{\fF}{\mathfrak{f}}
\nwc{\fG}{\mathfrak{g}}
\nwc{\fH}{\mathfrak{h}}
\nwc{\fI}{\mathfrak{i}}
\nwc{\fJ}{\mathfrak{j}}
\nwc{\fK}{\mathfrak{k}}
\nwc{\fL}{\mathfrak{l}}
\nwc{\fM}{\mathfrak{m}}
\nwc{\fN}{\mathfrak{n}}
\nwc{\fO}{\mathfrak{o}}
\nwc{\fP}{\mathfrak{p}}
\nwc{\fQ}{\mathfrak{q}}
\nwc{\fR}{\mathfrak{r}}
\nwc{\fS}{\mathfrak{s}}
\nwc{\fT}{\mathfrak{t}}
\nwc{\fU}{\mathfrak{u}}
\nwc{\fV}{\mathfrak{v}}
\nwc{\fW}{\mathfrak{w}}
\nwc{\fX}{\mathfrak{x}}
\nwc{\fY}{\mathfrak{y}}
\nwc{\fZ}{\mathfrak{z}}


\nwc{\tA}{\widetilde{A}}
\nwc{\tB}{\widetilde{B}}
\nwc{\tE}{E^{\vareps}}
\nwc{\tk}{\tilde k}
\nwc{\tN}{\tilde N}
\nwc{\tP}{\widetilde{P}}
\nwc{\tQ}{\widetilde{Q}}
\nwc{\tR}{\widetilde{R}}
\nwc{\tV}{\widetilde{V}}
\nwc{\tW}{\widetilde{W}}
\nwc{\ty}{\tilde y}
\nwc{\teta}{\tilde \eta}
\nwc{\tdelta}{\tilde \delta}
\nwc{\tlambda}{\tilde \lambda}
\nwc{\ttheta}{\tilde \theta}
\nwc{\tvartheta}{\tilde \vartheta}
\nwc{\tPhi}{\widetilde \Phi}
\nwc{\tpsi}{\tilde \psi}
\nwc{\tmu}{\tilde \mu}

\nwc{\To}{\longrightarrow} 

\nwc{\ad}{\rm ad}
\nwc{\eps}{\epsilon}
\nwc{\ep}{\epsilon}
\nwc{\vareps}{\varepsilon}

\def\bom{\mathbf{\omega}}
\def\om{{\omega}}
\def\ep{\epsilon}
\def\tr{{\rm tr}}
\def\diag{{\rm diag}}
\def\Tr{{\rm Tr}}
\def\i{{\rm i}}
\def\mi{{\rm i}}
\def\e{{\rm e}}
\def\sq2{\sqrt{2}}
\def\sqn{\sqrt{N}}
\def\vol{\mathrm{vol}}
\def\defi{\stackrel{\rm def}{=}}
\def\t2{{\mathbb T}^2}
\def\s2{{\mathbb S}^2}
\def\hn{\mathcal{H}_{N}}
\def\shbar{\sqrt{\hbar}}
\def\A{\mathcal{A}}
\def\N{\mathbb{N}}
\def\T{\mathbb{T}}
\def\R{\mathbb{R}}
\def\RR{\mathbb{R}}
\def\Z{\mathbb{Z}}
\def\C{\mathbb{C}}
\def\O{\mathcal{O}}
\def\Sp{\mathcal{S}_+}
\def\Lap{\triangle}
\nwc{\lap}{\bigtriangleup}
\nwc{\rest}{\restriction}
\nwc{\Diff}{\operatorname{Diff}}
\nwc{\diam}{\operatorname{diam}}
\nwc{\Res}{\operatorname{Res}}
\nwc{\Spec}{\operatorname{Spec}}
\nwc{\Vol}{\operatorname{Vol}}
\nwc{\Op}{\operatorname{Op}}
\nwc{\supp}{\operatorname{supp}}
\nwc{\Span}{\operatorname{span}}

\nwc{\dia}{\varepsilon}
\nwc{\cut}{f}
\nwc{\qm}{u_\hbar}

\def\hto0{\xrightarrow{\hbar\to 0}}
\def\htoo{\stackrel{h\to 0}{\longrightarrow}}
\def\rto0{\xrightarrow{r\to 0}}
\def\rtoo{\stackrel{r\to 0}{\longrightarrow}}
\def\ntoinf{\xrightarrow{n\to +\infty}}

\providecommand{\abs}[1]{\lvert#1\rvert}
\providecommand{\norm}[1]{\lVert#1\rVert}
\providecommand{\set}[1]{\left\{#1\right\}}

\nwc{\la}{\langle}
\nwc{\ra}{\rangle}
\nwc{\lp}{\left(}
\nwc{\rp}{\right)}

\nwc{\bequ}{\begin{equation}}
\nwc{\be}{\begin{equation}}
\nwc{\ben}{\begin{equation*}}
\nwc{\bea}{\begin{eqnarray}}
\nwc{\bean}{\begin{eqnarray*}}
\nwc{\bit}{\begin{itemize}}
\nwc{\bver}{\begin{verbatim}}

%\nwc{\eal}{\end{align}}
\nwc{\eequ}{\end{equation}}
\nwc{\ee}{\end{equation}}
\nwc{\een}{\end{equation*}}
\nwc{\eea}{\end{eqnarray}}
\nwc{\eean}{\end{eqnarray*}}
\nwc{\eit}{\end{itemize}}
\nwc{\ever}{\end{verbatim}}

\newcommand{\defeq}{\stackrel{\rm{def}}{=}}